\pdfoutput=1
\documentclass{article}

\usepackage{PRIMEarxiv}

\usepackage{subcaption}
\usepackage{natbib}
\usepackage[utf8]{inputenc} 
\usepackage[T1]{fontenc}    
\usepackage{hyperref}       
\usepackage{url}            
\usepackage{booktabs}       
\usepackage{amsfonts}       
\usepackage{nicefrac}       
\usepackage{microtype}      
\usepackage{lipsum}
\usepackage{fancyhdr}       
\usepackage{graphicx}       

\usepackage{hyperref}
\usepackage{longtable}
\usepackage{threeparttable}
\usepackage{multirow}
\usepackage{float}
\usepackage{wrapfig}
\usepackage{url}
\usepackage{booktabs,tabularx}
\usepackage{graphicx}
\usepackage{xspace}
\usepackage{enumitem}
\usepackage{flafter}

\usepackage{amsmath}
\usepackage{amsthm}

\newtheorem{proposition}{Proposition}
\newtheorem{corollary}{Corollary}
\newtheorem{lemma}{Lemma}[section]
\newtheorem{remark}{Remark}

\usepackage{algorithm}
\usepackage{algpseudocode}
\usepackage{comment}
\usepackage{subcaption}
\usepackage{nicefrac}
\usepackage{tikz}
\usetikzlibrary{arrows.meta,positioning}
\usepackage{bbm}

\begin{document}

\newcommand{\modelname}{ALF}
\newcommand{\modelfullnameunderlined}{Spectrally \underline{A}nchored \underline{L}atent \underline{F}low Matching}
\newcommand{\modelfullname}{Anchored Latent Flow}
\title{\modelname: \modelfullnameunderlined{} for Wireless Signal Generation}
\newcommand{\datasetname}{WaveScene\xspace}

\pagestyle{fancy}
\thispagestyle{empty}
\rhead{ \textit{}} 

\fancyhead[LO]{\modelfullnameunderlined}

\title{\modelname: \modelfullnameunderlined{} for Wireless Signal Generation}

\author{
\begin{tabular}{ccc}
Yiyang Li\thanks{Equal contribution (co-first authors)} &
Sai Shankar Narasimhan\footnotemark[1] &
Priam Alataris\thanks{Equal contribution (co-second authors)} \\[1em]
Avi Bagchi\footnotemark[2] &
Kaushik Chowdhury &
Sanjay Shakkottai
\end{tabular}
}

\maketitle
{\renewcommand{\thefootnote}{}\footnotetext{The authors are with The University of Texas at Austin. Email: 
\tt{\{liyiyang, nsaishankar, alataris, aviba, kaushik, sanjay.shakkottai\}@utexas.edu}}}

\begin{abstract}
Wireless time-series generation is useful for emerging applications that will drive the adoption and integration of machine learning tasks within next-generation networks, such as waveform classification in shared spectrum bands and interpretation of the physical world through integrated sensing and communications. Unlike a generic time series, wireless signals have time-frequency duality, where a frequency domain bandwidth constraint implicitly imposes latent structural constraints on the evolution of the time-domain sequence. In this paper, we develop \modelfullname{} (\modelname), a latent space flow model for label-conditioned signal generation that exploits this duality. Our flow model ``reasons'' through a spectrogram, which is a 2-D representation of the signal constructed by applying a windowed Fast Fourier Transform to successive, overlapping signal segments and mapping their frequency-domain power components across time. Specifically, we shape the latent space of the flow model through a training-time auxiliary objective that is supervised by the (latent) spectrogram (i.e. anchoring through the spectrogram). By guiding the flow field to reason through a spectrogram, the generated sequence maintains long-range coherence across time and greatly improves the signal quality at inference. Empirically, \modelname{} exhibits state-of-the-art generation quality and reconstruction error across various signal lengths, waveform classes, signal-to-noise-ratio (SNR) and channel conditions, reaching a train on synthetic, test on real (TSTR) accuracy of 84.5\%, a $1.9\times$ average memory reduction, and $100\times$ average compute speedup over existing models. Our design further enables downstream tasks including style transfer and data augmentation, and has broad utility for the wireless community.

\end{abstract}


\section{Introduction}
Generating a received radio-frequency (RF) transmission under certain propagation environments and signal-to-noise ratios without conducting an over-the-air experiment is challenging. Such generative capabilities enable applications ranging from targeted synthetic data augmentation to creating digital twins of specific environments. Recent progress in diffusion and flow modeling has helped solve similar problems in the image and video domains \citep{ho2020ddpm,ho2022video,lipman2023flow,liu2023rectified}. RF signals, however, possess highly engineered structures that demand precise modeling for trustworthy generation. Longer samples, spectral and phase coherence, and correlations across multiple receive antennas all introduce additional challenges.
\begin{figure}[t]
    \centering
    \includegraphics[width=\linewidth]{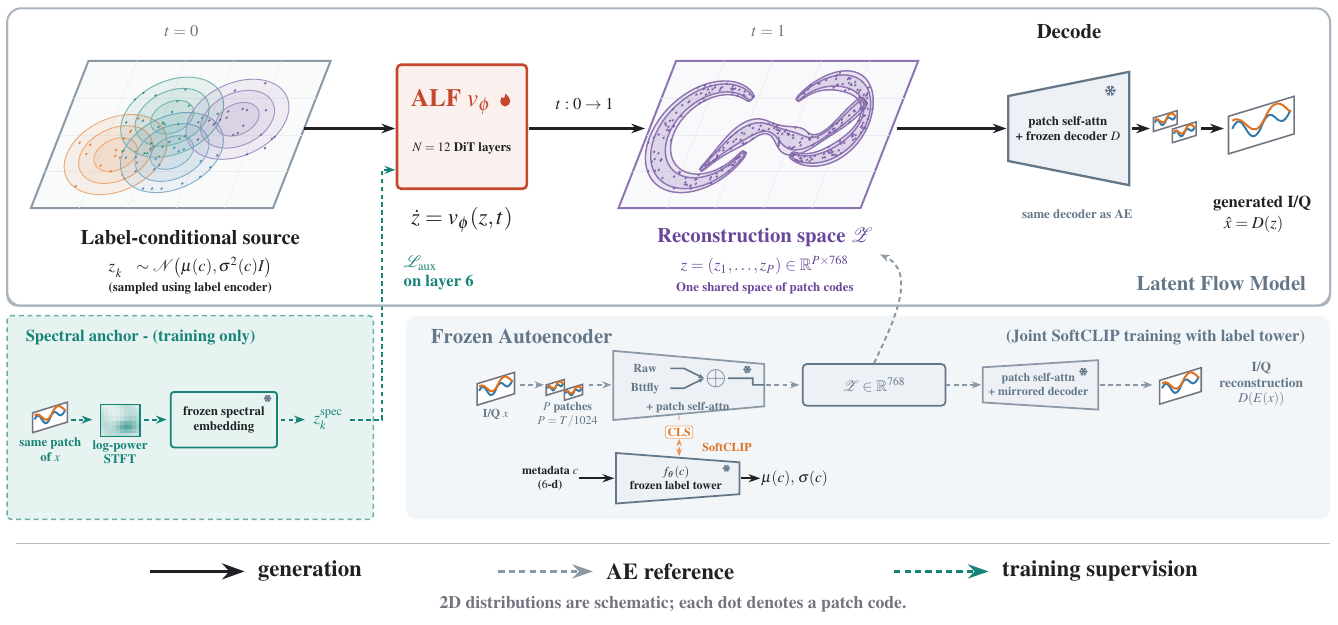}
    \caption{\modelfullname\ (\modelname) pipeline. A contrastively trained autoencoder learns to reconstruct signals while maintaining alignment with feature labels ($\S$~\ref{subsec:stage1}). \modelname{} then transports a multimodal feature space to the reconstruction space by reasoning through the spectrogram ($\S$~\ref{subsec:anchoredflow}).}
     \label{fig:intro-figure}
\end{figure}
RF waveforms are naturally constrained in both the time and frequency domain. For example, regulatory and task-specific bandwidth limitations restrict RF signals to a prescribed range of frequencies \citep{fcc_frequency_allocations_2026}. Due to time-frequency duality, these spectral constraints induce strong structure in the time domain, restricting plausible waveforms to a small subset of the waveform subspace. Additionally, ignoring the spectral structure can produce time series generations that violate spectral constraints or corrupt waveform morphology. Prior work has exploited spectral information through Fourier-based architectures and training objectives \citep{wu2023iclr-timesnet,yuan2024diffusion} and joint time-frequency diffusion for RF signals \citep{chi2024rfdiffusion}. However, prior work does not leverage the frequency domain as an intermediate reasoning space. We hypothesize that supervising the frequency-domain latent can improve coherence while enabling substantially longer RF generation.

Resolving how exactly to incorporate spectral information into the generative learning procedure is therefore a critical design choice. We take inspiration from \textit{anchoring} introduced by \citet{rout2025adlm} for diffusion language modeling. The central idea is to reveal the most informative tokens (anchors) first, as they guide the remaining generation. Accordingly, the space of anchors resides in the same space as the generated output. In the RF domain, however, individual signal samples carry far less global context to exploit than an important text token would. To this end, our key observation is that ``anchors'' need not occupy the same subspace as the final generated time-domain output. Instead, we anchor from the frequency-domain representation to shape the latent space of a flow. Specifically, we posit that learning to predict a compact spectrogram can effectively help guide the generation of time-domain data. \textbf{Our core contributions are:}

\begin{enumerate}[leftmargin=18pt, itemsep=1pt]
    \item We propose \textbf{\modelfullname{}} (\modelname, see Figure~\ref{fig:intro-figure}), a label-conditioned latent flow-matching model for single-input-multiple-output (SIMO) in-phase and quadrature (IQ) generation. \modelname{} transports an aligned, condition-aware signal latent space to a reconstruction-oriented latent space using a learned velocity field. For long-range coherence, we anchor the flow in the frequency domain by supervising the intermediate representation of the flow model to predict spectrogram latents mixed with global information (Appendix~\ref{app:theory} for a formal treatment).
    \item We introduce \textbf{WaveScene}, a large-scale SIMO dataset with 7.82 million IQ recording-label pairs across various waveform classes, propagation environments, and sequence lengths. Experiments on unseen channel realizations demonstrate competitive representation learning, long-sequence generation, and controlled signal editing.
    \item We show that spectrogram anchoring reshapes the flow's intermediate representations and improves generation over matched unanchored ablations. Overall, \modelname{} achieves a $38.02\%$ relative improvement in Train on Synthetic, Test on Real (TSTR) waveform classification accuracy compared to the closest baseline. We further demonstrate computational speedups by $100\times$, memory reductions by $1.9\times$, native style transfer, and autoregressive capabilities.
\end{enumerate}

\section{Related Work}
\label{sec:rel_work}

We describe the related work below, deferring a more detailed exposition to Appendix~\ref{ap:extended-work}.

\textbf{Diffusion Models} have emerged as a leading approach to generative modeling for images~\citep{sohl2015deep, ho2020ddpm,song2021ddim}, audio~\citep{kong2021diffwave}, and video~\citep{ho2022video}. More recent work has applied diffusion models to time-series tasks, including forecasting \citep{rasul2021timegrad,yan2021scoregrad,bilovs2023modeling,kollovieh2023predict,li2022generative} and imputation \citep{tashiro2021csdi,alcaraz2022diffusion,yuan2024diffusion}, using both conditional formulations and guidance-based approaches. Conditional generation of time-series has been developed in \citet{alcaraz2023diffusion}, \citet{narasimhan2024time}, and \citet{narasimhan2025constrained}. However, our experiments (see Table \ref{tab:tstr-headline}) show that these state-of-the-art (SOTA) approaches struggle to generate long-range coherent IQ waveforms, which are typically required in the RF domain. \textbf{Our key hypothesis} is that these approaches \textit{lack supervision of intermediate representations pertaining to the global frequency domain}, and therefore generate sub-optimal results. Alongside conditioning and guidance, recent work in the language domain introduces \emph{anchoring}, in which the model predicts important intermediate information that greatly reduces the entropy of the remaining data to be generated~\citep{rout2025adlm}. Similarly, representations have been shaped in the image domain using a clean image-target in \citep{yu2025repa}. In light of these, we introduce cross-domain spectral anchoring, in which intermediate model representations are supervised using a separately learned frequency-domain representation so that RF-domain-specific characteristics are used as first-class inductive biases.

\textbf{RF-Specific Generation Methods} must learn strict amplitude, frequency, and phase information that determine a wireless signal's spectral and temporal identity, as well as the characteristics of the propagation channel through which it travels. This specific structure motivates physics-informed approaches for constructing and evaluating generative models for the wireless domain \citep{phys-informed-generation, baur2025evaluation}. A branch of prior work focuses on modeling the wireless channel, either by estimating it from received pilot measurements and known transmitted pilots using Wasserstein GANs \citep{balevi2020highdimensionalchannelestimation}, score-based models \citep{arvinte2022score,arvinte2023mimo}, or diffusion priors \citep{zhou2026generativediffusionmodelshigh}, or by directly synthesizing channel realizations using location-conditioned diffusion models \citep{lee2025generatinghighdimensionaluserspecific} or a VAE coupled to a parametric physics-based geometric channel model \citep{wagle2025physicsinformedgenerativeapproacheswireless}. Generating I/Q sequences, however, is comparatively understudied and requires extra modeling for the joint variability of the transmitted waveform, propagation channel, and receiver noise over long sequences. We adapt RF-Diffusion, Time Weaver, and WaveStitch as baselines for this task \citep{chi2024rfdiffusion,narasimhan2024time,shankar2025wavestitch}.

\vspace{-8pt}
\section{Background}
\label{sec:background}

Consider a complex-valued baseband signal on the receiver side,
$x(\tau)=I(\tau)+jQ(\tau)$, with discrete samples $x[n]=I[n]+jQ[n]$ for
$n=0,\ldots,N_{\mathrm{samples}}-1$. For $N_{\mathrm{ant}}$ receive antennas,
stacking their synchronized sample sequences gives
$\mathbf{X}_{\mathrm{IQ}}\in
\mathbb{C}^{N_{\mathrm{ant}}\times N_{\mathrm{samples}}}$, which we refer to as an IQ waveform. Each $\mathbf{X}_{\mathrm{IQ}}$ is paired with heterogeneous conditioning information
$\mathbf{Y}$, which may contain continuous or categorical quantities. After mapping each component to a fixed-dimensional real
representation, such as a one-hot encoding for categorical variables, we write
$\mathbf{Y}\in\mathbb{R}^{N_{\mathrm{labels}}}$. A dataset is then
$\mathcal{D}=\{(\mathbf{X}_{\mathrm{IQ}}^i,\mathbf{Y}^i)\}_{i=1}^{N_{\mathrm{pairs}}}$,
where $N_{\mathrm{pairs}}$ is its size. For learning, we treat these pairs as independent and identically distributed (i.i.d.) samples from the empirical joint distribution $p_{\mathrm{data}}(\mathbf{X}_{\mathrm{IQ}},\mathbf{Y})$.

\textbf{Our primary objective is to} learn a conditional generation model such that the IQ waveforms generated by the model with input label $\mathbf{Y}$ distributionally match the conditional distribution $p(\mathbf{X}_{\mathrm{IQ}} | \mathbf{Y})$ for all valid $\mathbf{Y} \in \mathbb{R}^{N_{\mathrm{labels}}}$ (Figure~\ref{fig:intro-figure}). To this end, we build on top of \textbf{Rectified Flow} \citep{liu2023rectified}, which learns a velocity field that transports samples from a tractable source distribution to the target data distribution. Let $\mathbf{X}_0\sim p_{\mathrm{src}}$ denote a sample from a source distribution, such as $\mathcal{N}(\mathbf{0},\mathbf{I})$, and let $\mathbf{X}_1\sim p_{\mathrm{data}}$ denote a sample from the target data
distribution. Rectified flow defines the straight-line interpolation, $\mathbf{X}_t = (1-t)\mathbf{X}_0 + t\mathbf{X}_1, t\in[0,1]$, whose velocity, conditioned on the source-target pair $(\mathbf{X}_0, \mathbf{X}_1)$, is constant along the path $u_t = \frac{d\mathbf{X}_t}{dt} = \mathbf{X}_1-\mathbf{X}_0.$ We train a conditional
time-dependent velocity field $v_\theta(\mathbf{X}_t,t,\mathbf{Y})$ using
\vspace{-5pt}
\begin{equation}
\mathcal{L}_{\mathrm{RF}}
=
\mathbb{E}_{\substack{
(\mathbf{X}_1,\mathbf{Y})\sim p_{\mathrm{data}},\,
\mathbf{X}_0\sim p_{\mathrm{src}},\\
t\sim\mathcal{U}[0,1]}}
\left[
\left\|
v_\theta(\mathbf{X}_t,t,\mathbf{Y})
-
(\mathbf{X}_1-\mathbf{X}_0)
\right\|_2^2
\right].
\end{equation}
\vspace{-5pt}

Here, we uniformly sample the time $t$, \emph{i.e.,} $t \sim \mathcal{U}[0,1]$. Although each sampled pair defines a constant straight-line velocity, the regression objective learns the conditional mean transport induced by all such paths. At inference, we draw $\mathbf{X}_0\sim p_{\mathrm{src}}$ and integrate $d\mathbf{X}_t/dt=v_\theta(\mathbf{X}_t,t,\mathbf{Y})$ from $t=0$ to $t=1$ to
obtain a sample from the learned conditional distribution.

\section{\datasetname Dataset}
\label{sec:dataset}

\textbf{\datasetname} is a large-scale synthetic RF corpus of receiver-side baseband
IQ waveforms generated through geometry-based stochastic channels and paired
with respective labels. It covers 32 waveform classes and 1,066 parameterized variants spanning generic modulations, standardized communication protocols, and radar emissions. Each recording is a 4-antenna SIMO post-channel waveform captured by a horizontal uniform linear array with half-wavelength spacing ($N_{\mathrm{ant}}=4$), with $N_{\mathrm{samples}}\in\{2048,4096,8192,16384,32768\}$. The corpus covers four propagation environments: urban microcell (UMi), urban macrocell (UMa), rural macrocell (RMa), and indoor hotspot (InH). Channel generation follows the geometry-based stochastic models specified in 3GPP TR~38.901 \citep{3gpp_tr38901}. The propagation states are line-of-sight (LOS), non-line-of-sight (NLOS), and
outdoor-to-indoor (O2I), with frequencies determined by the scenario-specific
channel models.
\vspace{-10pt}
\newcolumntype{Y}{>{\raggedright\arraybackslash}X}
\begin{table}[H]
\centering
\scriptsize
\setlength{\tabcolsep}{2pt}
\caption{Comparison with representative public IQ datasets.}
\label{tab:dataset-comparison}
\begin{tabularx}{\textwidth}{@{}lXcccccc@{}}
\toprule
Dataset &
Coverage &
Size &
Length &
Sampling rate &
Multi-Rx &
Varying $f_c$ &
Channel geometry\\
\midrule

RadioML 2018 \citep{radioml2018} &
24 classes$^\dagger$ &
2.56M &
1,024 &
N/A &
$\times$ &
$\times$ &
$\times$ \\

Sig53 \citep{sig53} &
53 classes$^\dagger$ &
6.51M &
4,096 &
N/A &
$\times$ &
$\times$ &
$\times$ \\

IQFM \citep{mashaal2026iqfm} &
7 classes$^\dagger$ &
1.12M &
256 &
1, 10 MS/s &
$\checkmark$ &
$\times$ &
$\times$ \\

EM-134K \citep{merlin} &
10 subsets$^\ddagger$  &
134K &
1,024 &
20 MS/s &
$\times$ &
$\times$ &
$\times$ \\

\textbf{WaveScene} &
\textbf{32 classes} &
\textbf{7.82M} &
\textbf{2,048--32,768} &
\textbf{96 kS/s--40 MS/s} &
$\boldsymbol{\checkmark}$ &
$\boldsymbol{\checkmark}$ &
$\boldsymbol{\checkmark}$ \\

\bottomrule
\end{tabularx}

\vspace{2pt}
\parbox{\textwidth}{\scriptsize
$^\dagger$ Modulation or simple waveform classes only; no protocol or radar waveform families.
$^\ddagger$ EM-134K is an instruction-tuning corpus spanning ten heterogeneous signal/task subsets, rather than mutually exclusive waveform classes.
}
\end{table}
\vspace{-10pt}

The corpus contains approximately 7.82M recording-label pairs, including a 7.52M shared-channel set and a 300K channel-held-out set. Of this shared-channel set, 5.6M pairs form the development set, divided deterministically into 95\% training and 5\% validation data; the remaining 1.9M disjoint recordings form the seen-channel test set. We additionally construct the 300K channel-held-out test set whose channel realizations are unseen. Training uses 6 complementary labels. Waveform class (32), scene type (4), and propagation state (3) are categorical, while spectral occupancy, mean
per-antenna pre-combining SNR, and RMS delay spread are continuous ($N_{\mathrm{labels}} = 42$). Appendix~\ref{app:dataset} details the generation models and methods used.

\section{\modelfullname}
\label{sec:methods}

In this section, we introduce \textbf{\modelfullname} (\modelname), a latent-space flow-matching model for efficient and coherent long-range receiver-side baseband IQ waveform generation conditioned on environment/propagation labels. As discussed in Section~\ref{sec:rel_work}, we design \modelname{} by treating the time-frequency inter-relationship in the RF-domain as a first-class inductive bias; see Appendix~\ref{app:theory} for a formal justification. In particular, extending existing temporal generative models (Time Weaver, RF-Diffusion and WaveStitch) for long RF IQ sequences exposes two key challenges:

\begin{itemize}[leftmargin=18pt, itemsep=1pt]

    \item \textbf{Poor Inference Efficiency.} Existing generators typically operate in the IQ space alone. For long-range IQ sequence generation, this results in high memory and inference costs. Particularly, this becomes evident for IQ waveforms of length $32{,}768$ samples (Table~\ref{tab:inf-speed}). Moreover, 32k samples are still short; for example, a sequence of this length is roughly 10\% of a single 20MHz LTE frame with a 30MHz sample rate.

    \item \textbf{Weak Long-Range Spectral Coherence.}
    The received IQ waveform evolution over time is globally constrained by spectral occupancy and the time-frequency structure. However, existing denoising or generative
    models do not explicitly constrain the model's representations to preserve this global structure. Consequently, the errors in IQ space can accumulate over time and produce globally inconsistent spectral behavior.
\end{itemize}

\modelname{} addresses these two challenges through a two-stage training setup. 

\begin{figure}[t]
\centering
\includegraphics[width=\linewidth]{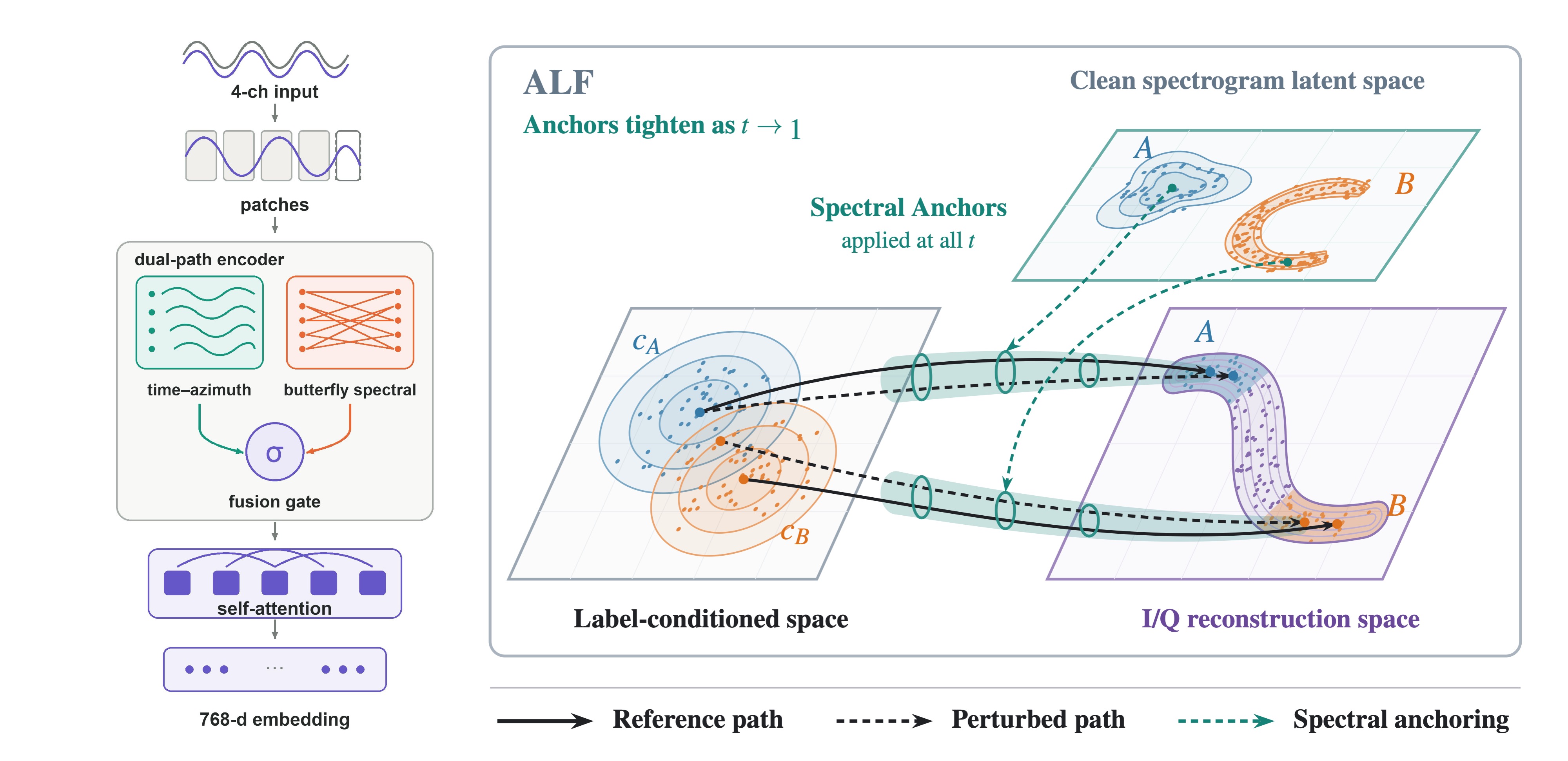}
    \caption{(Left) Fused dual-branch architecture with \textit{ButterflyNet} \citep{li2020butterflynetoptimalfunctionrepresentation, xu2020butterflynet2}. (Right) Spectral anchoring guides learned flow with frequency-domain latents.}
        \label{fig:alf}
    \end{figure}

\subsection{Stage 1: Label-Aligned Reconstruction-Aware Encoding}
\label{subsec:stage1}

Rather than performing generation directly in IQ space, we operate in a compact \textit{latent space representation}. More specifically, we learn a concise sequence of patch-level IQ latents that can be decoded back to the original waveform. For these patch-level latents, we use cross-patch attention to capture record-level context while preserving local waveform structure, and a SoftCLIP-style objective~\citep{gao2024softclip} to align a record-level summary with the conditioning labels.

Given a receiver-side baseband IQ waveform as described in Section \ref{sec:dataset}, \textit{i.e.,} $\mathbf{X}_{\mathrm{IQ}}\in\mathbb{C}^{4\times N_{\mathrm{samples}}}$, we expose its real and imaginary components and divide it into $P=\lceil N_{\mathrm{samples}}/N_{\mathrm{patch}}\rceil$ non-overlapping patches, where $N_{\mathrm{patch}}$ is the patch length. Each patch is processed by a convolutional raw-IQ branch and a ButterflyNet-inspired \citep{li2020butterflynetoptimalfunctionrepresentation, xu2020butterflynet2} spectral branch. A learned gate fuses the
two representations, and after attention across the patches (global patch mixing), we get $\mathbf{Z}_{\mathrm{IQ}}\in\mathbb{R}^{P\times768}$. At $N_{\mathrm{patch}}=1024$, this corresponds to $\sim11 \times$ compression per patch. Concisely, we represent this as $\mathbf{Z}_{\mathrm{IQ}} = \theta_{\mathrm{enc-IQ}}(\mathbf{X}_{\mathrm{IQ}})$. 
An IQ decoder parameterized by 
$\theta_{\mathrm{dec\text{-}IQ}}$ mixes the latent sequence across patches and
reconstructs the four-antenna waveform as
$\widehat{\mathbf{X}}_{\mathrm{IQ}}
=\theta_{\mathrm{dec\text{-}IQ}}(\mathbf{Z}_{\mathrm{IQ}})$, while a
\texttt{[CLS]} aggregate token further converts $\mathbf{Z}_{\mathrm{IQ}}$ into a
single record-level embedding $\mathbf{Z}_{\mathrm{IQ\text{-}CLS}}\in\mathbb{R}^{768}$. Meanwhile, a label encoder $\theta_{\mathrm{lbl}}$ maps the 6 labels of $\mathbf{X}_{\mathrm{IQ}}$ to a corresponding 768-dimensional embedding
$\mathbf{Z}_{\mathrm{lbl}}=\theta_{\mathrm{lbl}}(\mathbf{Y})$. 

The contrastive objective therefore operates on the normalized IQ embedding
$\mathbf z_i=\mathbf{Z}^{i}_{\mathrm{IQ\text{-}CLS}}/
\|\mathbf{Z}^{i}_{\mathrm{IQ\text{-}CLS}}\|_2$
and the normalized label embedding
$\mathbf c_j=\mathbf{Z}^{j}_{\mathrm{lbl}}/
\|\mathbf{Z}^{j}_{\mathrm{lbl}}\|_2$.
Their similarity is measured by the cosine logit
$\mathbf z_i^\top\mathbf c_j/\tau$. The symmetric soft-target CLIP loss
$\mathcal{L}_{\mathrm{CLIP}}$
\citep{radford2021clip,gao2024softclip} assigns greater weight to pairs with matching categorical labels and nearby continuous attributes, allowing semantically similar records to act as partial positives rather than treating every non-paired record as an equally negative example. The reconstruction objective $\mathcal{L}_{\mathrm{rec}}$, on the other hand, compares reconstructed
$\widehat{\mathbf{X}}^{i}_{\mathrm{IQ}}$ with
$\mathbf{X}^{i}_{\mathrm{IQ}}$ across complementary signal properties. Finally, a structural objective $\mathcal{L}_{\mathrm{struct}}$ is used to avoid learned embedding from collapsing. 

Overall, we jointly train \(\mathcal{L}_{\mathrm{S1}}
=\mathcal{L}_{\mathrm{rec}}
+\lambda_{\mathrm{CLIP}}(e)\mathcal{L}_{\mathrm{CLIP}}
+\lambda_{\mathrm{struct}}(e)\mathcal{L}_{\mathrm{struct}}\). During training, $\mathcal{L}_{\mathrm{rec}}$ and $\mathcal{L}_{\mathrm{struct}}$ are active throughout training. After an initial
reconstruction-focused warm-up, $\mathcal{L}_{\mathrm{CLIP}}$ is introduced, and $\lambda_{\mathrm{CLIP}}(e)$, $\lambda_{\mathrm{struct}}(e)$ are smoothly increased over subsequent epochs. Further implementation details and evaluations are provided in Appendix~\ref{app:encoder_architecture}. At the end of this stage, $\theta_{\mathrm{enc-IQ}}$, $\theta_{\mathrm{dec-IQ}}$, and $\theta_{\mathrm{lbl}}$ are frozen, fixing the latent space for stage 2 (~\ref{subsec:anchoredflow}).

\subsection{Stage 2: Spectrally Anchored Latent Flow Matching}
\label{subsec:anchoredflow}
To explicitly incorporate RF time-frequency structure into generation, we introduce \textit{spectral anchoring}. During flow matching (see Section~\ref{sec:background}), an intermediate representation of the velocity network is supervised to predict the corresponding spectrogram latent. Although the flow generates in the IQ latent space, our key insight is to constrain its intermediate representations, as in Figure~\ref{fig:alf}. This encourages the generated waveform to follow the target's global spectral structure, such as occupied bandwidth, at sampled flow times. In Appendix~\ref{app:theory}, we present a stylized model that shows spectral anchoring exposes the low-dimensional structure induced by a bandwidth constraint, reducing the parameter complexity of the feasible signal distribution relative to an unstructured time-domain parameterization.

Let $\mathbf{Z}_1$ denote the standardized IQ latent corresponding to the IQ waveform $\mathbf{X}_{\mathrm{IQ}}$. Since our objective is to learn the label-conditioned data distribution $p_{\mathrm{data}}(\cdot\mid\mathbf{Y})$, rather than initializing the flow from the same standard normal distribution for every label, we use the frozen probabilistic label encoder from Stage 1, $(\boldsymbol{\mu}_{\mathrm{lbl}}(\mathbf{Y}),\log\boldsymbol{\sigma}_{\mathrm{lbl}}^2(\mathbf{Y}))=\theta_{\mathrm{lbl}}(\mathbf{Y})$, to define the label-conditioned source distribution (implementation in \ref{sec:gaussian_prior}). For a recording containing $P$ patches, we independently sample $\mathbf{Z}_0^{(p)}\sim\mathcal{N}\!\left(\boldsymbol{\mu}_{\mathrm{lbl}}(\mathbf{Y}),\operatorname{diag}(\boldsymbol{\sigma}_{\mathrm{lbl}}^2(\mathbf{Y}))\right)$ for $p=1,\ldots,P$. These samples form $\mathbf{Z}_0$, and the conditional velocity network $v_{\theta}^{\mathrm{cond}}(\mathbf{Z}_t,t,\mathbf{Y})$ learns to transport them to $\mathbf{Z}_1$.

Importantly, to anchor the flow with frequency-domain information, we separately train a spectral autoencoder on the level-carrying log-power spectrogram associated with each IQ waveform (implementation in \ref{sec:spectral_anchor}). This autoencoder is distinct from the ButterflyNet-inspired branch of the Stage-1 IQ encoder. Its post-mixing latent sequence, denoted $\mathbf{Z}_{\mathrm{spec}}$, retains both patch-level spectral content and recording-level context. We then freeze the spectral autoencoder and independently standardize its output to obtain $\widetilde{\mathbf{Z}}_{\mathrm{spec}}$, which serves as the target for the anchoring objective.

\begin{table}[t]
\centering
\caption{TSTR/TRTS balanced accuracies for downstream objectives, by length (k = num classes)}
\label{tab:tstr-headline}
\resizebox{\textwidth}{!}{%
\setlength{\tabcolsep}{3pt}%
\begin{tabular}{@{}ll|cc|cc|cc|cc@{}}
\toprule
& & \multicolumn{2}{c|}{Waveform (k=32)} & \multicolumn{2}{c|}{Bandwidth (k=9)} & \multicolumn{2}{c|}{Sample Rate (k=13)} & \multicolumn{2}{c}{Carrier Freq (k=30)} \\
\cmidrule(lr){3-4}\cmidrule(lr){5-6}\cmidrule(lr){7-8}\cmidrule(lr){9-10}
Length & Method & TSTR & TRTS & TSTR & TRTS & TSTR & TRTS & TSTR & TRTS \\
\midrule
\multirow[c]{4}{*}{2,048}
& WaveStitch & 0.02 & 0.04 & 0.11 & 0.11 & 0.08 & 0.09 & 0.04 & 0.04 \\
& RF Diffusion$^\dagger$ & 0.03 & 0.03 & 0.11 & 0.11 & 0.08 & 0.08 & 0.04 & 0.04 \\
& Time Weaver & 0.59 & 0.33 & 0.39 & 0.26 & 0.39 & 0.21 & 0.24 & 0.16 \\
& \textbf{ALF (ours)} & \textbf{0.84} & \textbf{0.83} & \textbf{0.63} & \textbf{0.56} & \textbf{0.70} & \textbf{0.68} & \textbf{0.27} & \textbf{0.30} \\
\midrule
\multirow[c]{3}{*}{4,096}
& WaveStitch & 0.03 & 0.03 & 0.12 & 0.11 & 0.07 & 0.08 & 0.03 & 0.03 \\
& Time Weaver & 0.60 & 0.37 & 0.44 & 0.29 & 0.41 & 0.17 & 0.24 & 0.16 \\
& \textbf{ALF (ours)} & \textbf{0.85} & \textbf{0.82} & \textbf{0.62} & \textbf{0.58} & \textbf{0.67} & \textbf{0.65} & \textbf{0.29} & \textbf{0.32} \\
\midrule
\multirow[c]{3}{*}{8,192}
& WaveStitch & 0.03 & 0.03 & 0.12 & 0.11 & 0.08 & 0.08 & 0.03 & 0.04 \\
& Time Weaver & 0.61 & 0.39 & 0.46 & 0.29 & 0.42 & 0.19 & 0.23 & 0.18 \\
& \textbf{ALF (ours)} & \textbf{0.84} & \textbf{0.81} & \textbf{0.61} & \textbf{0.57} & \textbf{0.70} & \textbf{0.66} & \textbf{0.28} & \textbf{0.30} \\
\midrule
\multirow[c]{3}{*}{16,384}
& WaveStitch & 0.03 & 0.03 & 0.12 & 0.11 & 0.07 & 0.08 & 0.03 & 0.03 \\
& Time Weaver & 0.60 & 0.43 & 0.45 & 0.31 & 0.43 & 0.20 & 0.22 & 0.19 \\
& \textbf{ALF (ours)} & \textbf{0.83} & \textbf{0.81} & \textbf{0.61} & \textbf{0.58} & \textbf{0.70} & \textbf{0.66} & \textbf{0.28} & \textbf{0.29} \\
\midrule
\multirow[c]{3}{*}{32,768}
& WaveStitch & 0.03 & 0.03 & 0.12 & 0.11 & 0.07 & 0.08 & 0.03 & 0.03 \\
& Time Weaver & 0.62 & 0.44 & 0.52 & 0.37 & 0.44 & 0.24 & 0.23 & 0.20 \\
& \textbf{ALF (ours)} & \textbf{0.83} & \textbf{0.79} & \textbf{0.70} & \textbf{0.64} & \textbf{0.72} & \textbf{0.68} & \textbf{0.28} & \textbf{0.30} \\
\midrule
\multirow[c]{3}{*}{\textbf{Avg.}}
& WaveStitch & 0.03 & 0.03 & 0.12 & 0.11 & 0.07 & 0.08 & 0.03 & 0.03 \\
& Time Weaver & 0.61 & 0.40 & 0.46 & 0.31 & 0.42 & 0.20 & 0.23 & 0.18 \\
& \textbf{ALF (ours)} & \textbf{0.84} & \textbf{0.81} & \textbf{0.64} & \textbf{0.59} & \textbf{0.70} & \textbf{0.66} & \textbf{0.28} & \textbf{0.30} \\
\midrule
\multirow[c]{1}{*}{\shortstack{\textbf{ALF} \textbf{Advantage}}}
& \textbf{vs.\ Time Weaver}
& \textbf{38.02\%} & \textbf{102.88\%} & \textbf{38.29\%} & \textbf{91.51\%} & \textbf{64.09\%} & \textbf{228.97\%} & \textbf{21.47\%} & \textbf{66.69\%} \\
\bottomrule
\end{tabular}%
}
\parbox{\textwidth}{\footnotesize\emph{$\dagger$ RF Diffusion could not be trained for recordings of length greater than 2,048 due to memory limitations; longer recordings were not attempted.}}
\end{table}
To connect this spectral target to the flow, we attach a lightweight prediction head $\theta_{\mathrm{anchor}}$ to the hidden representation after the sixth Transformer block (Figure~\ref{fig:intro-figure}). At each sampled flow time $t$, the head maps $\mathbf{H}_t^{(6)}$ to an estimate of the standardized clean spectrogram latent $\widetilde{\mathbf{Z}}_{\mathrm{spec}}$. The target remains fixed across flow times, while its loss contribution is weighted by $t/\bar{t}$ to emphasize states closer to the target IQ distribution. We sample $t\sim p_t$ by drawing $\epsilon_t\sim\mathcal{N}(0,1)$ and setting $t=\operatorname{sigmoid}(\epsilon_t)$, with $\bar{t}=\mathbb{E}_{t\sim p_t}[t]$. The anchor head and velocity network are trained jointly with the following loss:
\begin{equation}
\resizebox{\dimexpr\linewidth-3em\relax}{!}{$\displaystyle
\mathcal{L}_{\mathrm{ALF}}
=
\mathbb{E}_{\substack{
(\mathbf{Z}_1,\mathbf{Y},\widetilde{\mathbf{Z}}_{\mathrm{spec}})
\sim p_{\mathrm{data}};\,
\mathbf{Z}_0\sim p_{\mathrm{src}}(\cdot\mid\mathbf{Y});\,
t\sim p_t
}}
\left[
\underbrace{
\left\|
v_{\theta}^{\mathrm{cond}}(\mathbf{Z}_t,t,\mathbf{Y})
-
(\mathbf{Z}_1-\mathbf{Z}_0)
\right\|_2^2
}_{\text{\rm Flow-matching loss}}
+
\underbrace{
\lambda_{\mathrm{anchor}}(e)\,
\frac{t}{\bar{t}}\,
\left\|
\theta_{\mathrm{anchor}}\!\left(\mathbf{H}_t^{(6)}\right)
-
\widetilde{\mathbf{Z}}_{\mathrm{spec}}
\right\|_2^2
}_{\text{\rm Spectral anchoring loss}}
\right].
$}
\label{eq:alf_objective}
\end{equation}

We ramp $\lambda_{\mathrm{anchor}}(e)$ to $0.25$ over the first two training epochs. The first term in Equation~\ref{eq:alf_objective} learns the label-conditioned transport from $\mathbf{Z}_0$ to $\mathbf{Z}_1$, while the second shapes the sixth-block representation using the target spectrogram latent. At inference, the generated $\mathbf{Z}_1$ is reconstructed using the frozen decoder $\theta_{\mathrm{dec\text{-}IQ}}$; the spectral autoencoder and anchor head are omitted.

\section{Results}
\label{sec:results}

Due to the highly structured nature of RF signals, the differences between valid and undecipherable sequences are hard to differentiate using standard metrics. Instead, it has been noted in the wireless community that demonstrating downstream performance is a better judge of generation quality \citep{baur2025evaluation}. To this end, we can show utility to waveform, channel, and receiver classification using our chosen labels (Appendix~\ref{ap:dataset-labels}). Although distributional metrics are insufficient measures of quality in this domain, comparisons across methods using PSD and Spectrogram MSEs can be found in Appendix~\ref{ap:distributional-metrics}. Table~\ref{tab:tstr-headline} compares ALF against related time series methods trained on the same 6 label conditions. 

\textbf{Procedures:} Classifications use lightweight, length-agnostic, 1-D CNN classifiers trained for 8,000 steps across generated corpora of 100k recordings and evaluated on held-out test sets of size 35,856. Permuting generated (synthetic) and reference (real) data as training and evaluation sets produces various metrics that characterize a model's capability in a downstream classification job.

\textbf{Baselines:} We compare conditional time-series generative models from \citet{narasimhan2024time}, \citet{chi2024rfdiffusion}, and \citet{shankar2025wavestitch}. These were selected for their RF focus or state-of-the-art metadata-conditioned generation. Table~\ref{tab:tstr-headline} shows how WaveStitch and RF Diffusion experience a total collapse when evaluated on downstream tasks. Time Weaver and ALF are the only models capable of achieving any downstream utility \citep{narasimhan2024time}.

\textbf{Metrics:} Train-on-Synthetic, Test-on-Real \textbf{(TSTR)} attempts to classify unseen, \textit{real} data after learning to recognize generated data, alluding to a corpus' indistinguishability from the truth. This requires a diverse range of representative generations that span all "tastes" of a signal---corpus generality. Where TSTR examines the entire dataset, Train-on-Real, Test-on-Synthetic \textbf{(TRTS)} evaluates if a synthetic dataset can be categorized correctly using an understanding of real signal features. These accuracies serve as certificates of individual recording-level fidelity.

\textbf{ALF shows high utility over baselines on frequency-dependent tasks.} With respect to waveform classification, ALF shows significant improvements of 38.02\% TSTR and 102.88\% TRTS performances (Table~\ref{tab:tstr-headline}). TSTR confirms that generations have generality, and TRTS confirms recording-level specificity. Bandwidth and sample rate classification improvements demonstrate how \textit{spectrogram reasoning} learns representations of derived labels (Appendix~\ref{ap:dataset-labels}). Furthermore, utility is shown for unseen labels in carrier frequency estimation. Compared to other models, ALF exploits the domain-specific biases necessary for downstream evaluations. More seen and unseen classifications (Indoor/Outdoor, Carrier Frequency, Delay Spread) can be found in Appendix~\ref{ap:add-quality-metrics}.

\begin{table}[t]
\centering
\setlength{\tabcolsep}{4pt}
\renewcommand{\arraystretch}{1.0}
\setlength{\aboverulesep}{0.3ex}
\setlength{\belowrulesep}{0.3ex}
\begin{minipage}{\textwidth}
\centering
\caption{(a) Anchored vs. unanchored ablations across RF metrics.
(b) Reclassification under low-density samples $\mathcal{N}(\mu,(k\sigma)^2)$. 50 steps per generation.}
\label{tab:combined-rf-robustness}
\small
\begin{tabular*}{\linewidth}{@{\extracolsep{\fill}}lrrrr@{\hspace{12pt}}rrrr@{}}
\toprule
\textbf{(a) Ablation}
& \multicolumn{4}{c}{Performance}
& \multicolumn{4}{c}{Gain vs.\ base LF (\%)} \\
\cmidrule(lr){2-5}\cmidrule(lr){6-9}
Method & TSTR$^{\dagger}$ $\uparrow$ & PSD $\downarrow$ & PAPR $\downarrow$ & CV $\downarrow$
& TSTR & PSD & PAPR & CV \\
\midrule
LF-N & 81.93 & 17.87 & 4.99 & 0.139 & --- & --- & --- & --- \\
LF-M & 83.49 & 17.48 & 4.87 & 0.133 & +1.9 & +2.2 & +2.4 & +3.9 \\
ALF & \textbf{84.49} & \textbf{16.86} & \textbf{4.70} & \textbf{0.126}
& \textbf{+3.1} & \textbf{+5.6} & \textbf{+5.7} & \textbf{+9.3} \\
\bottomrule
\end{tabular*}
\par\vspace{3pt}
\begin{tabular*}{\linewidth}{@{\extracolsep{\fill}}lrrr@{\hspace{12pt}}rrr@{}}
\toprule
\textbf{(b) Robustness} & \multicolumn{3}{c}{Accuracy (\%) $\uparrow$}
& \multicolumn{3}{c}{ALF gain over $X$ (\%)} \\
\cmidrule(lr){2-4}\cmidrule(lr){5-7}
Method $X$ & $k=1$ & $k=3$ & $k=5$ & $k=1$ & $k=3$ & $k=5$ \\
\midrule
ALF & \textbf{96.5} & \textbf{53.6} & \textbf{11.7} & --- & --- & --- \\
LF-M & 92.6 & 45.0 & 7.8 & \textbf{+4.2} & \textbf{+19.2} & \textbf{+50.0} \\
LF-N & 90.1 & 12.5 & 3.9 & \textbf{+7.1} & \textbf{+329.5} & \textbf{+200.0} \\
Time Weaver & 39.0 & 4.6 & 3.3 & \textbf{+147.3} & \textbf{+1075.3} & \textbf{+257.1} \\
\bottomrule
\end{tabular*}
\par\vspace{2pt}
\begin{minipage}{\linewidth}
\footnotesize
Note: LF-M/LF-N: unanchored latent flow with metadata/standard-normal priors; (a) $^{\dagger}$Seed 100. PSD: median PSD MSE (dB$^2$);
PAPR: error (dB); CV: envelope-CV error. (b) classifiers pretrained on ground truth. All ablations use a refiner for SNR over 30db (Appendix~\ref{ap:ref-unref})
\end{minipage}
\end{minipage}
\end{table}

\textbf{Anchoring improves RF-specific metrics, sampling robustness, and parameter complexity.}
Before presenting numerical ablations, we note the following. With an ``informative'' (support-preserving) spectral cross-domain anchor, we theoretically show that bandwidth constraints in the spectrum reduces model parameter complexity (in effect, a smaller model suffices for high-fidelity generation). A detailed exposition is in Appendix~\ref{app:theory}.

Next, Table~\ref{tab:combined-rf-robustness}~(a) shows that spectral anchoring provides gains beyond metadata priors alone for TSTR and specific RF amplitude properties. Together, PSD, PAPR, and envelope CV assess spectral power distribution, relative peak power, and overall amplitude variability. ALF outperforms both label-source initialization and base LF, with consistent gains in spectral and I/Q amplitude fidelity. Table~\ref{tab:combined-rf-robustness}~(b) demonstrates how anchoring slows the rate of generation collapse when sampling from low-density areas ($k\sigma$ from mean $\mu$). 

\textbf{Anchoring reasons through spectrotemporal structure.} To inspect the impact of latent space anchoring, we compare decoded intermediate representations of ALF generation trajectories to an unanchored latent flow (LF) ablation using the pre-trained spectrogram autoencoder.

Figure~\ref{fig:anchors-and-coherence} shows that ALF reasons an identifiable Costas hopping pattern, whereas the baseline shape remains ambiguous. As lengths scale, sharper frequency patterns better constrain time-domain generation, resulting in compounding sample coherence advantages. 
\begin{figure}[h]
    \centering
    \includegraphics[width=\linewidth]{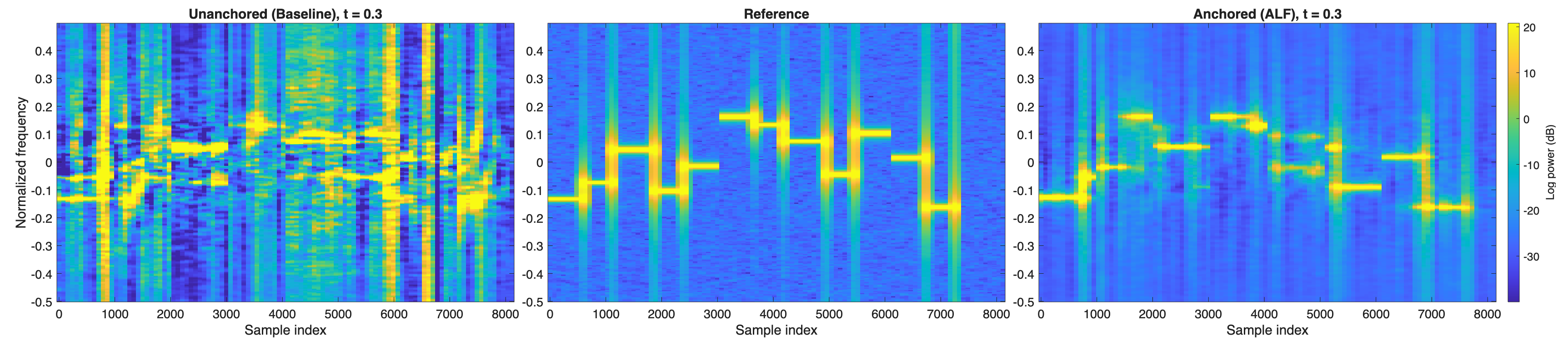}
    \vspace{-10pt}
    \caption{Visualization of ALF vs. LF (unanchored) latent space as decoded spectrograms. LF borrows the lightweight projection head from ALF.}
    \label{fig:anchors-and-coherence}
\end{figure}

\begin{figure}[!t]
    \centering
    \includegraphics[width=\linewidth]{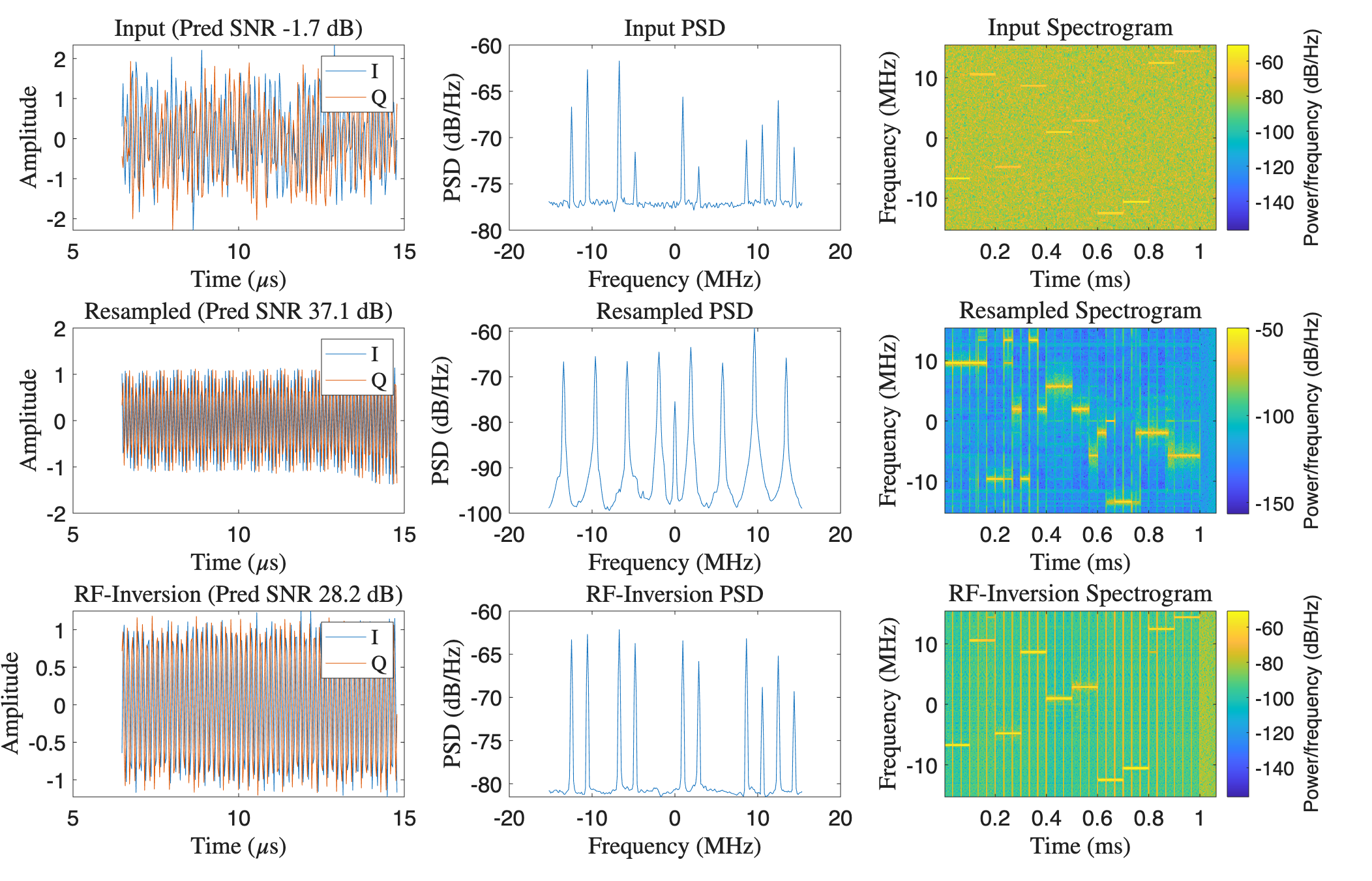}
    \vspace{-20pt}
    \caption{Resampling vs. RF-Inversion style transfer}
    \label{fig:style-transfer}
    \vspace{-10pt}
\end{figure}

\textbf{ALF reduces the inference bottleneck.} This is a major barrier to meaningful wireless signal generation at-length: it is time-prohibitive to generate long sequences. Although inference complexity remains bound by attention scaling laws \citep{vaswani2017attention}, ALF's latent reasoning architecture and per-patch decoding manifest as two orders of magnitude in inference time reduction to Time Weaver, along with its generative quality boost (see Tables~\ref{tab:inf-speed} and \ref{tab:tstr-headline}). Other models shorten this gap, but at the cost of unusable data. These gains make lengths of 100k+ feasible in fractions of a second.

\textbf{Our two-stage architecture enables native and efficient style transfer.}
The field \modelname\space learns flows from label-conditioned source latents $\mathbf{Z}_0$ to signal latents $\mathbf{Z}_1$ ($\S$~\ref{subsec:anchoredflow}). The learned field can also be integrated backward from $t=1$ to $t=0$, without a stochastic reverse process. Thus, \modelname\space has a clean analysis/synthesis pipeline, making it a model native to style transfer. We edit held-out recordings by encoding each input signal and reversing the learned field to obtain its source-space representation. Using a Gaussian log-likelihood search, we find the most likely label for that representation.

Style transfer, including signal denoising (see Figure~\ref{fig:style-transfer}), is executed in one of two ways. \textbf{Fresh Prior Resample}: Given the estimated label $\hat{c}$, increase SNR and regenerate by sampling from the higher SNR mode. \textbf{RF-Inversion}: Given the observed record and a higher SNR target, use a controller at inference to guide it to the corrected region of the reconstruction space \citep{rout2025semantic}. Figure~\ref{fig:style-transfer} showcases \modelname's capability of denoising recordings that are near or below the noise floor; likewise, this demonstrates learned signal representations in harsher SNR regimes. Notably, a frequency-hopped frequency modulation (FHFM) example shows the difference in output similarity between naive resampling and a controlled approach \citep{rout2025semantic}. Notably, although RF-Inversion marginally decreases the generated SNR, this approach maintains the original signal information in the hopping sequence while still effectively denoising. We also show how latent anchor guidance helps prevent occasional spectral blurring (\ref{app:anchored-vs-unanchored-ST}).

\begin{table*}[h]
\centering
\small
\caption{Latency and allocated-memory comparison at B=16 on one NVIDIA H100 GPU.}
\label{tab:inf-speed}
\resizebox{\textwidth}{!}{%
\begin{tabular}{cc|ccc|ccc}
\toprule
\multirow{2}{*}{Batch} &
\multirow{2}{*}{$N_{samples}$} &
\multicolumn{3}{c|}{Latency (ms/rec)} &
\multicolumn{3}{c}{Allocated Memory (MiB)} \\
\cmidrule(lr){3-5} \cmidrule(lr){6-8}
&
& \textbf{ALF} & TimeWeaver & ALF vs.\ TimeWeaver
& \textbf{ALF} & TimeWeaver & ALF vs.\ TimeWeaver \\
\midrule
16 & 2048
& \textbf{25.5} & 447.1  & $18\times$
& 1045.1 & \textbf{816.0}  & $0.8\times$ \\

16 & 4096
& \textbf{25.9} & 833.1  & $32\times$
& \textbf{1173.4} & 1272.0 & $1.1\times$ \\

16 & 8192
& \textbf{25.8} & 1718.0 & $67\times$
& \textbf{1431.9} & 2184.0 & $1.5\times$ \\

16 & 16384
& \textbf{25.7} & 3514.1 & $137\times$
& \textbf{1947.1} & 4006.0 & $2.1\times$ \\

16 & 32768
& \textbf{34.4} & 7242.1 & $211\times$
& \textbf{2976.8} & 7652.0 & $2.6\times$ \\

\midrule
\textbf{16} & \textbf{Avg}
& \textbf{27.5} & \textbf{2750.9} & $\mathbf{100\times}$
& \textbf{1714.9} & \textbf{3186.0} & $\mathbf{1.9\times}$ \\
\bottomrule
\end{tabular}%
}
\end{table*}
\textbf{Autoregressive generation unlocks arbitrary-length sequences and distribution shifts.} We demonstrate this adaptation of ALF through \textit{suffix generation}. This post-training learns to flow new patches alongside fixed previously generated tokens. Thus, new sections now learn two conditions: their target label, and a context window of the signal. Once in the reconstruction space, patches are decoded without their prefix tokens. We provide additional information in Appendix~\ref{ap:ar-gen}.
\vspace{5pt}
\begin{figure}[h]
    \centering
    \includegraphics[width=1.0\linewidth]{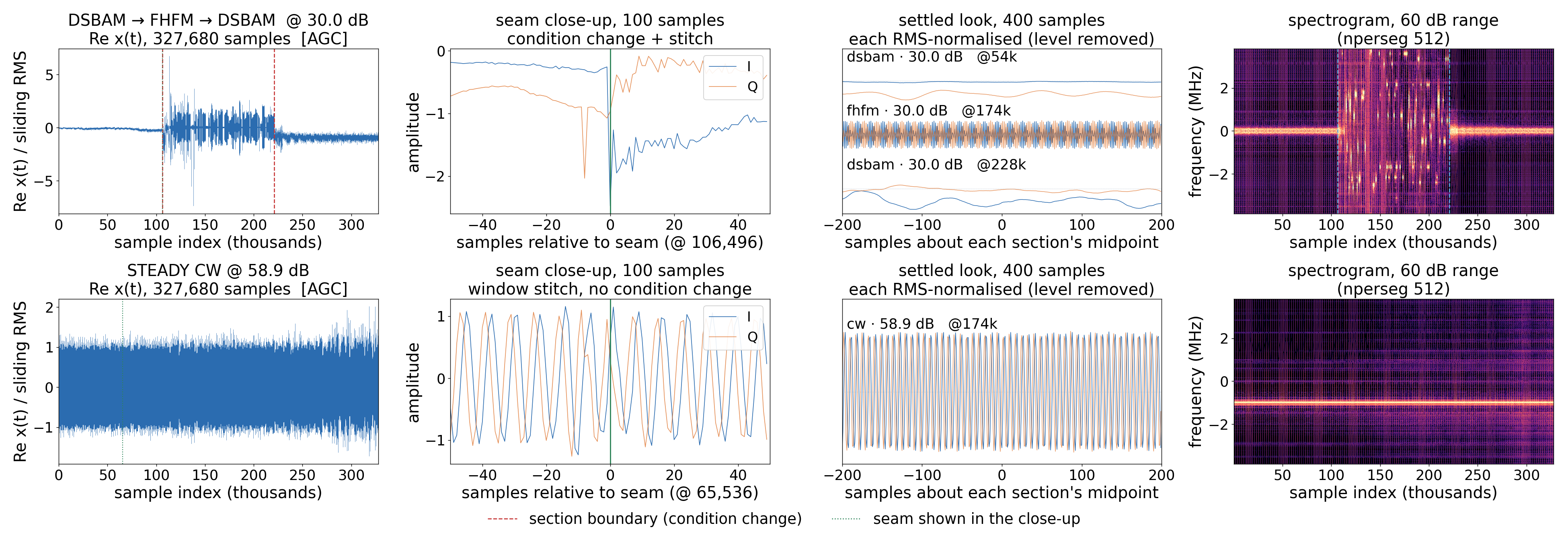}
    \vspace{-20pt}
    \caption{327k-sample-long AR generations show waveform shifts and single-condition continuity}
    \label{fig:ar-singlerow}
\end{figure}

\textbf{Limitations:} It is notable that ALF does not show improvement on scene- or propagation-centered tasks. These binary classifications are characterized by energy fluctuations which are normalized away, resulting in noise readings ($\S$~\ref{ap:add-fig-tex}, Table~\ref{tab:norm-vs-rescaled}). However, TRTS metrics show ALF has equal, if not better, features across individual time series. This suggests that although our model's corpus is less general than Time Weaver in these tasks, records maintain scene-specific features.

Our results suggest that spectral anchoring shapes the flow's intermediate representations and improves generation without adding inference-time computation. Nevertheless, several limitations remain. First, \datasetname{} has limited payload diversity and uses simulated channels that remain static within each recording, omitting mobility and Doppler evolution, hardware impairments, interference, and other variability present in real transmissions. Second, the current anchor is derived from a log-power spectrogram and therefore emphasizes time-frequency energy without explicitly encoding spectral phase or cross-antenna relationships. Finally, the present evaluation focuses primarily on distributional similarity and downstream classification; broader evaluation of reversible-flow transformations, receiver-level performance, and over-the-air recordings remains necessary.

\vspace{-5pt}
\section{Conclusion}
\label{sec:conclusion}

We introduced \modelname{}, a latent flow-matching framework for conditional generation of long, multi-antenna IQ recordings, together with \datasetname{}, a dataset spanning diverse waveform classes, propagation conditions, and record lengths. By combining metadata-conditioned initialization with training-time spectral anchoring, \modelname{} improves spectral fidelity and downstream utility without adding inference-time computation. Across held-out channels and all evaluated lengths, it reaches an average TSTR accuracy of 84.5\%, a 38.02\% relative improvement over SOTA conditional time-series generator Time Weaver while generating records $100\times$ faster on average.

\noindent\textbf{Future work.} Important next steps include training and evaluating on over-the-air recordings, including receiver-level validation under hardware and environmental impairments. Alternative flow constraints and anchor representations could incorporate phase, cross-antenna structure, multiple time-frequency resolutions, or supervision at different flow layers and timesteps. Extending the model to substantially longer recordings and stateful streaming generation would further test whether the architecture can maintain coherence beyond the sequence lengths considered here.

\section*{Acknowledgments}

This work was supported in part by NSF Grants 2112471 (NSF AI EDGE) and 2505865 (NSF IFML), and the Wireless Networking and Communications Group (WNCG) Industrial Affiliates Program. The work of S. S. Narasimhan was supported by Lockheed Martin Corporation under the Causal Inference To Assess/Deny Emitters via Learning (CITADEL) and the Sensor Digital RF Data Tokenization for use in Generative AI initiatives. We are grateful for computing support on the Vista GPU Cluster through the Center for Generative AI (CGAI) and the Texas Advanced Computing Center (TACC) at the University of Texas at Austin.

\bibliography{iclr2027_conference}
\bibliographystyle{iclr2027_conference}

\clearpage
\appendix
\section*{Appendix}

\newtheorem{assumption}{Assumption}
\theoremstyle{remark}
\theoremstyle{definition}
\newtheorem{example}{Example}

\newcommand{\cS}{\mathcal{S}}
\newcommand{\cC}{\mathcal{C}}
\newcommand{\cX}{\mathcal{X}}
\newcommand{\cK}{\mathcal{K}}
\newcommand{\cD}{\mathcal{D}}
\newcommand{\supp}{\operatorname{supp}}

\section{Spectral Anchoring Parameter Complexity}
\label{app:theory}

\subsection{Overview}

In this section, we argue that spectral anchoring is particularly effective for wireless signal generation. Specifically, anchoring by the spectrogram explicitly exposes the low-dimensional structure induced by band limit constraints, thereby reducing the parameter complexity of an unstructured time-domain representation. Consider the signal, 

\[
x = (x_0, \ldots, x_{n-1}) \in \mathcal{D}^{n}.
\]

Generally, such a signal has $n$ degrees of freedom. Let $d=|\mathcal{D}|$. Since probabilities must sum to 1, there are $d^n-1$ free parameters\footnote{A generative model, in essence, learns a joint distribution. With discrete/finite-alphabet, equivalently, it learns to fill-in the entries of a joint/conditional distribution table (as example, see the table for $p_{Y|C}(y \mid c)$ in Appendix~\ref{app:pf-setup} below), where each entry in this table can be thought of as a ``parameter'' that is to be learned. Dependencies across entries in the table can be exploited to reduce the number of parameters to be learned, leading to the ``free parameter'' count. A smaller number of free parameters implies that a smaller capacity model suffices to learn a distribution with high fidelity. The free parameter count can be translated to a training sample complexity bound using standard concentration inequalities. Concretely, suppose that an $\epsilon$-accurate estimate is needed for each free parameter entry, with probability at least $1-\delta$. A standard concentration bound (Hoeffding/Chernoff) implies that each entry requires order of $\epsilon^{-2}\ln(\nicefrac{1}{\delta})$ samples. More generally, learning a distribution with $D$ free parameters (i.e., supported on $D+1$ outcomes) to within total variation distance $\epsilon$ requires order of $(D+\ln(\nicefrac{1}{\delta}))/\epsilon^{2}$ samples, so the sample complexity grows linearly with the free parameter count.
Thus in this section, we focus directly on the free parameter analysis for spectral anchoring.}. 
If the signal is band-limited to $k$ frequency bins, however, the DFT of $x$ exposes a useful structure (see below in Appendix~\ref{app:pf-setup} for formal definitions). 

\[
\hat{x}
=
(\underbrace{0,\ldots,0}_{\text{known}},
\hat{x}_{j_1},\ldots,\hat{x}_{j_k},
\underbrace{0,\ldots,0}_{\text{known}}).
\]

We can observe that the signal has $k$ degrees of freedom under the Fourier basis. The key insight is that, because the Fourier transform is invertible, this reduction in effective degrees of freedom also constrains the time-domain signal. While band limitation is explicit in the frequency domain through the restriction to \(k\) active coefficients, the corresponding structure is implicit in the time domain, where it manifests as dependencies across the \(n\) samples. Spectral anchoring reveals these constraints directly to the generative model, reducing the parameter complexity required relative to an unstructured time-domain parameterization. In the subsequent exposition, using a stylized simple model, we formalize this intuition. While this model differs from the details of our implementation, we believe that it captures some of the core intuition motivating our approach.

\subsection{Setup}
\label{app:pf-setup}

\paragraph{Signals and spectra.}
Let $\mathcal{D} \subset \mathbb{C}$ be a finite set of $d = |\mathcal{D}|$ possible I/Q values, and let
\[
x = (x_0, \ldots, x_{n-1}) \in \mathcal{D}^n
\]
denote a length-$n$ I/Q sequence, where each sample $x_i$ takes a value in $\mathcal{D}$.

The DFT of $x$ is
\begin{equation}
  \hat{x}_j
  =
  \sum_{t=0}^{n-1}
  x_t e^{-2\pi i jt/n},
  \qquad j=0,\dots,n-1.
  \label{eq:dft}
\end{equation}
Now we will define a set of band-limited signals. For a set of permissible frequency indices
$K\subseteq\{0,\dots,n-1\}$, let
\begin{equation}
  \cX_K
  =
  \bigl\{\,x\in\cD^n :
  \hat{x}_j=0 \text{ for all } j\notin K\,\bigr\}
  \label{eq:XK}
\end{equation}
be the signals whose spectrum is supported on $K$. 


\subsection{Unanchored Parameter Count}

\paragraph{Counting Strategy.}

To count the free parameters required to represent a conditional distribution
$p_{Y|C}(y \mid c)$, suppose that $C$ can take $L$ possible values
$c_1,\ldots,c_L$ and that $Y$ can take $d$ possible values
$y_1,\ldots,y_d$. We represent the distribution as a conditional probability
table,
\[
\begin{array}{c|cccc}
& Y=y_1 & Y=y_2 & \cdots & Y=y_d \\
\hline
C=c_1
& p(y_1\mid c_1) & p(y_2\mid c_1) & \cdots & p(y_d\mid c_1) \\
C=c_2
& p(y_1\mid c_2) & p(y_2\mid c_2) & \cdots & p(y_d\mid c_2) \\
\vdots
& \vdots & \vdots & \ddots & \vdots \\
C=c_L
& p(y_1\mid c_L) & p(y_2\mid c_L) & \cdots & p(y_d\mid c_L)
\end{array}
\]
where each row corresponds to a possible conditioning context and each column
corresponds to a possible value of the predicted variable. Since the
probabilities in each row sum to one, each row contains $d-1$ free parameters (since probabilities must sum to 1). Therefore, representing the full conditional distribution requires
\[
L(d-1)
\]
free parameters.

\paragraph{Applying chain rule.}

First consider the case where the model has no information about the spectrum of the
recording (i.e. no anchor). The unanchored joint is,

\begin{equation}
  p(x) \;=\; \prod_{i=1}^{n} p\bigl(x_i \mid x_{<i}\bigr).
  \label{eq:unanchored_chain}
\end{equation}

It must therefore be able to represent any distribution on
$\cD^n$, so it allocates a row in the counting table to every prefix $x_{<i}$. That is, $L=d^{i-1}$ and the number of parameters is 
\begin{equation}
  D_{\mathrm{un}} \;=\; (d-1)\sum_{i=1}^{n} d^{\,i-1} \;=\; d^{\,n}-1 .
  \label{eq:D-un}
\end{equation}

\subsection{Anchored Parameter Count}

Recall the set of band-limited signals $\cX_K$ from \eqref{eq:XK}. Let $\mathcal{K}$ be a family of frequency support patterns, each of size at most $k$. Let $\mathcal{X}_{\mathcal{K}} = \bigcup_{K \in \mathcal{K}} \mathcal{X}_K$, denoting the set of signals satisfying any allowed $\mathcal{K}$. We will argue that a band-limited signal may have \(n\) time-domain samples, but if it contains only \(|K|\) active frequency components, it has only \(|K|\) independent degrees of freedom.

\begin{lemma}[Band-limited signal identifiability]
\label{thm:vandermonde}
For every $K\subseteq\{0,\dots,n-1\}$, the first $|K|$ samples of any
$x\in\cX_K$ determine $x$.
\label{lem:signa-id}
\end{lemma}

\begin{proof}
Consider any $x\in\cX_K$. By definition of $\cX_K$, the Fourier
coefficients outside $K$ are zero:
\[
\hat{x}_j=0
\qquad\text{for all }j\notin K.
\]
Thus, although the DFT contains $n$ coefficients in total, only the
$|K|$ coefficients $(\hat{x}_j)_{j\in K}$ are unknown. Since $\hat{x}_j=0$ for $j\notin K$, the inverse DFT reduces to

\begin{equation}
    x_t
    =
    \frac{1}{n}
    \sum_{j\in K}
    \hat{x}_j e^{2\pi ijt/n},
    \qquad t=0,\dots,n-1.
\label{inverse-dft}
\end{equation}

Now restrict this equation to the first $|K|$ time samples,
$t=0,\dots,|K|-1$. This gives $|K|$ linear equations in the $|K|$
unknown Fourier coefficients $(\hat{x}_j)_{j\in K}$. Writing
$K=\{j_1,\dots,j_{|K|}\}$, these equations can be expressed as
\[
\begin{bmatrix}
x_0\\
x_1\\
\vdots\\
x_{|K|-1}
\end{bmatrix}
=
\frac{1}{n}
\begin{bmatrix}
1 & 1 & \cdots & 1\\
e^{2\pi i j_1/n} & e^{2\pi i j_2/n} & \cdots & e^{2\pi i j_{|K|}/n}\\
(e^{2\pi i j_1/n})^2 & (e^{2\pi i j_2/n})^2 & \cdots &
(e^{2\pi i j_{|K|}/n})^2\\
\vdots & \vdots & \ddots & \vdots\\
(e^{2\pi i j_1/n})^{|K|-1}
&
(e^{2\pi i j_2/n})^{|K|-1}
&
\cdots
&
(e^{2\pi i j_{|K|}/n})^{|K|-1}
\end{bmatrix}
\begin{bmatrix}
\hat{x}_{j_1}\\
\hat{x}_{j_2}\\
\vdots\\
\hat{x}_{j_{|K|}}
\end{bmatrix}.
\]
The middle matrix is a Vandermonde matrix (i.e. geometric progression in each column), so it is invertible. Hence the
first $|K|$ samples uniquely determine all active Fourier coefficients
$(\hat{x}_j)_{j\in K}$. Substituting the coefficients
back into the inverse DFT in equation~\ref{inverse-dft} determines $x_t$ for every $t=0,\dots,n-1$. Therefore, the first $|K|$ time samples
uniquely determine the signal.
\end{proof}

\begin{corollary}[Counting band-limited grid signals]
\[
|\cX_K|\le d^{|K|}
\qquad\text{and}\qquad
|\cX_\cK|\le|\cK|d^k
\]
\label{corr:counting-band}
\end{corollary}
\begin{proof}
We first count the number of distinct length-\(n\) time-domain signals that are possible under a particular bandwidth constraint. Since each time-domain sample
takes a value in $\cD$ and $|\cD|=d$, there are at most $d^{|K|}$ possible tuples $(x_0,\dots,x_{|K|-1})$. Because each such tuple
determines at most one signal in $\cX_K$ by Lemma~\ref{lem:signa-id}, it follows that
\[
|\cX_K|\le d^{|K|}.
\]

Now count the number of distinct signals that satisfy any bandwidth constraint in $\cK$. Recall,
\[
\cX_\cK
=
\bigcup_{K\in\cK}\cX_K.
\]
Since every $K\in\cK$ satisfies $|K|\le k$,
\[
|\cX_K|
\le
d^{|K|}
\le
d^k.
\]
Therefore we can applied a union bound,
\[
|\cX_\cK|
=
\left|\bigcup_{K\in\cK}\cX_K\right|
\le
\sum_{K\in\cK}|\cX_K|
\le
\sum_{K\in\cK}d^k
=
|\cK|d^k.
\]
\end{proof}

We formally define an anchor. While this work uses the spectrogram operator, we can consider the broader class of anchors given by deterministic functions of a signal that preserve its spectral support.

\begin{assumption}[Band-limited spectral anchor]
\label{ass:anchor}
The anchor is $s=A(x)$ for a map $A$ with
$\supp A(x)=\supp\hat{x}$. The anchor alphabet contains only band-limited
spectra, i.e., $\supp s\subseteq K$ for some $K\in\cK$ whenever
$s\in\cS$. The joint is supported on consistent pairs (i.e. a signal cannot be paired with two different anchors):
$p(x,s)>0 \Rightarrow A(x)=s$.
\end{assumption}

Assumption~\ref{ass:anchor} requires the anchor to be computed from the
signal itself, so it adds no randomness of its own. Thus, once $x$ is known, so is $s$. Because $s_j=0$ exactly when $\hat{x}_j=0$, the anchor always
reveals which frequency components are active, even when it leaves out
details such as phase, and the band-limit constraint can be read directly
from $s$. Every anchor corresponds to a signal satisfying one of the allowed bandwidth constraints in \(\cK\), and each signal is uniquely paired with its corresponding anchor. Conditioning on \(s\) therefore identifies which of the \(n\) frequency bins are active, while the anchor may contain varying amounts of information about their values. Both anchors we study below satisfy
the assumption. The full spectrum $s_j=\hat{x}_j$ reveals the active bins
together with their exact values, and the magnitude spectrum
$s_j=|\hat{x}_j|$ reveals the active bins together with their energy only.

Under Assumption~\ref{ass:anchor}, the sequences compatible with an anchor
value $s$ form the set
\begin{equation}
  \cX(s) \;=\; \{\,x\in\cD^n : A(x)=s\,\}\;\subseteq\;\cX_{\supp s}.
  \label{eq:Xs}
\end{equation}
The bandwidth constraint is therefore imposed through $s$, restricting the joint contexts $(s,x_{<i})$. We note that 
Assumption~\ref{ass:anchor} holds for anchoring either with the full spectrum, $s_j=\hat{x}_j$ or the magnitude spectrum, $s_j=|\hat{x}_j|$, (and possibly quantized).

\begin{lemma}[Complexity of band-limited anchoring]
\label{prop:main}
Under Assumption~\ref{ass:anchor}: The sets $\cX(s)$, $s\in\cS$, are disjoint subsets of $\cX_\cK$. Consequently,
    \[
    \sum_{s\in\cS}|\cX(s)|
    \le
    |\cK|d^k.
    \]
\end{lemma}

\begin{proof}
In this proof, we are upper bounding the number of distinct possible anchors. Because the anchor is a deterministic function of the signal, each signal
$x$ has exactly one anchor $A(x)$. Therefore, for any two distinct anchors
$s\neq s'$, $\cX(s)\cap\cX(s')=\emptyset$. Moreover, by Assumption~\ref{ass:anchor}, every anchor $s\in\cS$ has
spectral support contained in some $K\in\cK$. By Corollary~\ref{corr:counting-band},
\[
\sum_{s\in\cS}|\cX(s)|
=
\left|
\bigcup_{s\in\cS}\cX(s)
\right|
\le
|\cX_\cK|\le
|\cK|d^k.
\]

\end{proof}

We previously showed that representing the unanchored distribution \(p(x)\) requires
$
D_{\mathrm{un}} = d^n - 1
$
free parameters. We now introduce the deterministic anchor \(s=A(x)\) and count the free parameters required to represent the joint distribution \(p(x,s)=p(s)p(x\mid s)\). To keep the notation simple, we assume that every recording compatible with an anchor occurs with positive probability, i.e., $p(x,s)>0$ for every $s\in\cS$ and $x\in\cX(s)$. (Otherwise, replace $\cX(s)$ by the recordings that actually occur with $s$; all upper bounds below are unchanged.)

\begin{proposition}[Support-aware anchored count]
\label{lem:supp}
Under Assumption~\ref{ass:anchor},
\begin{equation}
  D^{\mathrm{supp}}_{\mathrm{anch}}
  =
  \underbrace{\bigl(|\cS|-1\bigr)}_{\text{anchor }p(s)}
  +
  \underbrace{\sum_{s\in\cS}
  \bigl(|\cX(s)|-1\bigr)}_{\text{generator }p(x\mid s)}
  =
  \Bigl|\bigcup_{s\in\cS}\cX(s)\Bigr|-1
  \le
  |\cK|d^k-1.
  \label{eq:split}
\end{equation}
\end{proposition}

\begin{proof}
For each $s\in\cS$, the conditional distribution $p(x\mid s)$ is supported
on $\cX(s)$ and therefore requires $|\cX(s)|-1$ free parameters. The anchor
distribution $p(s)$ requires $|\cS|-1$ free parameters. Hence,
\[
D^{\mathrm{supp}}_{\mathrm{anch}}
=
\bigl(|\cS|-1\bigr)
+
\sum_{s\in\cS}\bigl(|\cX(s)|-1\bigr)
=
\sum_{s\in\cS}|\cX(s)|-1.
\]
Because the anchor map $A$ is deterministic, the sets $\cX(s)$ are disjoint.
Therefore,
\[
\sum_{s\in\cS}|\cX(s)|
=
\Bigl|\bigcup_{s\in\cS}\cX(s)\Bigr|.
\]
Finally, Lemma~\ref{prop:main} gives
\[
\Bigl|\bigcup_{s\in\cS}\cX(s)\Bigr|
\le
|\cK|d^k,
\]
which proves the result.
\end{proof}

In the unanchored case, we naively allocate parameters to all \(d^n\) possible time-domain sequences as in the unstructured representation of \(p(x)\). In Equation~\ref{eq:split}, we show that the anchored factorization \(p(x,s)=p(s)p(x\mid s)\) restricts the parameterization to the set of feasible signals, which is substantially smaller due to the bandwidth constraint. Consequently, the anchored factorization yields a lower parameter complexity than an unstructured time-domain parameterization. Note that the parameter complexity bound is determined by the bandwidth constraints rather than the choice of anchor. The choice of anchor instead determines how this complexity is allocated between the anchor model and the conditional generator, as described in the following remark.

\begin{remark}[Effect of anchor informativeness]
The choice of anchor determines how the support-aware parameter complexity is divided between the anchor model \(p(s)\) and the conditional generator \(p(x\mid s)\).

\begin{itemize}
    \item At one extreme, a full complex-spectrum anchor determines \(x\) through the inverse DFT, so \(|\cX(s)|=1\) for every \(s\) and the conditional generator contributes no free parameters. All parameter complexity is therefore allocated to \(p(s)\).
    
    \item At the other extreme, a support-only anchor reveals only which frequency bins are active, leaving the corresponding spectral values unresolved and more variability to \(p(x\mid s)\).
    
    \item A magnitude-based spectral anchor lies between these extremes, resolving spectral magnitudes while leaving information such as phase to the conditional generator.
\end{itemize}
\end{remark}

\subsection{Specializing to a Contiguous Bandwidth Model}
We now express these bounds in terms of physical quantities to illustrate the change in parameter complexity between unanchored and anchored representations. Suppose a signal of bandwidth $B$ is sampled at rate $f_s$ for $n$ samples, i.e., for a duration $T=n/f_s$. Each DFT bin then spans $f_s/n$ Hz, so the signal occupies about $(B/f_s)\,n$ of the $n$ frequency bins. Since the signal may sit anywhere in the band, we allow every contiguous block of this width.
 
\begin{corollary}[Gain under a contiguous bandwidth constraint]
\label{cor:bw}
We specialize the allowed-support family
$\cK$ to all circular windows of $k$ contiguous DFT bins. We assume $d\ge2$ and $n\ge2$. 

For any anchor satisfying Assumption~\ref{ass:anchor},
the support-aware parameter count in~\eqref{eq:split} satisfies
\begin{equation}
    \log\frac{D_{\mathrm{un}}}{D^{\mathrm{supp}}_{\mathrm{anch}}}
    \ge
    \left(1-\frac{B}{f_s}\right)n\log d-\log(2nd)
    =
    (f_s-B)T\log d-\log(2nd).
    \label{eq:gain}
\end{equation}
\end{corollary}

\begin{proof}
We specialize Lemma~\ref{prop:main} and
Proposition~\ref{lem:supp} to the contiguous bandwidth family defined above.
Because the occupied band may begin at any of the $n$ DFT bins, $|\cK|=n$.  Moreover,
$
k
=
\left\lceil \frac{B}{f_s}n\right\rceil
\le
\frac{B}{f_s}n+1.
$

The unstructured time-domain parameterization is
\[
D_{\mathrm{un}}
=d^n-1
\ge
\frac{d^n}{2}.
\]

By Proposition~\ref{lem:supp}, restricting the parameterization to signals
consistent with the allowed bandwidth supports gives
\[
D^{\mathrm{supp}}_{\mathrm{anch}}
\le
|\cK|d^k.
\]
Substituting $|\cK|=n$ and the bound on $k$ yields
\[
D^{\mathrm{supp}}_{\mathrm{anch}}
\le
n d^{(B/f_s)n+1}.
\]

Therefore,
\[
\frac{D_{\mathrm{un}}}
     {D^{\mathrm{supp}}_{\mathrm{anch}}}
\ge
\frac{d^n/2}{n d^{(B/f_s)n+1}}
=
\frac{1}{2nd}
d^{(1-B/f_s)n}.
\]
Taking logarithms gives
\[
\log\frac{D_{\mathrm{un}}}
          {D^{\mathrm{supp}}_{\mathrm{anch}}}
\ge
\left(1-\frac{B}{f_s}\right)n\log d-\log(2nd).
\]
Finally, using $n=f_sT$ gives
\[
\left(1-\frac{B}{f_s}\right)n
=
(f_s-B)T,
\]
which yields the second form.
\end{proof}
 
\begin{remark}[Bandwidth Gain Regimes] The gain exponent in~\eqref{eq:gain} is the bandwidth that the anchor rules out, $f_s-B$, multiplied by the duration $T$. Equivalently, each of the $n-k\approx(f_s-B)\,T$ frequency bins that the anchor forces to zero divides the parameter count by a factor of $d$. Three regimes follow:
\begin{itemize}
  \item \emph{Narrowband} ($B\ll f_s$): the gain is nearly $d^{\,n}$. The extreme case is a DC signal, where the only allowed support is $K=\{0\}$, so $|\cK|=1$ and $k=1$. Then $\cX_{\{0\}}$ consists of the $d$ constant sequences $x=c\mathbf{1}$ with $c\in\cD$, and, given the anchor, $x_0$ fixes every later sample. Proposition~\ref{lem:supp} gives $D^{\mathrm{supp}}_{\mathrm{anch}}\le d-1$, which does not grow with $n$ at all, compared with $D_{\mathrm{un}}=d^{\,n}-1$.
  \item \emph{Fixed fraction} ($B/f_s=\beta<1$): the gain grows exponentially in the duration, as $d^{\,(1-\beta)n}$ up to polynomial factors.
  \item \emph{Full band} ($B=f_s$): here $k=n$, and the bound guarantees no gain. This is expected because, without a band limit, every sequence in $\cD^n$ is possible, and the anchor rules nothing out.
\end{itemize}

 \label{rem:2}
 \end{remark}
 
\section{Extended Related Works}
\label{ap:extended-work}

\paragraph{Multimodal IQ Embeddings.}
IQFM uses self-supervised pretraining on IQ streams and LoRA adaptation for downstream tasks including modulation classification and RF fingerprinting \citep{mashaal2026iqfm}. Other approaches exploit complementary wireless representations: IQFormer fuses IQ and time-frequency features for modulation recognition, while multimodal wireless foundation models learn across IQ, spectrograms, channel state information, and channel impulse responses \citep{shao2025iqformer,aboulfotouh2025multimodal}. However, these methods are optimized for downstream discrimination and do not natively support conditioning variables for controlled signal generation across modalities. In particular, we use CLIP-style contrastive alignment \citep{radford2021clip} between signal embeddings and semantic metadata alongside reconstruction, producing a decodable, label-aware latent space for conditional waveform generation as shown in Section~\ref{subsec:stage1}.

\paragraph{Conditional Time Series and Audio Generation}
Time Weaver incorporates heterogeneous metadata, including categorical, continuous, and time-varying variables, into a conditional diffusion model \citep{narasimhan2024time}. WaveStitch additionally uses partially observed signals at inference and reconciles overlapping, parallel-generated windows through a stitching objective \citep{shankar2025wavestitch}. RF-Diffusion modifies the signal diffusion process itself to operate across time and frequency \citep{chi2024rfdiffusion}. In adjacent waveform domains, DiffWave conditions audio synthesis on a mel spectrogram, whereas AudioLDM generates in a compressed mel-spectrogram latent space and decodes the result to audio \citep{kong2021diffwave,liu2023audioldm}. These approaches use frequency information as a generation condition or within the generative process. ALF instead uses a separately learned spectrogram latent to supervise an intermediate flow representation during training; the target spectrogram is not supplied at inference.

\paragraph{Intermediate Representation Supervision.}
Anchored Diffusion Language Models predict distributions over important tokens and use those predictions to guide subsequent denoising \citep{rout2025adlm}. In image generation, REPA instead aligns intermediate denoiser representations with features from a pretrained encoder of the clean image \citep{yu2025repa}. ALF uses a complementary spectral representation to expose structure that may be difficult to learn directly in the IQ generation space. Specifically, a training-time spectral anchor encourages intermediate flow representations to capture global time-frequency structure associated with the target waveform. The spectral target is then removed during generation, avoiding extra computational cost at inference.

\paragraph{Generative Channel Reconstruction.}
Another branch of work models the wireless channel $\mathbf{H}$, which characterizes propagation in the simplified received-signal relation $\mathbf{Y}=\mathbf{H}*\mathbf{X}+\mathbf{N}$. Here, $\mathbf{X}$ is the transmitted waveform and $\mathbf{N}$ is receiver noise; in channel estimation, the transmitted pilot portion of $\mathbf{X}$ is known. Generative estimators reconstruct $\mathbf{H}$ from such measurements using priors learned by Wasserstein GANs, score-based models, or diffusion models \citep{balevi2020highdimensionalchannelestimation,arvinte2022score,zhou2026generativediffusionmodelshigh}. For example, the GAN prior was evaluated on $\mathbf{H}\in\mathbb{C}^{16\times64}$, while a subsequent score-based method evaluated matrices up to $\mathbb{C}^{64\times256}$ \citep{balevi2020highdimensionalchannelestimation,arvinte2023mimo}. For channel synthesis, location-conditioned DDIM generates $\mathbb{C}^{4\times32}$ beamspace channels from three-dimensional user positions, while a physics-informed VAE generates channels through a parametric geometric model \citep{lee2025generatinghighdimensionaluserspecific,wagle2025physicsinformedgenerativeapproacheswireless}. These approaches support applications such as beam alignment, channel compression, and data augmentation. Their channel-focused outputs do not, however, represent the additional variability of the transmitted waveform and receiver noise in complete received recordings $\mathbf{Y}\in\mathbb{C}^{4\times N}$ covered in ALF.

\section{Dataset Design and Validation}
\label{app:dataset}
Here we discuss the detailed implementation and properties of the WaveScene dataset.
\subsection{Physical Channel Model and Modeling Challenges}

For a clean complex-baseband waveform $x[n]$, the received signal at antenna
$a$ is
\begin{equation}
y_a[n] =
\sum_{\ell=1}^{L} h_{a,\ell}x[n-d_{\ell}] + w_a[n],
\qquad a\in\{1,\ldots,4\},
\end{equation}
where $h_{a,\ell}$ is the antenna-dependent complex coefficient of path
$\ell$, $d_{\ell}$ is its shared discrete delay, and $w_a[n]$ is independent
complex receiver noise. The four receive elements observe the same transmitted
waveform, path delays, and propagation geometry, while their complex path
coefficients vary according to the array response. Their spatial dependence can
be represented by
\begin{equation}
\mathbf{R}_{h,\ell}
=
\mathbb{E}\!\left[
\mathbf{h}_{\ell}\mathbf{h}_{\ell}^{H}
\right],
\qquad
\mathbf{h}_{\ell}
=
[h_{1,\ell},h_{2,\ell},h_{3,\ell},h_{4,\ell}]^{T}.
\end{equation}
Consequently, a recording cannot be reduced to four independent SISO
observations: relative phase, amplitude, covariance, and shared multipath
structure carry information about the propagation environment.

All recordings use a single transmit antenna and a 4-element horizontal uniform
linear receive array with half-wavelength spacing. The element spacing is
rescaled for each carrier frequency. The array elements are vertically polarized
and omnidirectional. UMi, UMa, and RMa channels are generated with Sionna
v1.2.1~\citep{hoydis2023sionna}, using its implementation of the geometry,
large-scale parameters, visibility conditions, path loss, and shadow fading
specified by 3GPP TR~38.901~\citep{3gpp_tr38901}. InH channels are generated
with QuaDRiGa v2.8.1~\citep{jaeckel2014quadriga}, using the
\texttt{3GPP\_38.901\_Indoor\_Open\_Office} scenario and a corresponding
four-element horizontal receive array.

The corpus presents two complementary modeling challenges. First, recordings
contain as many as 32,768 complex samples, requiring preservation of waveform
timing, spectral structure, and protocol- or radar-specific behavior over long
contexts. Note that we use a static channel impulse response (CIR), so temporal
evolution of a mobile channel is out of scope. Second, the four synchronized
receive streams must remain spatially consistent: a model must preserve
cross-antenna phase, amplitude, and covariance relationships in addition to
producing realistic marginal signals.

\subsection{Waveform Parameterization and Standards Alignment}

The pre-channel waveform library used in WaveScene is constructed around documented physical-layer timing, bandwidth, and carrier constraints instead of arbitrary signal perturbations that may conflict with technical documentations shown in Table \ref{tab:waveform-standards}. It contains 1,066 parameterized waveform variants from 32 classes spanning generic modulations, communications protocols, and radar emissions. Toolbox-
backed WLAN, Bluetooth, LTE, and NR classes use standards-oriented physical-layer implementations. Other classes preserve their documented timing, bandwidth, pulse, hopping, and signal-morphology assumptions, although we do not uniformly claim receiver-decodable bitstreams for these classes.

\begin{table}[H]
\centering
\footnotesize
\setlength{\tabcolsep}{3pt}
\renewcommand{\arraystretch}{1.10}
\caption{Representative physical assumptions used in waveform generation.}
\label{tab:waveform-standards}

\begin{tabularx}{\linewidth}{@{}p{0.16\linewidth}Yp{0.27\linewidth}@{}}
\toprule
Family & Representative choices & Physical basis \\
\midrule
Linear modulation &
RRC shaping with roll-off $\alpha\in\{0.20,0.25,0.35\}$; varied symbol and
capture rates &
Digital modulation; DVB-S2 \\

WLAN &
20-MHz PPDU configurations, MCS values, preambles, and guard intervals &
IEEE 802.11 PHYs \\

Bluetooth &
LE 1M/2M GFSK and BR/EDR 2-Mb/s $\pi/4$-DQPSK PHYs &
Bluetooth Core radio PHY \\

LTE, NR, and NB-IoT &
OFDM numerologies, resource grids, and physical-channel configurations &
3GPP LTE and NR PHYs \\

Zigbee &
2-Mchip/s O-QPSK with 32-chip DSSS symbol morphology &
IEEE 802.15.4 PHY \\

LoRa &
CSS with SF6--SF9 and 125/250/500-kHz bandwidths &
Semtech LoRa PHY \\

ATSC &
8-VSB segment timing, pilot, synchronization, and vestigial-sideband morphology &
ATSC A/53 Part 2 \\

Radar &
Pulse, PRI, chirp, hop, dwell, and phase-code parameterizations &
Classical radar waveform texts \\
\bottomrule
\end{tabularx}

\vspace{2pt}
\parbox{\linewidth}{\scriptsize
\emph{Sources.}
Linear modulation: \citep{proakis2008digital,etsi302307_1};
WLAN: \citep{ieee80211_2020,ieee80211ax_2021};
Bluetooth: \citep{bluetooth_core54};
LTE, NR, and NB-IoT: \citep{3gpp36211,3gpp38211};
Zigbee: \citep{ieee802154_2020};
LoRa: \citep{semtech_sx1276};
ATSC: \citep{atsc_a53_part2};
radar: \citep{skolnik2008radar,richards2010radar,levanon2004radar}.}
\end{table}
For waveform classes with a configured bandwidth $B_{\mathrm{cfg}}$, the
release checker enforces
\begin{equation}
B_{99} \leq \rho B_{\mathrm{cfg}},
\qquad
B_{\mathrm{cfg}} \leq f_s,
\end{equation}
where $B_{99}$ is the measured 99\%-power occupied bandwidth, $\rho$ is a
class-specific tolerance, and $f_s$ is the complex-baseband sampling rate.
The occupied band is obtained from a Welch power spectral density estimate \citep{welch1967fft}
using the interval between the 0.5\% and 99.5\% cumulative-power quantiles.
This definition accommodates asymmetric spectra such as SSB-AM and ATSC
8-VSB.

Sampling rate, bandwidth, and recording length are selected jointly. WLAN is
generated at its native 20-MHz rate and, where applicable, at valid oversampled
rates; LoRa uses capture rates appropriate for its 125/250/500-kHz bandwidths;
narrowband analog signals use correspondingly lower rates; and radar sampling
rates scale with chirp bandwidth or hop span. For pulsed and frequency-hopping
radar, occupied bandwidth is measured over the active pulse or hop sequence
rather than over silent receive intervals. This prevents waveform duty cycle
from being mistaken for narrow spectral occupancy.

\paragraph{Controlled payload content.}
Payload bits are deliberately controlled rather than independently randomized
across recordings. A deterministic canonical bit stream is generated once, and
the modulation and protocol waveforms draw length-appropriate prefixes from
this stream before class-specific framing, coding, scrambling, and modulation.
This choice prevents arbitrary payload content from becoming an additional
unlabeled source of variation, allowing the corpus to emphasize waveform
morphology and propagation effects. Consequently, the dataset does not assess
generalization to independently varying payload content; extending the corpus
with independently seeded payloads is a natural direction for future work.

\subsection{Propagation-State Construction of Shared-Channel Set}

For the 7.52M shared-channel set, Propagation states are intentionally not balanced independently of scene
because their frequencies form part of the assumed deployment model \citep{3gpp_tr38901}. For urban
microcell (UMi) and urban macrocell (UMa), 80\% of generated links are
initialized as outdoor-to-indoor (O2I) and 20\% as outdoor. O2I forces an indoor
receiver, whereas outdoor links are assigned line-of-sight (LOS) or
non-line-of-sight (NLOS) status by Sionna's scene and distance-dependent
visibility model. We use the LOS/NLOS proportions from the
realized geometries instead of forcing a balanced distribution. Indoor hotspot (InH) links contain no O2I state; their LOS or NLOS condition is generated by
the QuaDRiGa Indoor Open Office model.

Rural macrocell (RMa) requires an explicit approximation. The documented
deployment mixture \citep{3gpp_tr38901} contains 50\% indoor and 50\% in-car users, but Sionna does
not provide a vehicle-cabin penetration model. Due to this limitation, we approximate by retaining the indoor half as O2I
and representing the in-car half as outdoor without vehicle penetration loss. Sionna then divides this outdoor component into LOS and NLOS links according to the realized geometry. These leads to the empirical distributions shown in Table \ref{tab:scene-state} and Figure \ref{fig:scene-distributions}:

\begin{table}[H]
\centering
\footnotesize
\caption{Empirical propagation-state composition of the 7.52M shared-channel set.}
\label{tab:scene-state}
\begin{tabular}{@{}lrrr@{}}
\toprule
Scene & LOS (\%) & NLOS (\%) & O2I (\%) \\
\midrule
UMi & 8.0  & 11.9 & 80.0 \\
UMa & 5.1  & 14.9 & 80.0 \\
RMa & 19.2 & 30.0 & 50.8 \\
InH & 68.1 & 31.9 & --   \\
\bottomrule
\end{tabular}
\end{table}

\subsection{Label Definitions and Normalization}
\label{ap:dataset-labels}

The 6 learning factors are derived from the clean waveform manifests and
generated channel realizations rather than estimated from the final noisy IQ.
Waveform class, scene type, and realized propagation state are categorical.
The continuous targets are spectral occupancy, mean per-antenna pre-combining
SNR, and RMS delay spread.

Spectral occupancy is defined as
\begin{equation}
o = \frac{B_{99}}{f_s},
\end{equation}
where $B_{99}$ is the measured 99\%-power occupied bandwidth, and $f_s$ is the native complex-baseband sampling rate. Unlike an absolute
bandwidth label, $o$ describes the fraction of the sampled spectrum occupied by
the waveform and is therefore comparable across capture rates.

The reported SNR is the mean per-antenna pre-combining SNR:
\begin{equation}
\mathrm{SNR}_{\mathrm{dB}}
=
P_{\mathrm{tx,dBm}}
+
G_{\mathrm{path,dB}}
-
P_{\mathrm{noise,dBm}}.
\end{equation}
Here, $P_{\mathrm{tx,dBm}}$ is transmit power, $G_{\mathrm{path,dB}}$ is
channel path gain obtained by averaging CIR energy across the four receive
antennas, and $P_{\mathrm{noise,dBm}}$ is the receiver noise power. Noise power
is calculated as
\begin{equation}
P_{\mathrm{noise,dBm}}
=
-174
+10\log_{10}(B_{99})
+10\log_{10}(T/290)
+\mathrm{NF},
\end{equation}
with temperature $T=290$ K and noise figure $\mathrm{NF}=9$ dB. The four
streams contain independent AWGN, and no coherent combining gain is included
in the learning target.

RMS delay spread is calculated from the power-weighted discrete CIR:
\begin{equation}
\tau_{\mathrm{rms}}
=
\sqrt{
\sum_{\ell=1}^{L}
w_{\ell}(\tau_{\ell}-\bar{\tau})^2
},
\qquad
\bar{\tau}
=
\sum_{\ell=1}^{L}w_{\ell}\tau_{\ell},
\qquad
w_{\ell}
=
\frac{|h_{\ell}|^2}
{\sum_{j=1}^{L}|h_j|^2}.
\end{equation}
In this expression, $\tau_{\ell}$ is the delay of tap $\ell$, $h_{\ell}$ is
its complex coefficient after averaging tap energy across receive antennas,
and $w_{\ell}$ is its normalized power contribution.

\begin{figure}[t]
\centering
\begin{subfigure}[t]{0.49\linewidth}
    \centering
    \includegraphics[width=\linewidth]{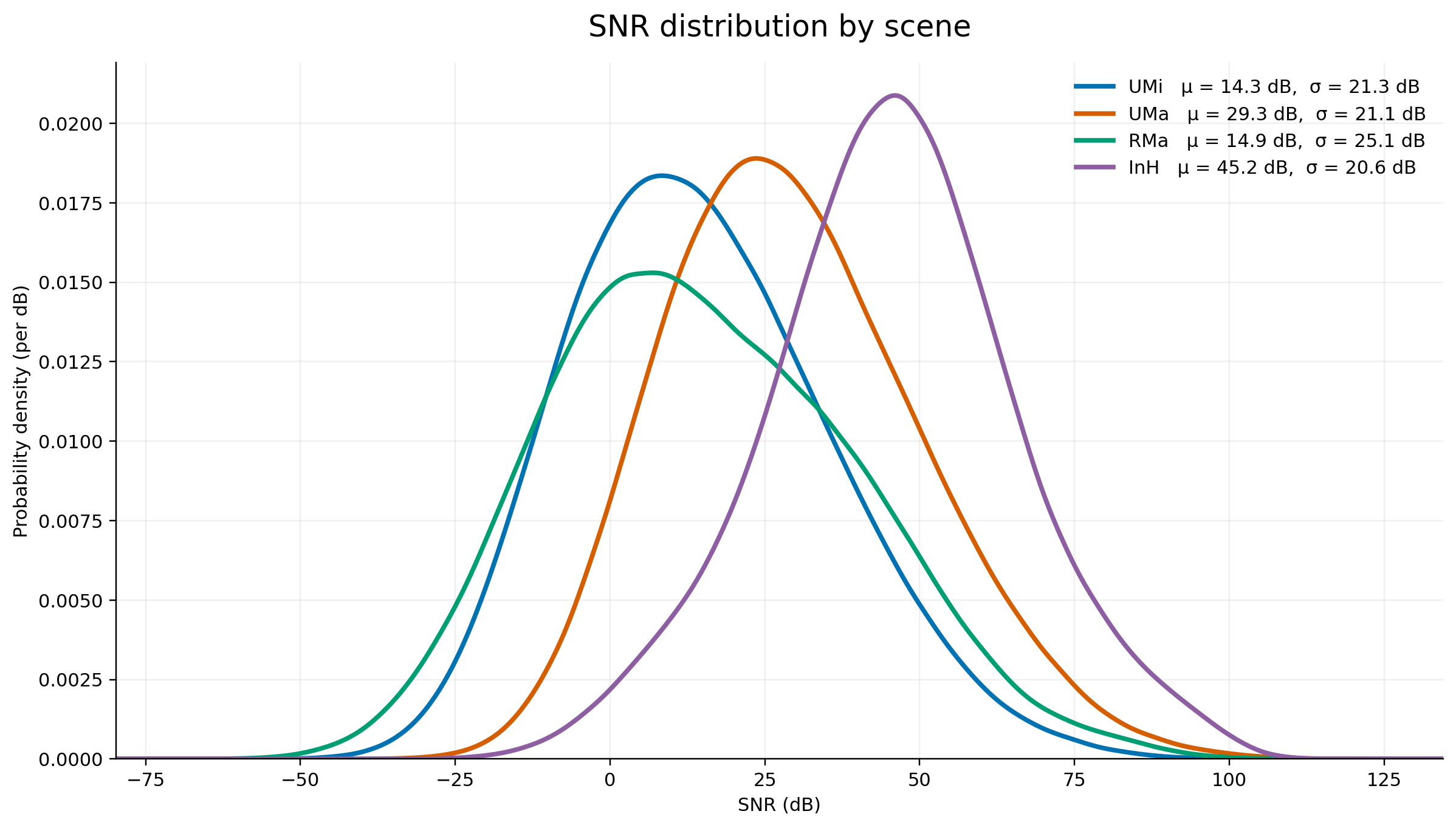}
    \caption{SNR}
    \label{fig:snr-scene}
\end{subfigure}\hfill
\begin{subfigure}[t]{0.49\linewidth}
    \centering
    \includegraphics[width=\linewidth]{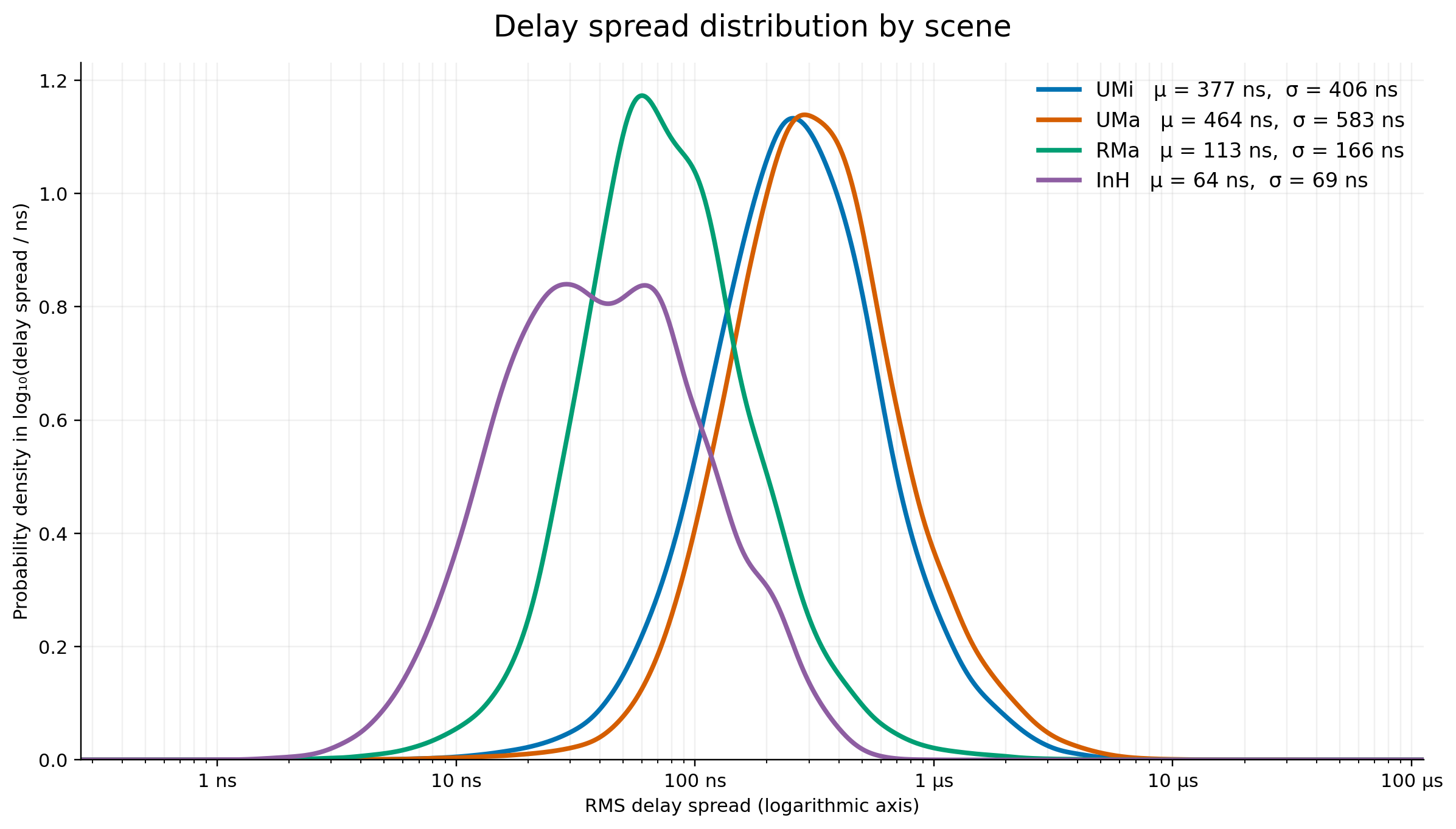}
    \caption{RMS delay spread}
    \label{fig:delay-spread-scene}
\end{subfigure}
\caption{Empirical SNR and RMS delay-spread distributions by propagation scene
of the shared-channel set.}
\label{fig:scene-distributions}
\end{figure}

The continuous targets are normalized as
\begin{align}
\widetilde{\mathrm{SNR}}
&=
\operatorname{clip}
\left(
\frac{\mathrm{SNR}_{\mathrm{dB}}-25}{25},
-3,3
\right),\\
\tilde{o}
&=
\operatorname{clip}
\left(
\frac{\log_{10}(o)+1}{0.75},
-3,3
\right),\\
\tilde{\tau}
&=
\operatorname{clip}
\left(
\frac{\log_{10}(\tau_{\mathrm{rms}}/\mathrm{ns})-2.4}{0.8},
-3,3
\right).
\end{align}
Clipping is rare in the 7.52M-recording corpus. For SNR, 0.018\% of recordings
fall below the lower threshold and 0.102\% exceed the upper threshold. No
occupancy values are clipped. For delay spread, 0.0014\% fall below the lower
threshold and fewer than 0.0001\% exceed the upper threshold.

Carrier frequency, transmit power, horizontal distance, path gain, channel
identity, and diagnostic channel quantities are retained as metadata despite
not being used as primary learning factors. These fields remain available for
auditing and stratified analysis when needed.

\subsection{Held-Out Channel Benchmark}

The held-out-channel benchmark contains 300,000 recordings whose channel
identities do not occur in the development corpus. Channel identity jointly
includes scene, topology, stochastic channel variation, propagation state, and
carrier frequency. Holding out this complete identity prevents the benchmark
from evaluating a previously observed CIR under only a different waveform or
noise realization.

The benchmark contains exactly 9,375 recordings from each of the 32 waveform
classes. It is also well distributed across propagation scenes and carrier
frequencies, while retaining the physically motivated state composition
described above. Its purpose is to measure generalization to unseen propagation realizations, but it does not provide a deliberately shifted waveform or label distribution. The deterministic list of selected recording and channel identifiers is distributed with the benchmark.

\subsection{Reproducibility and Limitations}

Dataset construction is controlled by versioned waveform manifests and
generation plans. Each waveform manifest specifies its physical parameters,
sampling rate, stored length, active region, padding, and measured occupied
bandwidth. A master seed is deterministically expanded into topology,
state-selection, waveform-selection, channel, and noise seeds using
SHA-256-derived identifiers. Each channel cache stores a fingerprint tied to
its waveform-slot plan, generation configuration, geometry rules,
channel-engine version, array configuration, and cache schema. The pipeline
rejects caches whose fingerprints do not match the active configuration.

The release checker validates all 1,066 waveform manifests and generated
pre-channel waveforms. It verifies identifier agreement, stored lengths, normalization,
padding constraints, occupied-bandwidth conditions, and SHA-256 waveform
uniqueness. We will release the dataset-generation code, configuration
manifests, seed policies, label definitions, and validation scripts upon paper
acceptance.

Although WaveScene covers diverse waveform families and propagation conditions,
it has several limitations. Channels remain static within each recording, so the current release does not include mobility, Doppler evolution,
or handover. We align waveforms to a normalized temporal origin; asynchronous
detection and crop localization therefore fall outside the benchmark. The
receiver model includes thermal noise and receiver noise figure, but excludes
carrier-frequency offset, sampling-clock offset, IQ imbalance, nonlinear
hardware distortion, co-channel interference, and external atmospheric or
man-made noise. For RMa, the nominal in-car condition uses a
Sionna-compatible approximation: half of these links are modeled as O2I and
half as outdoor links, with Sionna sampling LOS or NLOS for the outdoor portion.
All selected carriers lie within the 0.5--100~GHz scope of TR~38.901. Communications payloads use a deterministic reference source, as
described above, so the benchmark does not evaluate generalization to
independently varying message content. Finally, several manually implemented
protocol classes preserve documented RF timing and waveform morphology without
providing complete receiver-decodable protocol stacks. The radar classes model
emitted or intercepted waveforms and exclude two-way target propagation,
range-dependent return delay, radar cross section, target Doppler, and clutter.

\section{Encoder Architecture and Evaluation Details}
\label{app:encoder_details}

\subsection{Encoder Architecture}
\label{app:encoder_architecture}
This section provides implementation details ommited from Section~\ref{subsec:stage1}.
Stage~1 of ALF uses $N_{\mathrm{patch}}=1024$ and a 768-dimensional latent
for each patch. The raw-IQ and ButterflyNet-inspired spectral branches are
fused by a learned gate; the spectral branch is initialized with a radix-2
Fourier inductive bias. A 4-layer, 8-head Transformer mixes information across
patches to form per-patch embedding
$\mathbf{Z}_{\mathrm{IQ}}\in\mathbb{R}^{P\times768}$, and
a \texttt{[CLS]} token produces per-recording embedding
$\mathbf{Z}_{\mathrm{IQ\text{-}CLS}}\in\mathbb{R}^{768}$. A separate 4-layer, 8-head decoder Transformer first mixes the complete latent
sequence across patches. Each resulting token is then decoded by a dual-path
patch decoder: a transposed-convolutional raw-IQ path provides the primary
waveform reconstruction, while an inverse-Butterfly path provides a gated
spectral correction. A separate linear head predicts the patch log-RMS
amplitude, so received level does not dominate waveform-shape reconstruction.

\paragraph{Reconstruction objective.}
Let $\widehat{\mathbf X}_{\mathrm{IQ}}$ denote the decoded recording. The
reconstruction loss is
\begin{equation}
\begin{aligned}
\mathcal{L}_{\mathrm{rec}}
={}&
\mathcal{L}_{\mathrm{IQ}}
+\mathcal{L}_{\mathrm{corr}}
+0.5\mathcal{L}_{\mathrm{FFT}}
+\mathcal{L}_{\mathrm{STFT}} \\
&+0.5\mathcal{L}_{\mathrm{env}}
+0.25\mathcal{L}_{\mathrm{cov}}
+0.25\mathcal{L}_{\mathrm{level}} .
\end{aligned}
\label{eq:encoder_reconstruction}
\end{equation}
Here, $\mathcal{L}_{\mathrm{IQ}}$ is a sample-wise $\ell_1$ loss on the
real--imaginary waveform components, and
$\mathcal{L}_{\mathrm{corr}}$ is one minus the cosine similarity of the
flattened reconstructed and target IQ waveforms. The latter is scale-invariant
but phase-sensitive, preventing a low-amplitude output from always achieving a low loss.

The remaining terms compare complementary signal properties:
$\mathcal{L}_{\mathrm{FFT}}$ is an $\ell_1$ loss between full-patch
$\log(1+\lvert\mathrm{FFT}(\cdot)\rvert^2)$ spectra;
$\mathcal{L}_{\mathrm{STFT}}$ compares log-power spectrograms using a
256-point STFT with hop size 64;
and $\mathcal{L}_{\mathrm{env}}$ compares magnitude envelopes pooled into
32 temporal windows. To preserve SIMO spatial structure,
$\mathcal{L}_{\mathrm{cov}}$ is the mean absolute discrepancy between
trace-normalized complex antenna covariance matrices. Finally,
$\mathcal{L}_{\mathrm{level}}$ is an $\ell_1$ loss on the scalar amplitude
predicted by the decoder's level path.

\paragraph{Soft-target contrastive alignment.}
For a batch of $B$ recordings, let
$\mathbf z_i$ and $\mathbf c_j$ be the normalized IQ and label embeddings
defined in Section~\ref{subsec:stage1}, and let
\begin{equation}
\ell_{ij}=\frac{\mathbf z_i^\top\mathbf c_j}{\tau},
\qquad \tau=0.2.
\end{equation}
Rather than using a one-hot pairing target, we construct a metadata affinity
matrix. Let $w_i$, $s_i$, and $r_i$ denote waveform class, scene type, and
propagation state; let $a_i$, $\rho_i$, and $d_i$ denote spectral occupancy,
mean per-antenna pre-combining SNR in dB, and RMS delay spread in ns, respectively. We use
\begin{equation}
\begin{aligned}
m_{ij}
={}&
\mathbbm{1}[w_i=w_j]\,
\mathbbm{1}[s_i=s_j]\,
\mathbbm{1}[r_i=r_j] \\
&\times
\exp\left[
-\frac{1}{2}
\left(
\left(\frac{\rho_i-\rho_j}{3}\right)^2
+
\left(\frac{\log_{10}a_i-\log_{10}a_j}{0.15}\right)^2
+
\left(\frac{\log_{10}d_i-\log_{10}d_j}{0.25}\right)^2
\right)
\right].
\end{aligned}
\label{eq:softclip_affinity}
\end{equation}
The normalized row targets are
$q_{ij}=m_{ij}/\sum_{k=1}^{B}m_{ik}$. The symmetric SoftCLIP loss is
\begin{equation}
\mathcal{L}_{\mathrm{CLIP}}
=
-\frac{1}{2B}
\sum_{i=1}^{B}\sum_{j=1}^{B}
\left[
q_{ij}\log\operatorname{softmax}_{j}(\ell_{ij})
+
q_{ji}\log\operatorname{softmax}_{j}(\ell_{ji})
\right].
\label{eq:softclip_loss}
\end{equation}
Consequently, recordings must agree on the three discrete label fields to be
positive pairs, while nearby occupancy, SNR, and delay-spread values receive
larger relative weight. This avoids treating semantically similar recordings as
equally negative merely because they have varied labels.

\paragraph{Structural objective.}
The last loss applied to
$\mathbf Z_{\mathrm{IQ\text{-}CLS}}$ is the structure loss, $\mathcal{L}_{\mathrm{struct}}$.It encourages the representation to
retain each metadata factor, distribute factors across distinct embedding
groups, and preserve nonzero variation across embedding dimensions.

Specifically, the 768-dimensional embedding is partitioned into $G=8$ groups
of 96 dimensions. For each of the six metadata fields
\[
\mathcal{F}=
\{\text{waveform},\text{occupancy},\text{SNR},\text{scene},
\text{propagation state},\text{delay spread}\},
\]
a lightweight gated MLP probe predicts the corresponding target. Categorical
factors use cross-entropy and continuous factors use mean-squared error after
the training-set normalization. Denoting the sum of these six losses by
$\mathcal{L}_{\mathrm{probe}}$, the complete structural loss is
\begin{equation}
\begin{aligned}
\mathcal{L}_{\mathrm{struct}}
={}&
\alpha_{\mathrm{probe}}\mathcal{L}_{\mathrm{probe}}
+\alpha_{\mathrm{gate}}\mathcal{L}_{\mathrm{gate}}
+\alpha_{\mathrm{group}}\mathcal{L}_{\mathrm{group}} \\
&+
\alpha_{\mathrm{dim}}\mathcal{L}_{\mathrm{dim}}
+\alpha_{\mathrm{var}}\mathcal{L}_{\mathrm{var}}.
\end{aligned}
\label{eq:struct_loss}
\end{equation}
Each factor has a softmax-normalized gate over the eight groups.
$\mathcal{L}_{\mathrm{gate}}$ penalizes overlap between the average gates for
the waveform--occupancy factors and those for the propagation-related factors
(SNR, scene, propagation state, and delay spread). This encourages, but does
not force, the two types of information to use distinct portions of the
embedding.

For a minibatch embedding matrix $\mathbf Z\in\mathbb{R}^{B\times768}$,
$\mathcal{L}_{\mathrm{group}}$ is the mean squared cross-covariance between
distinct 96-dimensional groups, while
\begin{equation}
\mathcal{L}_{\mathrm{dim}}
=
\frac{1}{768}
\sum_{u\neq v}
\operatorname{Cov}(\mathbf Z)_{uv}^{2},
\qquad
\mathcal{L}_{\mathrm{var}}
=
\frac{1}{768}\sum_{u=1}^{768}
\left[\max\!\left(0,\gamma-\operatorname{Std}(\mathbf Z_{:u})\right)\right]^2,
\label{eq:decorrelation_variance}
\end{equation}
with $\gamma=0.02$. The group and dimension covariance terms discourage
redundant encoding, whereas the variance-floor term directly guards against
dimensional collapse. The gated factor probe is used only as a training-time
auxiliary and is discarded after Stage~1; it is distinct from the post-hoc MLP
probes used for evaluation.

\paragraph{Optimization schedule.}
We jointly optimize
$\mathcal{L}_{\mathrm{S1}}=
\mathcal{L}_{\mathrm{rec}}+
\lambda_{\mathrm{CLIP}}(e)\mathcal{L}_{\mathrm{CLIP}}+
\lambda_{\mathrm{struct}}(e)\mathcal{L}_{\mathrm{struct}}$
for 20 epochs using AdamW with learning rate $2\times10^{-4}$, weight decay
$0.05$, and gradient clipping at $1.0$. Reconstruction and the initial
structural terms are active from the outset. The CLIP coefficient is scheduled
from $0$ to $0.25$ and then $0.5$, while the structural coefficients are
increased from their initially weak settings over later epochs. This staging
allows the decoder to first establish a stable reconstructive latent before
strong alignment and decorrelation pressures are applied.

\paragraph{Training and probing protocol.}
Stage~1 of ALF is trained on the deterministic training partition of the
\datasetname{} 5.6M development set, comprising approximately $5.3$M recording--label pairs for training; the remaining $5\%$ of that development set ($0.3$M pairs) is used for validation. The seen and unseen-channel test sets later discussed are never used during training. We train for 20 epochs with AdamW,
a learning rate of $2\times10^{-4}$, weight decay of $0.05$, gradient clipping at $1.0$, and a global batch size of 288 across 4 GPUs. All Stage-1 components, including the IQ encoder, decoder, and label encoder, are subsequently frozen.

The representation probes reported in this work are post-hoc evaluations and
are distinct from the training-time factor probes in
$\mathcal{L}_{\mathrm{struct}}$. For each frozen encoder and prediction task,
we train an identical one-hidden-layer MLP on a fixed 5M subset of the
5.3M training partition, select the model using a disjoint 256K validation
set, and repeat training with 5 fixed random seeds. Each resulting probe is
evaluated unchanged on both a 300K-recording subset of the complete 1.9M seen-channel test set
and the 300K channel-held-out test set. The seen-channel subset is matched to
the channel-held-out set across waveform class, scene type, and propagation
state, isolating the effect of channel realization while preserving the
\datasetname{} label space.

\subsection{Adaptation of External RF Encoders}
\label{app:baseline_adaptation}

We compare ALF with WavesFM \citep{aboulfotouh2025wavesfm},
IQFM \citep{mashaal2026iqfm}, and IQFormer \citep{shao2025iqformer}. Our \datasetname{} recordings are longer and span a broader combination of
waveform families, channel environments, and propagation states than the data
used by the external encoders in their original publications. We therefore regard these as best reproducible adaptations of the available baseline
implementations: they preserve each method's native encoder and validated input recipe while matching the \datasetname{} training recordings and frozen-probe protocol wherever retraining is possible. We do not add ALF's cross-patch mixing, label encoder, reconstruction objective, or structural constraints to any external baseline.

\paragraph{WavesFM.}
We use the released pretrained checkpoint, retaining its IQ tokenizer and
8-block Transformer. The encoder receives four-antenna,
$N_{\mathrm{view}}=256$ IQ views with per-recording maximum-absolute
normalization, and its 256-dimensional \texttt{[CLS]} feature is used for
probing. For deterministic evaluation, we use the center view from the
fixed window bank. No \datasetname{} waveform, scene, or channel label is used
to train or adapt this encoder.

\paragraph{IQFM.}
We retain IQFM's native multi-antenna SSL-Joint ShuffleNetV2-x0.5 encoder and
its self-supervised objective. Each recording contributes one
$N_{\mathrm{view}}=256$ view per training visit, selected from up to 8 evenly
spaced views; the two contrastive augmentations are constructed from that
selected view. Evaluation uses the deterministic center view. Thus, the
training budget counts source recordings rather than treating multiple views from
one recording as independent recordings.

\paragraph{IQFormer.}
We retain IQFormer's published single-antenna IQ--STFT encoder and remove its
classification layer to expose the 64-dimensional pooled feature. Its
time--frequency preprocessing uses a Blackman-window STFT with window length
31, hop 1, FFT size 128, and the first 32 frequency bins. IQFormer is trained
with waveform-class supervision over the 32 \datasetname{} waveform classes before probe training; therefore, its waveform result should be interpreted with the favorable supervised pretraining setting explicitly in mind. We do not use its
unpublished auxiliary power feature.

\paragraph{ALF encoder.}
Also refered to as ALF's stage 1, ALF encoder processes the complete native recording as a sequence of
1,024-sample patches and uses the frozen
$\mathbf{Z}_{\mathrm{IQ\text{-}CLS}}$ embedding from Stage~1 training (\ref{subsec:stage1}). Its encoder
therefore has access to recording-level context unavailable to external baselines due to their different $N_{\mathrm{view}}=256$.

\paragraph{Training comparability.}
Stage~1 ALF training uses the established 5.36M-recording training partition.
The retrained IQFM and IQFormer encoders use a fixed 5M-recording training
subset with a disjoint 256K-recording validation subset. All downstream
frozen-probe and label-efficiency evaluations use the same fixed 5M-recording
fit subset, disjoint 256K-recording validation subset, and matched test sets
for every encoder. Runtime efficiency is measured at the same input
lengths for every model. Baseline-specific augmentation, supervision, and
optimization details are retained when they are intrinsic to the published
method. WavesFM instead uses its released pretrained checkpoint. The external
baselines' native 256-sample views are a limitation of this comparison; we
preserve them to avoid silently replacing their validated training regimes with
a new long-context architecture.

\paragraph{Encoder size and long-recording scaling.}
Tables~\ref{tab:encoder_size} and~\ref{tab:encoder_runtime} compare the
encoder-only model size and batch-1 inference cost. For ALF, we exclude all
reconstruction-only and generation-only components, including the decoder-side
cross-patch mixer, dual-path waveform decoder, amplitude head, label encoder,
spectral autoencoder, and latent flow. ALF therefore uses a 60.31M-parameter
encoder to produce its 768-dimensional recording-level embedding.

We measure FP32 CUDA-event latency and peak forward activation memory on an
NVIDIA H100 NVL using inputs of
$N_{\mathrm{samples}}\in\{2048,8192,32768\}$; data loading and downstream
MLP probes are excluded. ALF increases from 9.6~ms to 12.6~ms across this
length range, whereas IQFormer increases from 6.4~ms to 89.8~ms. The external
encoders are evaluated through their implemented variable-length forwards for
this scaling diagnostic, although their reported frozen-probe results use their
native 256-sample views. Consequently, Table~\ref{tab:encoder_runtime}
characterizes long-recording computational scaling rather than an equal-context representation comparison.

\begin{table}[H]
\centering
\small
\caption{Encoder-only model size.}
\label{tab:encoder_size}
\begin{tabular}{lrrr}
\toprule
Encoder & Parameters (M) & FP16 weights (MiB) & Embedding dim. \\
\midrule
\textbf{ALF} & \textbf{60.31} & \textbf{115.0} & \textbf{768} \\
WavesFM & 6.40 & 12.2 & 256 \\
IQFM & 0.34 & 0.7 & 1024 \\
IQFormer & 0.35 & 0.7 & 64 \\
\bottomrule
\end{tabular}
\end{table}

\begin{table}[H]
\centering
\scriptsize
\caption{Encoder-only batch-1 FP32 inference cost on one GPU. Each model uses its native antenna
convention; ALF, WavesFM, and IQFM receive 4 antennas, while IQFormer receives
antenna 0. Values exclude resident model weights, data loading and downstream heads.}
\label{tab:encoder_runtime}
\resizebox{\linewidth}{!}{%
\begin{tabular}{lrrrrrr}
\toprule
Encoder
& \multicolumn{2}{c}{$N_{\mathrm{samples}}=2{,}048$}
& \multicolumn{2}{c}{$N_{\mathrm{samples}}=8{,}192$}
& \multicolumn{2}{c}{$N_{\mathrm{samples}}=32{,}768$} \\
& Latency (ms) $\downarrow$ & Peak activation (MiB) $\downarrow$
& Latency (ms) $\downarrow$ & Peak activation (MiB) $\downarrow$
& Latency (ms) $\downarrow$ & Peak activation (MiB) $\downarrow$ \\
\midrule
\textbf{ALF} & \textbf{9.6} & \textbf{10.6} & \textbf{11.2} & \textbf{42.5} & \textbf{12.6} & \textbf{170.0} \\
WavesFM & 1.9 & 3.4 & 1.9 & 3.5 & 2.0 & 4.3 \\
IQFM & 2.8 & 1.1 & 2.8 & 4.4 & 2.7 & 9.8 \\
IQFormer & 6.4 & 41.1 & 23.9 & 68.3 & 89.8 & 173.4 \\
\bottomrule
\end{tabular}%
}
\end{table}

\subsection{Frozen-Encoder Probing Protocol}
\label{app:encoder_probes}

All encoders are frozen before downstream probing. For every encoder and
target, we train the same one-hidden-layer MLP,
\begin{equation}
\mathbf h_i=
\operatorname{GELU}(\mathbf W_1\overline{\mathbf z}_i+\mathbf b_1),
\qquad
\widehat{\mathbf y}_i=
\mathbf W_2\operatorname{Dropout}_{0.1}(\mathbf h_i)+\mathbf b_2,
\label{eq:probe_mlp}
\end{equation}
where $\overline{\mathbf z}_i$ is the feature standardized using statistics
computed only from the probe-training partition, and the hidden width is 512.

Each probe is trained on 5M recordings from the fixed training subset and selected
using the disjoint 256K-recording validation set. We use AdamW with learning rate
$10^{-3}$, weight decay $10^{-2}$, batch size 16,384, and 5 epochs. For every
encoder--task pair, we train 5 matched seeds
$(13,29,42,71,101)$ and retain the checkpoint with the best validation
macro-F1 for classification or lowest validation RMSE for regression. The same
saved probe is evaluated without adaptation on both the matched 300K-recording
seen-channel test subset and the 300K channel-held-out test set.

Waveform recognition covers all 32 classes and is summarized by macro-F1.
Mean per-antenna pre-combining SNR is evaluated as continuous regression in dB,
while RMS delay spread is regressed in $\log_{10}(\mathrm{ns})$. Accuracy,
macro recall, AUROC, MAE, and rank correlation are retained as secondary metrics.
Table~\ref{tab:frozen_encoder_probes} reports mean performance and two-sided
95\% Student-$t$ confidence intervals across the 5 matched seeds.

\begin{table}[H]
\centering
\scriptsize
\caption{Frozen-encoder probing results. Waveform scores are macro-F1 (\%);
SNR is regressed in dB; and delay spread is regressed in
$\log_{10}(\mathrm{ns})$. Values are mean $\pm$ two-sided 95\%
Student-$t$ confidence intervals over 5 matched probe seeds.}
\label{tab:frozen_encoder_probes}
\setlength{\tabcolsep}{3.2pt}
\resizebox{\linewidth}{!}{%
\begin{tabular}{lccc|ccc}
\toprule
& \multicolumn{3}{c}{Seen channels}
& \multicolumn{3}{c}{Channel-held-out channels} \\
Encoder
& Waveform F1 (\%) $\uparrow$
& SNR RMSE $\downarrow$
& Delay RMSE $\downarrow$
& Waveform F1 (\%) $\uparrow$
& SNR RMSE $\downarrow$
& Delay RMSE $\downarrow$ \\
\midrule
WavesFM
& $50.35\pm0.08$ & $15.937\pm0.024$ & $0.4822\pm0.0007$
& $51.08\pm0.08$ & $15.348\pm0.016$ & $0.4784\pm0.0009$ \\
IQFormer
& $66.18\pm0.09$ & $9.416\pm0.023$ & $0.4482\pm0.0003$
& $67.11\pm0.13$ & $9.062\pm0.014$ & $0.4419\pm0.0005$ \\
IQFM
& $60.00\pm0.55$ & $14.173\pm0.051$ & $0.4294\pm0.0013$
& $60.64\pm0.51$ & $13.608\pm0.036$ & $0.4243\pm0.0012$ \\
\textbf{ALF}
& $\mathbf{99.21\pm0.01}$ & $\mathbf{1.259\pm0.004}$
& $\mathbf{0.1752\pm0.0001}$
& $\mathbf{99.37\pm0.01}$ & $\mathbf{1.201\pm0.004}$
& $\mathbf{0.1678\pm0.0001}$ \\
\bottomrule
\end{tabular}%
}
\end{table}

\subsection{Label efficiency}
\label{app:probe_scaling}

To measure how readily task information can be extracted from each frozen
representation, we repeat waveform probing with nested deterministic prefixes
of \(1\)K, \(10\)K, \(100\)K, \(1\)M, and \(5\)M recordings from the same ordered
probe-training split. The validation set, probe architecture, optimizer,
training duration, five seeds, and 300K test sets remain unchanged. Even the
1K prefix contains all 32 waveform classes, with 22--40 recordings per class.

\begin{figure*}[t]
\centering
\includegraphics[width=0.98\textwidth]{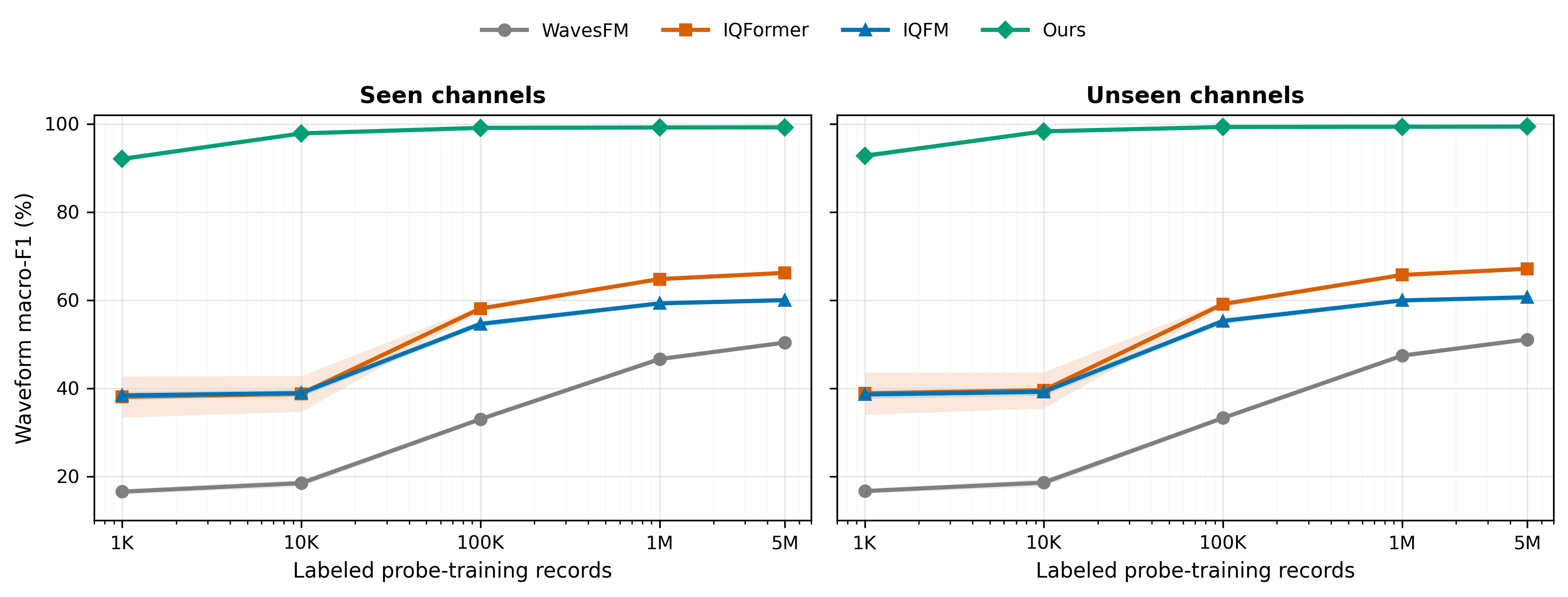}
\caption{Waveform classification as the number of labeled probe-training
recordings is varied. Curves show mean macro-F1 and shaded 95\% confidence
intervals across five matched seeds. Our frozen representation reaches
\(92.0\%\) and \(92.7\%\) macro-F1 with only 1K labeled recordings on seen and
unseen channels, respectively, and exceeds \(99\%\) by 100K recordings.}
\label{fig:probe_sample_efficiency}
\end{figure*}

The nearly coincident seen and channel-held-out curves show that ALF's label efficiency extends beyond the channel realizations used for training. Because \datasetname{} training data already span substantial channel variability, the channel-held-out set evaluates controlled channel generalization, not a severe distribution shift. Our encoder reaches \(97.8\%\) seen-channel and
\(98.3\%\) unseen-channel macro-F1 with 10K recordings and largely saturates by
100K. In comparison, the external encoders continue improving between 100K and
5M recordings but plateau at substantially lower performance. IQFormer's wider
confidence interval at 1K and 10K additionally indicates greater sensitivity
to the small probe-training set.

\begin{figure*}[t]
\centering
\includegraphics[width=0.72\textwidth]{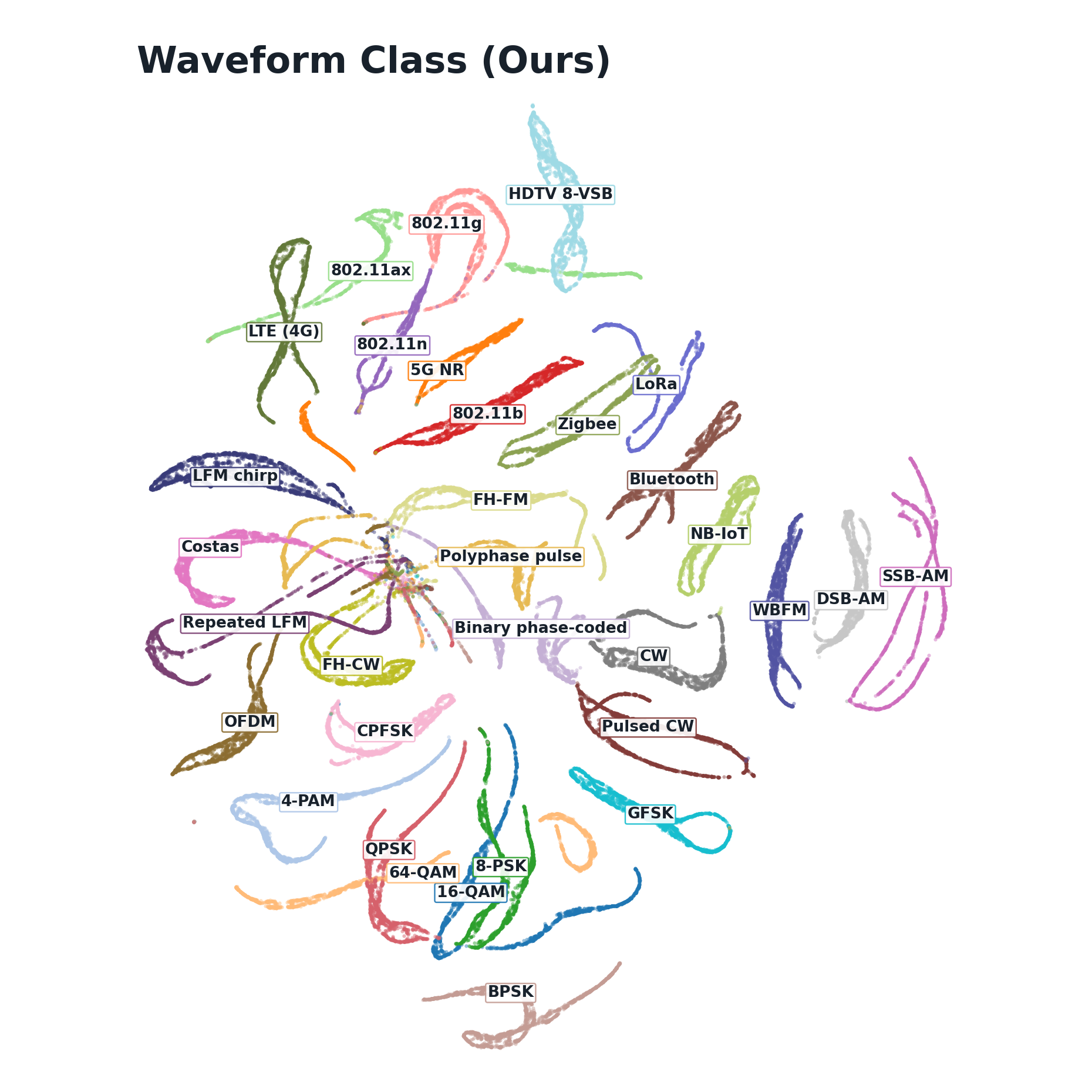}
\caption{UMAP of the ALF encoder's signal embedding over waveform class.}
\label{fig:umap_waveform}
\end{figure*}
\pagebreak

\newpage
\subsection{Embedding geometry}
\label{app:encoder_umap}

\begin{minipage}{0.96\textwidth}
\centering
\captionsetup{type=figure}

\begin{subfigure}[t]{0.49\textwidth}
    \centering
    \includegraphics[width=\linewidth]{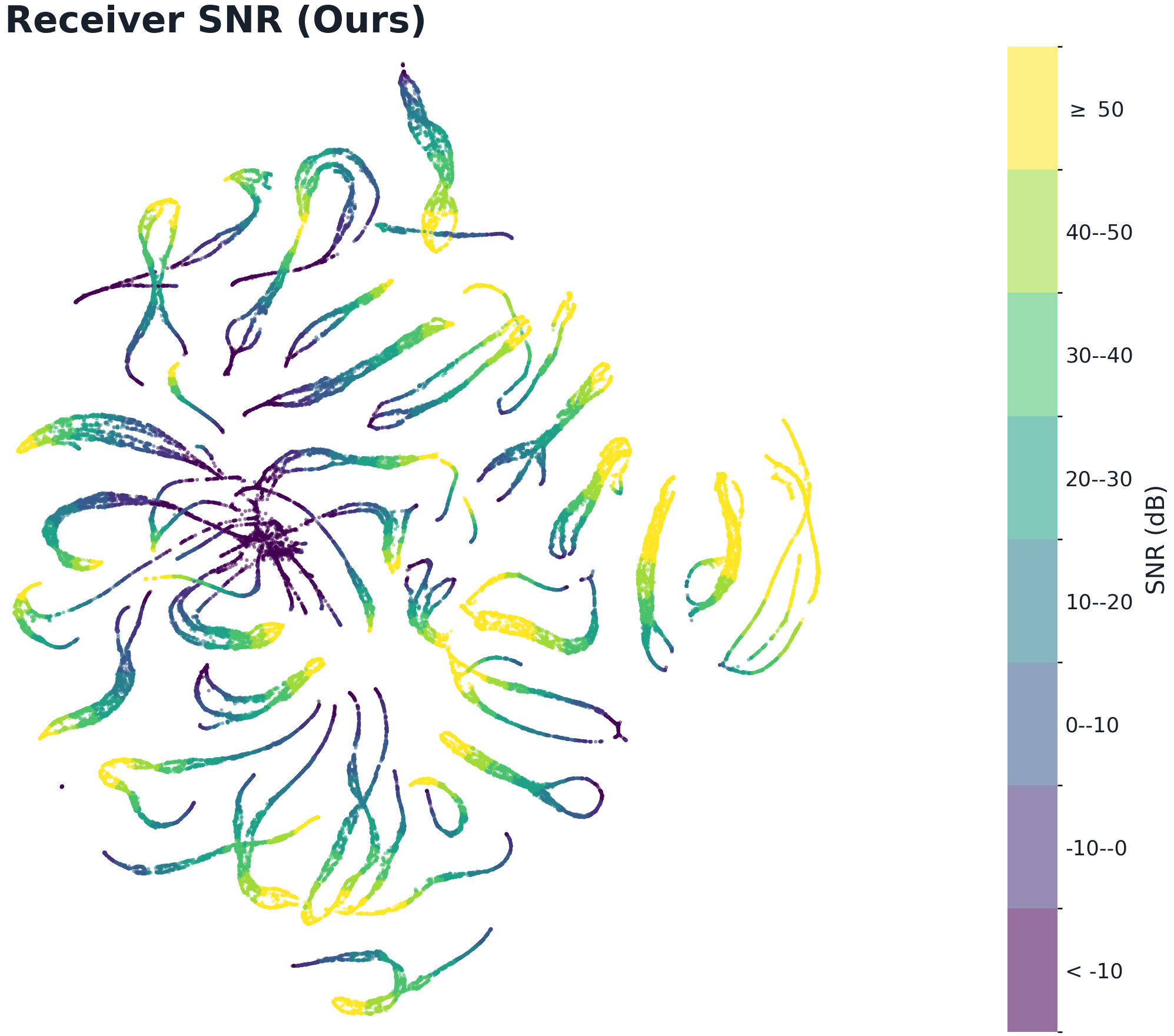}
    \caption{Receiver SNR.}
\end{subfigure}
\hfill
\begin{subfigure}[t]{0.49\textwidth}
    \centering
    \includegraphics[width=\linewidth]{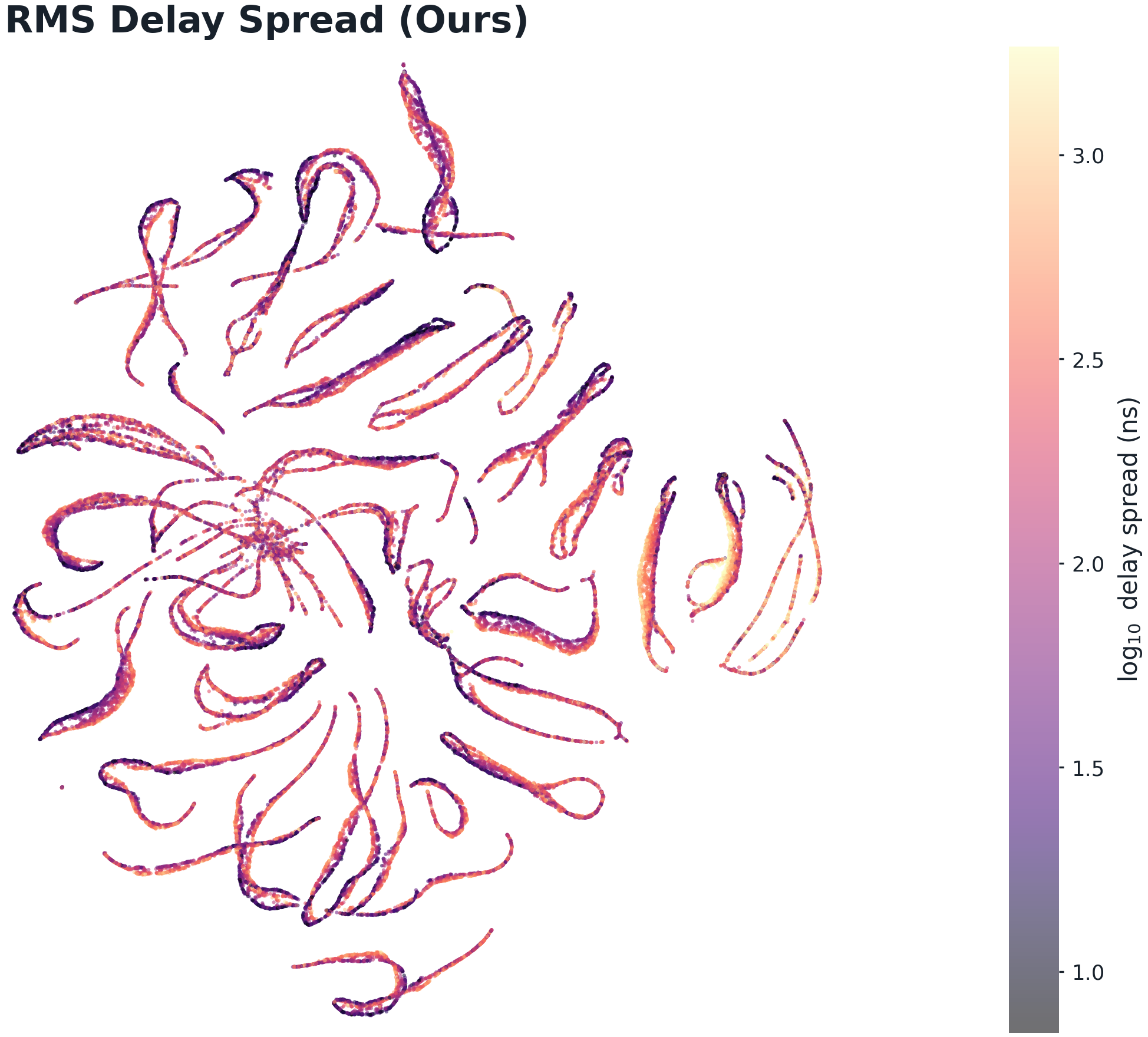}
    \caption{RMS delay spread.}
\end{subfigure}

\caption{ALF encoder's UMAP coordinates colored by receiver SNR and RMS delay
spread.}
\label{fig:umap_propagation}
\end{minipage}

\newpage

We visualize the final 768-D signal embedding on 64,000 deterministically
selected held-out seen-channel recordings using UMAP \citep{mcinnes2018umap}. Embeddings are first normalized to unit
length and reduced to 50 principal components, which retain 45.1\% of their
variance. UMAP is then fitted with 50 neighbors, minimum distance 0.15, cosine
distance, and random seed 42. All panels use the same recordings and coordinates, only the displayed labels change.

Figures~\ref{fig:umap_waveform} and~\ref{fig:umap_propagation} show that
ALF's final embedding organizes waveform class into distinct manifolds while
retaining smooth variation in mean per-antenna pre-combining SNR and RMS delay
spread. In contrast, the adapted external encoders exhibit more overlapping
qualitative geometry in Figure~\ref{fig:baseline_umap}. As with any UMAP
visualization, these plots are descriptive; the frozen-encoder probe results
(Table~\ref{tab:frozen_encoder_probes}) provide the quantitative comparison.

\subsection{Zero-shot alignment and exact retrieval}
\label{app:encoder_zeroshot}

Beyond reconstruction, we evaluate whether the frozen encoder is directly useful for downstream RF understanding. These experiments test its label-aligned embedding space without adapting the signal encoder: conditional zero-shot prediction measures attribute-level organization, while exact retrieval measures whether a signal can recover its complete metadata configuration.

For conditional zero-shot prediction, we construct a metadata prototype for each candidate target value. For waveform classification, the candidates are the 32 waveform classes. For SNR and occupancy, we partition the training-set values into 32 quantile bins and use one prototype per bin. A prototype is the average metadata embedding from 512 training recordings having
the corresponding waveform class or continuous-value bin; their remaining metadata attributes retain their natural training-set variation. At test time, we assign each signal to the prototype with the largest cosine similarity. For continuous targets, the prediction is the training-set median of the selected bin. Test labels are used only to score these predictions.

Exact paired retrieval tests a stricter form of alignment. For each signal
query, we form its positive candidate from the complete 6-label metadata
configuration of the same recording: waveform class, scene type, propagation
state, mean per-antenna pre-combining SNR, spectral occupancy, and RMS delay
spread. We rank this candidate against the metadata configurations of every
recording in the same test gallery using cosine similarity. Recall at $K$ (R@$K$)
is the fraction of queries for which the matched 6-factor configuration appears
among the $K$ highest-ranked candidates.

\begin{minipage}{0.9\textwidth}
\centering
\includegraphics[width=0.85\linewidth]{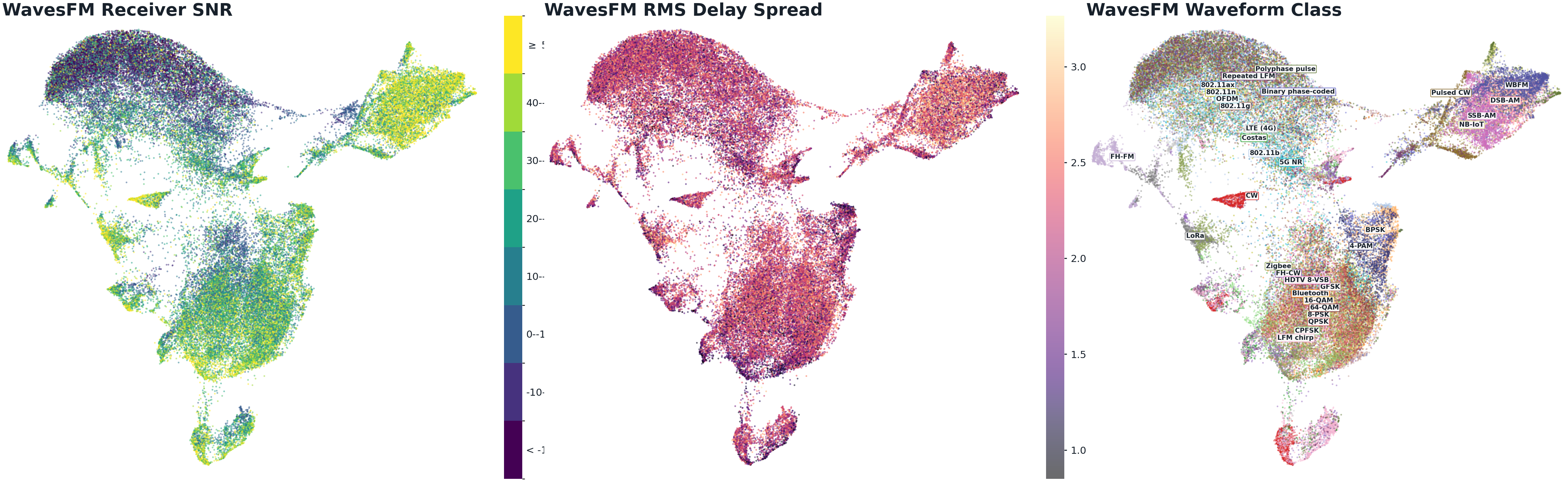}

\vspace{2pt}
\includegraphics[width=0.85\linewidth]{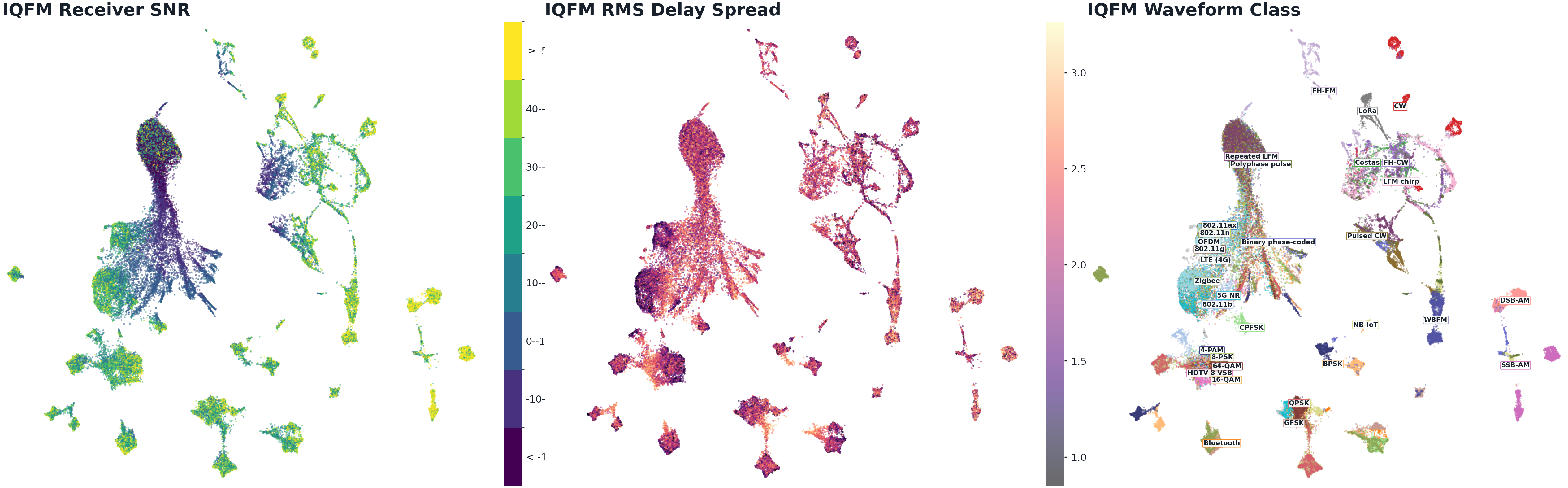}

\vspace{2pt}
\includegraphics[width=0.85\linewidth]{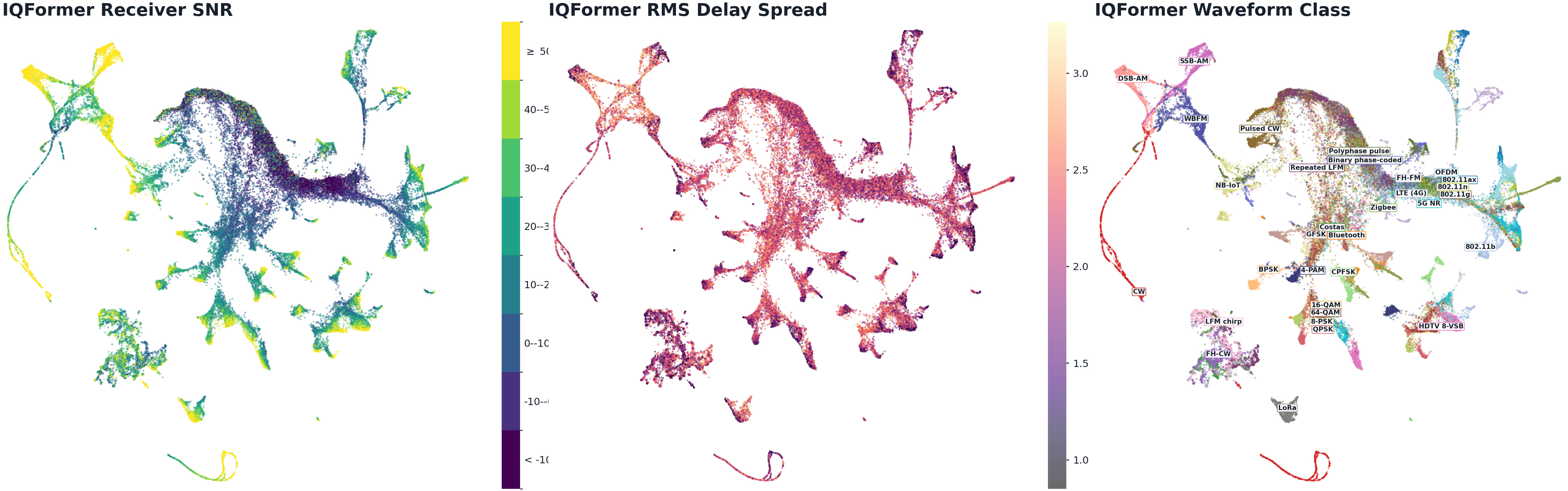}

\captionof{figure}{Frozen-encoder UMAPs for WavesFM, IQFM, and IQFormer.
Within each row, waveform, receiver-SNR, and delay-spread panels share the
same coordinates.}
\label{fig:baseline_umap}
\end{minipage}

\begin{table}[H]
\centering
\small
\caption{Zero-shot alignment and exact paired retrieval for the proposed
encoder. Conditional zero-shot values are mean \(\pm\) 95\% confidence interval
over five prototype-sampling seeds.}
\label{tab:encoder_zeroshot}
\begin{tabular}{lcc}
\toprule
Metric & Seen channels & Unseen channels \\
\midrule
Zero-shot waveform accuracy (\%) \(\uparrow\)
    & \(95.55\pm0.21\) & \(95.87\pm0.23\) \\
Zero-shot SNR RMSE (dB) \(\downarrow\)
    & \(5.204\pm0.213\) & \(4.811\pm0.176\) \\
Zero-shot occupancy RMSE (log fraction) \(\downarrow\)
    & \(0.151\pm0.009\) & \(0.134\pm0.010\) \\
\midrule
Exact retrieval R@1 (\%) \(\uparrow\)
    & \(28.17\) & \(27.72\) \\
Exact retrieval R@5 (\%) \(\uparrow\)
    & \(60.02\) & \(59.62\) \\
Exact retrieval R@10 (\%) \(\uparrow\)
    & \(72.84\) & \(72.66\) \\
\bottomrule
\end{tabular}
\end{table}

Each retrieval gallery contains 300K candidate metadata configurations, so every
query is ranked against a large set of closely related negative recordings that may
share waveform, channel attributes, or both. Despite this strict matched-gallery
setting, the encoder retrieves the complete 6-factor configuration within its
top 10 candidates for $72.84\%$ of seen-channel queries and $72.66\%$ of
channel-held-out queries, demonstrating that its aligned embedding preserves
substantially more than waveform identity alone.

\newpage
\section{ALF Details and Evaluations}

\subsection{Hyperparameters}
\label{appx:hyperparams}
\begin{longtable}{l l p{0.52\linewidth}}
\caption{\modelname{} hyperparameters. Values marked $\dagger$ are inherited
from the frozen autoencoder and are not free parameters of the flow.}
\label{tab:appx-hyperparams} \\
\toprule
Hyperparameter & Value & Role \\
\midrule
\endfirsthead
\toprule
Hyperparameter & Value & Role \\
\midrule
\endhead
\bottomrule
\endfoot
\multicolumn{3}{l}{\emph{Backbone}} \\
Depth                 & $12$    & Number of DiT blocks \citep{peebles2023scalable} \\
Width                 & $768$   & Residual and token width of every DiT block \\
Heads                 & $12$    & Attention heads, head dimension $64$\\
MLP expansion         & $4$     & Hidden width of each block's feed-forward layer \\
Max patch tokens      & $32$    & Longest record the model emits in one pass \\
Patch length          & $1024^{\dagger}$ & IQ samples per patch \\
Patch dimension       & $768^{\dagger}$  & Latent width per patch \\
Flow parameters       & $131{,}134{,}464$ & Size of the flow network, excluding the anchor head, the frozen autoencoder, and the conditioning encoder. \\
\midrule
\multicolumn{3}{l}{\emph{Conditioning}} \\
Conditioning dimension & $768$  & Width of the label vector \\
Conditioning factors   & $6$    & Waveform class, scene type, propagation state, estimated SNR, occupancy, RMS delay spread $\S$~\ref{ap:dataset-labels} \\
Label encoder          & frozen & The encoder is fixed \\
Prior                  & enabled & Sampling from label encoder mean and variance \\
\midrule
\multicolumn{3}{l}{\emph{Latent preprocessing and objective}} \\
Objective              & $v=Z_1-Z_0$ & Rectified-flow velocity target \\
Standardizer batches   & $60$    & Batches per rank used to fit the per-dimension latent mean and variance \\
Clip $\sigma$          & $8.0$   & Winsorising threshold applied after standardising \\
Time distribution      & logit-normal $(0,1)$ & Distribution the flow time $t$ is drawn from during training \\
Signal scale           & $10^{4}$ & Fixed input scaling the encoder was fitted under. \\
\midrule
\multicolumn{3}{l}{\emph{Optimization}} \\
Optimizer              & AdamW   & - \\
Learning rate          & $1\times10^{-4}$ & Peak rate. \\
Weight decay           & $0.02$  & Applied to trainable flow parameters only \\
Gradient clip          & $1.0$   & Global gradient-norm clip applied before each step. \\
Batch size             & $64\times4$ & Records per rank across four ranks; effective batch size $256$. \\
Epochs                 & $20$    & Full passes over the training split. \\
Schedule               & warmup $+$ cosine & Stepped once per epoch: one warmup epoch, then cosine decay. The rate is flat at peak for the first two epochs and at ${\sim}0.7\%$ of peak at the end. \\
Seed                   & $42$    & Initialization and data order. \\
\midrule
\multicolumn{3}{l}{\emph{Anchor}} \\
Anchor layer           & $6$ of $12$ & At the midpoint of DiT Depth \\
Anchor hidden width    & $768$   & Width of the head's MLP \\
Anchor weight          & $0.25$  & Weight of the auxiliary term relative to the flow loss\\
Hold / ramp epochs     & $0$ / $2$ & Weight ramps linearly \\
Time weighting         & $t / \bar{t}$ & Weights each record's anchor loss toward low $t$ \\
Head parameters        & $2{,}363{,}904$ & Discarded at inference, only used for anchoring\\
\midrule
\multicolumn{3}{l}{\emph{Sampling}} \\
ODE steps              & $50$    & Integration steps from noise to data. \\
Integrator             & Heun    & Second-order predictor--corrector solver. \\
Guidance scale         & $1.0$   & No classifier-free guidance; conditioning enters only through the modulation path. \\
Output normalisation   & RMS     & Rescales each record to unit mean power \\
Records per corpus     & $100{,}000$ & Generated under held-out test-split conditioning. \\
Sampling seed          & $0$     & Sampling noise draw. \\
Edge refiner           & $\geq30$\,dB & Refined variants only: a post-hoc pass over records whose conditioning SNR clears this gate. Unrefined corpora are otherwise identical. \\
\end{longtable}

\vspace{-20pt}

\subsection{Label-Conditioned Gaussian Prior}
\label{sec:gaussian_prior}
Rather than transporting isotropic white Gaussian noise $\mathbf{Z}_0 \sim \mathcal{N}(0, \mathbf{I})$ to clean latents, \modelname{} initializes its rectified-flow trajectories from the multimodal source distribution $p_{\mathrm{src}}(\cdot \mid \mathbf{Y})$ parameterized by the frozen probabilistic label tower $\theta_{\mathrm{lbl}}(\mathbf{Y})$. The label encoder consumes exactly the six physical RF metadata factors $\mathbf{Y}$---waveform class, SNR, channel occupancy, propagation scenario, channel state, and RMS delay spread (Appendix~\ref{ap:dataset-labels})---mapping them into a 768-dimensional diagonal Gaussian. For an arbitrary record of length $P$ patches (spanning 2{,}048 to 32{,}768 IQ samples at 1{,}024 samples per patch), the initial flow state at $t=0$, denoted $\mathbf{Z}_0 \in \mathbb{R}^{P \times 768}$, is constructed by taking $P$ independent draws from that record's conditioning cloud:
\begin{equation}
    \mathbf{Z}_0^{(p)} = \boldsymbol{\mu}_{\mathrm{lbl}}(\mathbf{Y}) + \boldsymbol{\sigma}_{\mathrm{lbl}}(\mathbf{Y}) \odot \mathbf{z}_p, \quad \mathbf{z}_p \sim \mathcal{N}(0, \mathbf{I}), \quad p \in \{1, \dots, P\}.
    \label{eq:label_source}
\end{equation}
Because the rectified flow interpolation $\mathbf{Z}_t = (1-t)\mathbf{Z}_0 + t\mathbf{Z}_1$ and velocity target $v = \mathbf{Z}_1 - \mathbf{Z}_0$ require only that training and sampling draw from the identical distribution, this formulation restricts the flow's starting domain to a class-structured hyperspherical shell ($\|\boldsymbol{\mu}_{\mathrm{lbl}}(\mathbf{Y})\| \approx 24.9$ versus $\sqrt{768} \approx 27.7$) retaining approximately $43\%$ of standard normal entropy per dimension ($\sigma_{\mathrm{lbl}} \approx 0.444$). Under this label-conditioned source, conditioning dropout is strictly set to $0.0$; an unconditional branch would otherwise default to standard normal noise, invalidating the velocity-mixing assumption of classifier-free guidance. Conditioning information is thus present twice: statically as the spatial boundary state $\mathbf{Z}_0$ of the flow trajectory, and dynamically across all flow times $t$ via adaptive LayerNorm (adaLN) modulation in the DiT backbone, which jointly injects the normalized direction of $\boldsymbol{\mu}_{\mathrm{lbl}}(\mathbf{Y})$ alongside a learned patch-length embedding of $P$.

\subsection{Separate Spectral Autoencoder and Mid-Depth Anchoring}
\label{sec:spectral_anchor}
To enforce physically consistent spectral properties without incurring test-time computational overhead, \modelname{} incorporates the auxiliary mid-depth supervisory anchor driven by the independently trained, frozen spectral autoencoder (38.6M parameters). Operating on per-patch STFT log-power spectrograms ($256$ frequency bins $\times$ $13$ time frames per patch, computed with an FFT size and window of $256$, hop size $64$, and antenna power averaging), this autoencoder is trained with absolute level preservation, ensuring the latent representation faithfully encodes both spectral shape and wide dynamic range (spanning $-15$\,dB to $+68$\,dB SNR). Its post-mixing latent sequence $\mathbf{Z}_{\mathrm{spec}}$ reflects global recording context via a 4-layer Transformer mixer while maintaining local invertibility back to the patch spectrum via a per-patch 2D convolutional decoder.

During flow training, the spectral autoencoder remains frozen. The auxiliary prediction head $\theta_{\mathrm{anchor}}$ ($768 \to 768 \to 768$ MLP, 2.36M parameters) is tapped off the intermediate hidden representation $\mathbf{H}_t^{(6)}$ after the $6^{th}$ Transformer block of the 12-block flow backbone. The head acts as a mid-depth denoising probe trained to predict the standardized clean spectrogram latent $\widetilde{\mathbf{Z}}_{\mathrm{spec}}$ via mean squared error. Velocity matching is deliberately avoided because coupling the IQ latent noise draw with an independently trained spectral latent space is ill-posed. As shown in Equation~\ref{eq:alf_objective}, the auxiliary spectral loss is weighted by $t / \bar{t}$:
\begin{equation}
    \mathcal{L}_{\mathrm{anchor}} = \frac{t}{\bar{t}} \, \big\| \theta_{\mathrm{anchor}}\!\left(\mathbf{H}_t^{(6)}\right) - \widetilde{\mathbf{Z}}_{\mathrm{spec}} \big\|_2^2,
    \label{eq:aux_spectral_loss}
\end{equation}
which concentrates gradient pressure toward the target IQ distribution ($t \to 1$) where intermediate activations carry recoverable spectral structure. The anchor term enters the total objective with weight $\lambda_{\mathrm{anchor}}(e)$ ramped linearly to $0.25$ over the first two epochs, with its output projection zero-initialized so as not to destabilize early flow optimization. Because $\theta_{\mathrm{anchor}}$ provides representational regularization exclusively during backpropagation, it is discarded at inference, adding zero parameters and zero FLOPs to generation.

\subsection{Training Consistency and Generation Protocol}
\label{sec:training_consistency}
To isolate the empirical gains of the inductive biases during ablation, all flow models were trained to a convergence under identical optimization protocols: AdamW with decoupled weight decay ($0.02$), gradient-norm clipping at $1.0$, a batch size of $256$ across 4 ranks, and a 20-epoch warmup-cosine learning rate schedule peaking at $10^{-4}$ (with checkpoints selected by lowest held-out validation loss; epoch 19 for \modelname{}). For downstream evaluation, all synthetic corpora were generated under an identical protocol: 100{,}000 records conditioned on the held-out test split under fixed sampling seed $0$, integrated from $t=0$ to $t=1$ using a second-order Heun ODE solver at $50$ steps, guidance scale $1.0$ (no classifier-free guidance), decoded to time-domain waveforms using the frozen decoder $\theta_{\mathrm{dec\text{-}IQ}}$, and scaled via per-record RMS power normalization to unit mean energy (see Table~\ref{tab:appx-hyperparams} for full hyperparameter specifications). The 5.6M development set with the same deterministic train/validation is used for training.

The three baselines differ in how far they support the corpus's 5
recording lengths ($N_{\mathrm{samples}}\in\{2048,4096,8192,16384,32768\}$). \textbf{TimeWeaver} \citep{narasimhan2024time} is natively variable-length: a single model is trained jointly on all five buckets because nothing in the denoiser is learned per length. Batches are length-homogeneous (bucket sampling, with no padding or attention mask avoids feeding masked zeros into the target), and the per-rank batch is scaled as $16\cdot 2048/L$ so that $B\!\cdot\!C\!\cdot\!L$ is held constant; together with the sum-reduced training loss, this makes the loss scale length-invariant, so the original learning rate transfers mostly unchanged. \textbf{WaveStitch} \citep{shankar2025wavestitch} covers
the same range only by tiling: its S4 backbone \citep{gu2022iclr} sees a fixed $256$-sample window,
so a length-$L$ record is cut into $L/256$ non-overlapping windows, the same
conditioning vector is broadcast to every window, each window is denoised
independently, and the results are concatenated. We generate in the
metadata-only mode, which uses neither stitching nor self-guidance, so the
windows are uncoupled: the model has no mechanism for structure longer than one
window, and a seam occurs every $256$ samples. \textbf{RF-Diffusion} \citep{chi2024rfdiffusion} is not
adapted at all---it applies full $O(N^{2})$ self-attention to raw samples with
no patching, which is tractable at $2048$ but not at $8192$ and above without an
architectural change, so it is both trained and evaluated at $L=2048$ only. Its
numbers therefore correspond to the $2048$ column and are not comparable to the
other arms' length-pooled figures.
Again, the 5.6M development set with the same deterministic train/validation is used for training, and we trained till a convergence is reached.

\newpage
\subsection{Inference}
\label{subsec:inference}
\begin{wrapfigure}{hr}{0.4\textwidth}
  \vspace{-0.5em}
  \centering
  \includegraphics[width=0.38\textwidth]{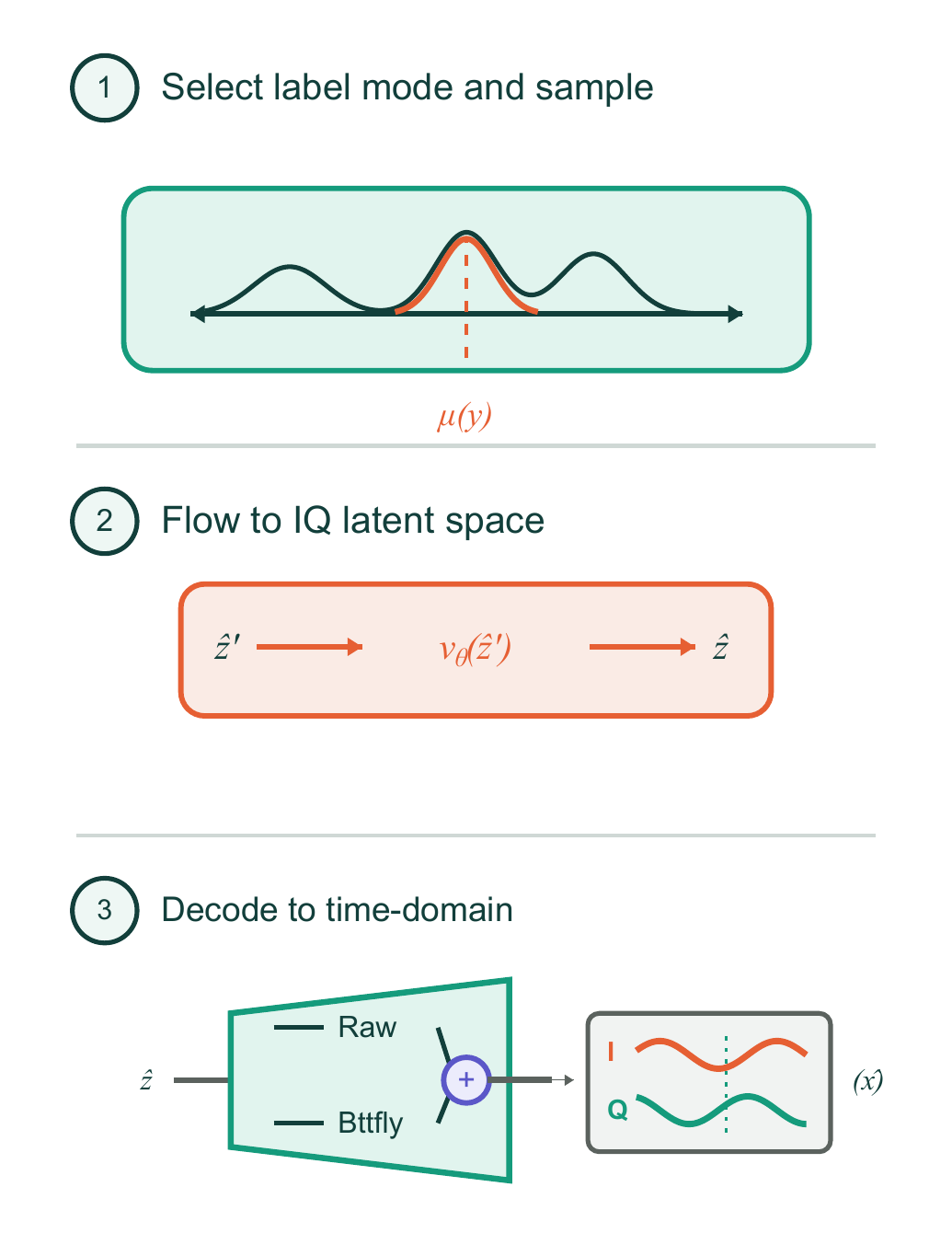}
  \caption{Inference Pipeline}
  \label{fig:inference}
  \vspace{-15pt}
\end{wrapfigure}

The cost of a multi-stage anchored flow training finally yields its benefits at inference. Pure data generation is divided into three pieces: Sample, Flow, Decode (Figure~\ref{fig:inference}). The only user-specified information needed is the desired features that characterize the target waveform. Additionally, our per-patch design ($\S$~\ref{subsec:stage1},\ref{subsec:anchoredflow}) makes inference length-agnostic, with potentially no limit on generation length (Appendix~\ref{ap:ar-gen}). Label specification conditions the flow model's prior, giving $p_{z(1)}(z|c)$. \modelname\space continues to pay dividends; its latent nature achieves significant speed advantages to prior works (Table~\ref{tab:inf-speed}, $\S$~\ref{sec:results}).

\subsection{Refined vs. Unrefined}
\label{ap:ref-unref}
Due to the nature of the per-patch AE structure, although trained for continuity, some results still show small fluctuations at patch seams. In the frequency domain, these high-frequency components translate to stripes in generations. Although signal integrity is preserved (i.e., overall SNR is the same), these artifacts resulted as increase in MSE metrics (PSD and Spectrogram). Thus, we built a refining mechanism for post-decode smoothing. 

The refiner is a small dilated 1-D residual CNN with per-antenna inputs $I$, $Q$, $|z|$, and the cosine and sine of the lag-one phase increment; the mixer tokens of the two patches adjoining each join, $(h_{p-1}, h_p)$, are linearly projected and added to the feature map within $\pm 32$ samples of that join. This network emits a per-antenna polar correction $(\Delta a, \Delta\phi)$, $\tanh$-bounded to $\pm 0.5$ in log-amplitude and $\pm \pi/2$ in phase, applied multiplicatively as $z \leftarrow z\, e^{m\Delta a} e^{j m \Delta\phi}$. We use a polar correction instead of a complex linear blend because blending phase-decorrelated signals can create amplitude notches and broadband artifacts. For 802.11ax ($\kappa \approx 0.15$), scaling and rotation avoid these issues. At inference, an SNR threshold is set to avoid tampering with intentionally noisy generations, and records above the determined level are passed through this refinement pipeline to smooth seam transitions.

\begin{table}[H]
\centering
\caption{MSE within 95\% Confidence over 100k Recordings}
\label{tab:mse-comparison}
\small
\begin{tabular}{lrrrr}
\toprule
& \multicolumn{2}{c}{Unrefined} & \multicolumn{2}{c}{Refined} \\
\cmidrule(lr){2-3}\cmidrule(lr){4-5}
Run & Med PSD MSE & Med Spec MSE & Med PSD MSE & Med Spec MSE \\
\midrule
Baseline w/o Prior
& 19.77 $\pm$ 0.20
& 127.85 $\pm$ 0.48
& 17.87 $\pm$ 0.16
& 120.29 $\pm$ 0.50 \\
Baseline w/ Prior
& 19.17 $\pm$ 0.18
& 128.71 $\pm$ 0.50
& 17.48 $\pm$ 0.15
& 121.09 $\pm$ 0.47 \\
\textbf{\modelname\space (ours)}
& \textbf{18.52 $\pm$ 0.19}
& \textbf{126.54 $\pm$ 0.46}
& \textbf{16.86 $\pm$ 0.14}
& \textbf{119.08 $\pm$ 0.45} \\
\midrule
\modelname\space vs. Base w/o Prior
& +6.3\% & +1.0\% & +5.6\% & +1.0\% \\
\modelname\space vs. Base w/ Prior
& +3.4\% & +1.7\% & +3.5\% & +1.7\% \\
\bottomrule
\end{tabular}
\end{table}
Table~\ref{tab:mse-comparison} shows the distributional improvements of roughly 2dB and 6.5dB for PSD and Spectrogram MSEs, respectively, across all ablations. A design choice was made to include this refiner for all latent flow models (ALF, LF-M, LF-N) compared in the main paper (Table~\ref{tab:combined-rf-robustness},$\S$~\ref{sec:results}).

\subsection{Robustness}
Table~\ref{tab:combined-rf-robustness} shows that ALF slowed the record collapse curve. Figure~\ref{fig:robust} visualizes the reclassification data in Table~\ref{tab:source-corruption-multiarm} and improvements, including a broader view across inference steps. Notably, the same trend continues through 5 steps; at a single inference step, ALF remains the best-performing. Finally, note that Time Weaver experiences the largest collapse in this region due to the naturally noisier process of diffusion models over flows. 
\begin{figure}[H]
    \centering
    \includegraphics[width=0.8\linewidth]{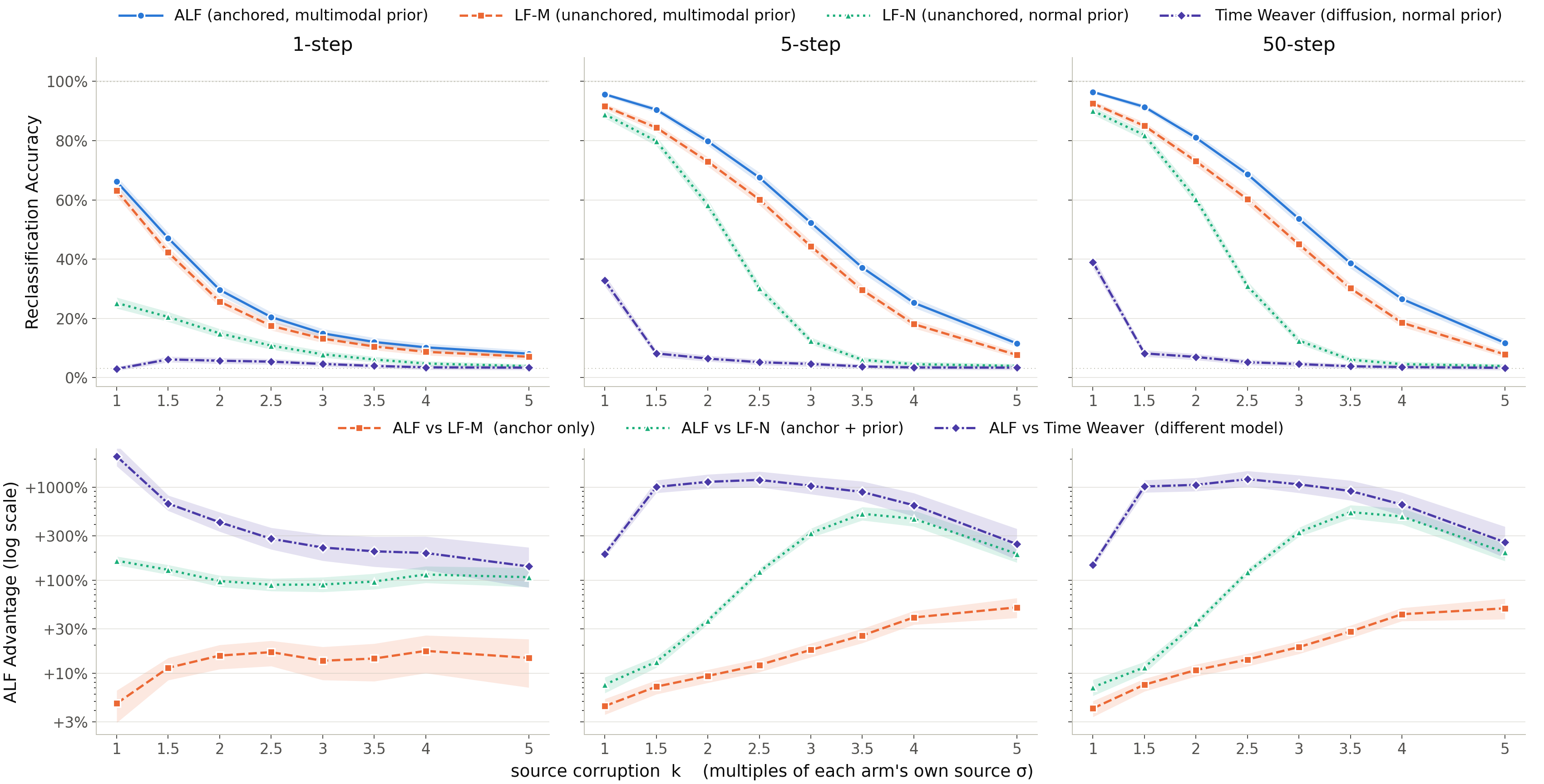}
    \caption{Signal reclassification accuracies for low-density sampling and various inference steps}
    \label{fig:robust}
\end{figure}
\vspace{-20pt}
\begin{table*}[h]
\centering
\scriptsize
\begin{threeparttable}

\caption{Recovery under source corruption at matched sigma multiples $\mathbf{k}\sigma$}
\label{tab:source-corruption-multiarm}

\setlength{\tabcolsep}{2.5pt}
\begin{tabular}{c r r c r r c c r r c r r c c}
\toprule
& \multicolumn{7}{c}{5 flow steps}
& \multicolumn{7}{c}{50 flow steps} \\
\cmidrule(lr){2-8} \cmidrule(lr){9-15}
& \multicolumn{3}{c}{Multimodal Gaussian}
& \multicolumn{4}{c}{Standard normal}
& \multicolumn{3}{c}{Multimodal Gaussian}
& \multicolumn{4}{c}{Standard normal} \\
\cmidrule(lr){2-4} \cmidrule(lr){5-8}
\cmidrule(lr){9-11} \cmidrule(lr){12-15}
$\mathbf{k}$
& \textbf{ALF} & LF & $\Delta_{\%}^{\mathrm{LF}}$
& LF & TW & $\Delta_{\%}^{\mathrm{LF}}$ & $\Delta_{\%}^{\mathrm{TW}}$
& \textbf{ALF} & LF & $\Delta_{\%}^{\mathrm{LF}}$
& LF & TW & $\Delta_{\%}^{\mathrm{LF}}$ & $\Delta_{\%}^{\mathrm{TW}}$ \\
\midrule
1.0
& \textbf{95.7} & 91.6 & \textbf{+4.5\%}
& 88.9 & 32.8 & \textbf{+7.6\%} & \textbf{+192.1\%}
& \textbf{96.5} & 92.6 & \textbf{+4.2\%}
& 90.1 & 39.0 & \textbf{+7.1\%} & \textbf{+147.3\%} \\
1.5
& \textbf{90.5} & 84.4 & \textbf{+7.2\%}
& 79.8 & 8.2 & \textbf{+13.3\%} & \textbf{+1009.2\%}
& \textbf{91.4} & 85.0 & \textbf{+7.6\%}
& 81.9 & 8.2 & \textbf{+11.6\%} & \textbf{+1020.7\%} \\
2.0
& \textbf{79.8} & 72.9 & \textbf{+9.4\%}
& 58.3 & 6.4 & \textbf{+36.9\%} & \textbf{+1145.4\%}
& \textbf{81.0} & 73.1 & \textbf{+10.9\%}
& 60.3 & 7.0 & \textbf{+34.4\%} & \textbf{+1062.3\%} \\
2.5
& \textbf{67.5} & 60.1 & \textbf{+12.3\%}
& 30.2 & 5.2 & \textbf{+123.8\%} & \textbf{+1202.1\%}
& \textbf{68.7} & 60.2 & \textbf{+14.1\%}
& 30.9 & 5.2 & \textbf{+122.2\%} & \textbf{+1223.8\%} \\
3.0
& \textbf{52.3} & 44.3 & \textbf{+17.9\%}
& 12.4 & 4.6 & \textbf{+321.8\%} & \textbf{+1037.8\%}
& \textbf{53.6} & 45.0 & \textbf{+19.2\%}
& 12.5 & 4.6 & \textbf{+329.5\%} & \textbf{+1075.3\%} \\
3.5
& \textbf{37.1} & 29.6 & \textbf{+25.6\%}
& 6.0 & 3.8 & \textbf{+519.0\%} & \textbf{+890.4\%}
& \textbf{38.7} & 30.2 & \textbf{+28.2\%}
& 6.0 & 3.8 & \textbf{+540.9\%} & \textbf{+913.9\%} \\
4.0
& \textbf{25.2} & 18.0 & \textbf{+39.9\%}
& 4.5 & 3.4 & \textbf{+459.2\%} & \textbf{+641.3\%}
& \textbf{26.5} & 18.5 & \textbf{+43.3\%}
& 4.5 & 3.5 & \textbf{+485.9\%} & \textbf{+651.8\%} \\
5.0
& \textbf{11.5} & 7.6 & \textbf{+51.2\%}
& 3.9 & 3.3 & \textbf{+192.9\%} & \textbf{+244.9\%}
& \textbf{11.7} & 7.8 & \textbf{+50.0\%}
& 3.9 & 3.3 & \textbf{+200.0\%} & \textbf{+257.1\%} \\
\bottomrule
\end{tabular}

\begin{tablenotes}[flushleft]
\scriptsize
\item[] \textit{Note.} LF-M: unanchored LF with a metadata prior; LF-N: unanchored LF with a standard-normal prior. TimeWeaver (TW) uses equal-NFE DDIM sampling with 10 and 100 steps, respectively.
\end{tablenotes}

\end{threeparttable}
\end{table*}

\subsection{Style Transfer Cont.}
\label{app:anchored-vs-unanchored-ST}
\begin{figure}[H]
    \centering
    \includegraphics[width=0.8\linewidth]{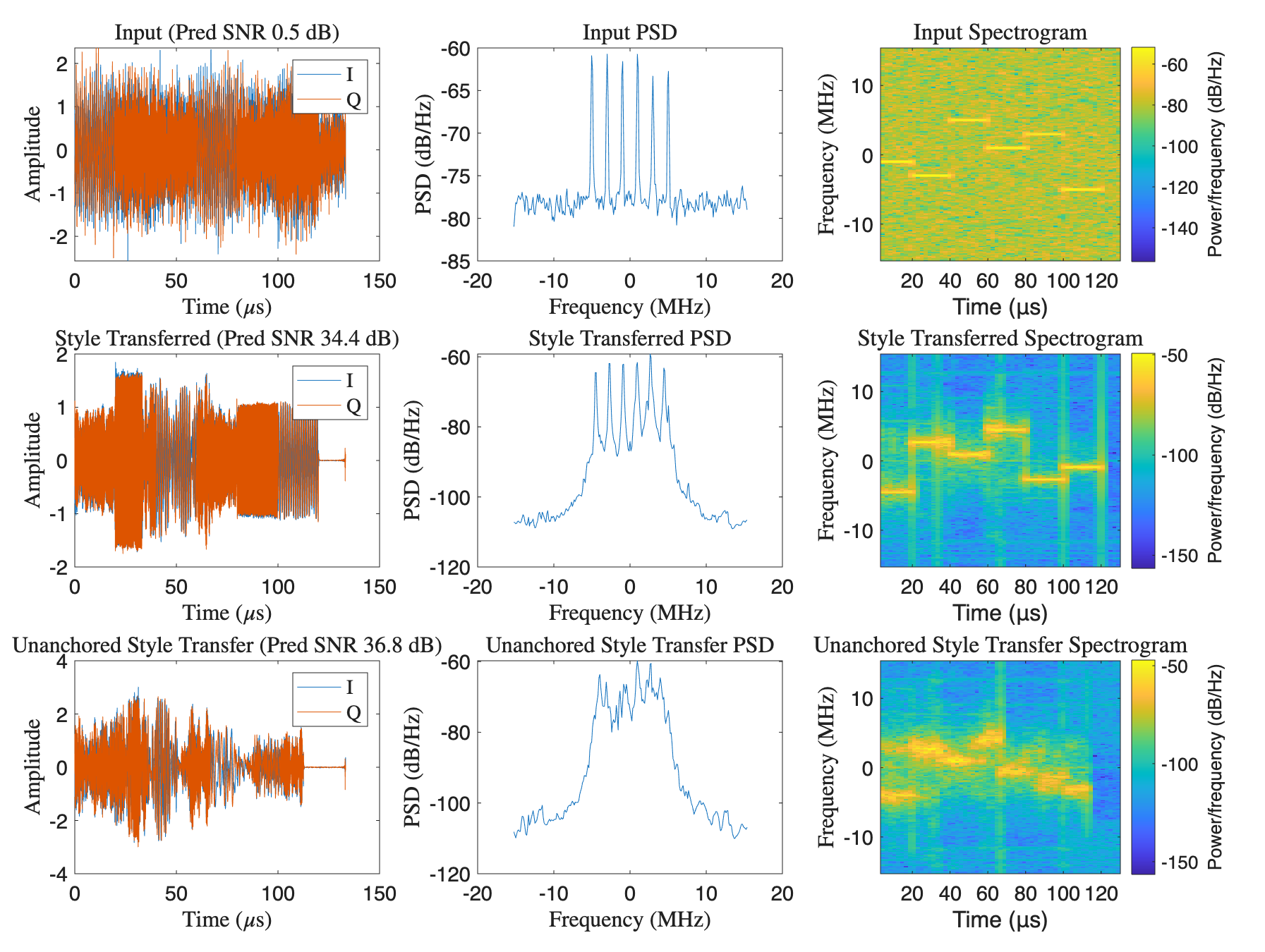}
    \caption{Unanchored style transfers are more prone to spectral blurring without anchored guidance}
\end{figure}
Signals are encoded, and their modes are decided through a max likelihood defined as:
\begin{equation}
\hat{c} \;=\; \arg\max_{c \in \mathcal{C}}
\frac{1}{P}\sum_{p=1}^{P}
\log \mathcal{N}\!\bigl(z_{0}^{(p)}\,\big|\,\mu_{c},\,\mathrm{diag}(\sigma_{c}^{2})\bigr)
\end{equation}
We also highlight that anchoring reduces the observed blurring effect in the unanchored result. The following uses naive sampling while ablating ALF and LF-M (unanchored metadata-prior) models. As seen below, the unanchored style transfer's spectrogram has far less sharp frequency detail. This is expected because anchoring constrains flowed latent trajectories during style-transfer regeneration (Appendix~\ref{app:theory}).

\subsection{Within-Recording Phase Consistency}
\label{sec:within-record-phase}
\begin{figure}[H]
    \centering
    \includegraphics[width=0.8\linewidth]{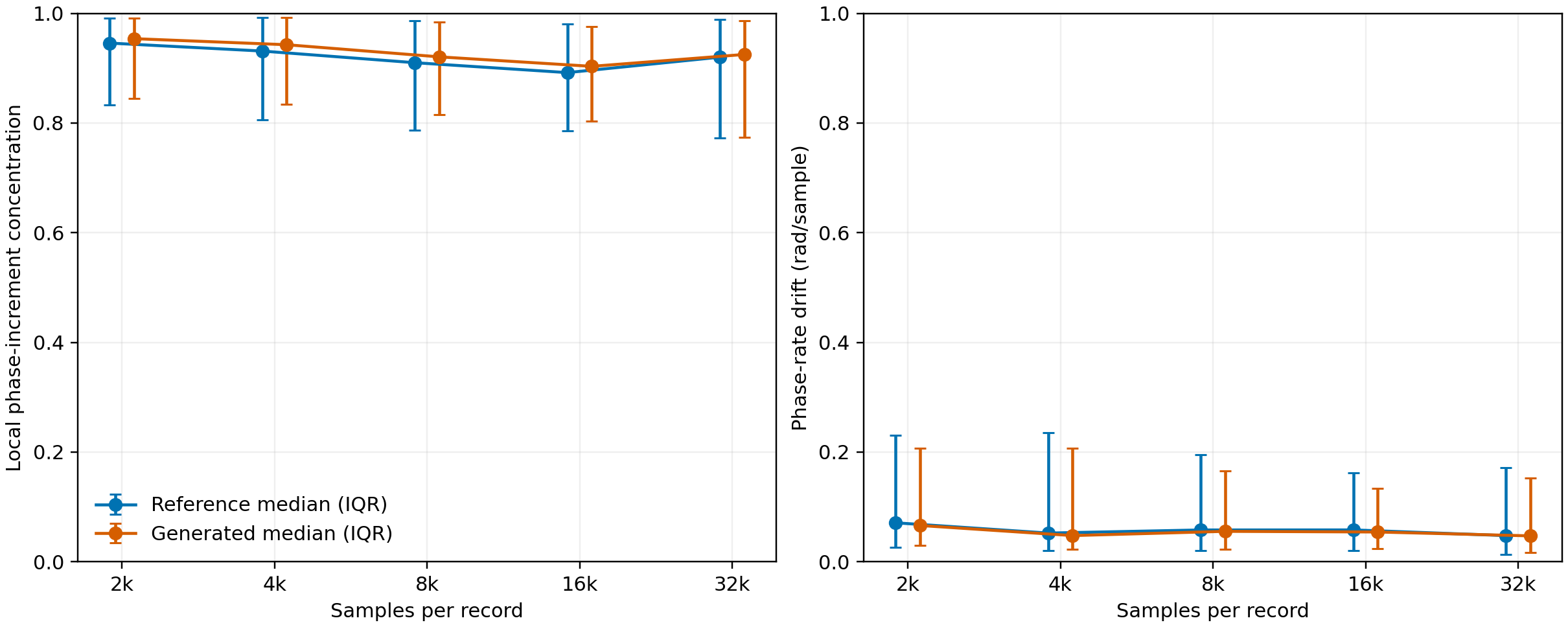}
    \caption{Within-record phase behavior by waveform length. Points show medians and bars show interquartile ranges across records. Left plot measures the consistency of adjacent-sample phase changes within normalized windows; higher is better. Right plot shows phase-rate drift, which measures variation of the window-level phase increment across each record; lower is better.}
    \label{fig:phase-length}
\end{figure}

\begin{figure}[t]
    \centering
    \includegraphics[width=0.95\linewidth]{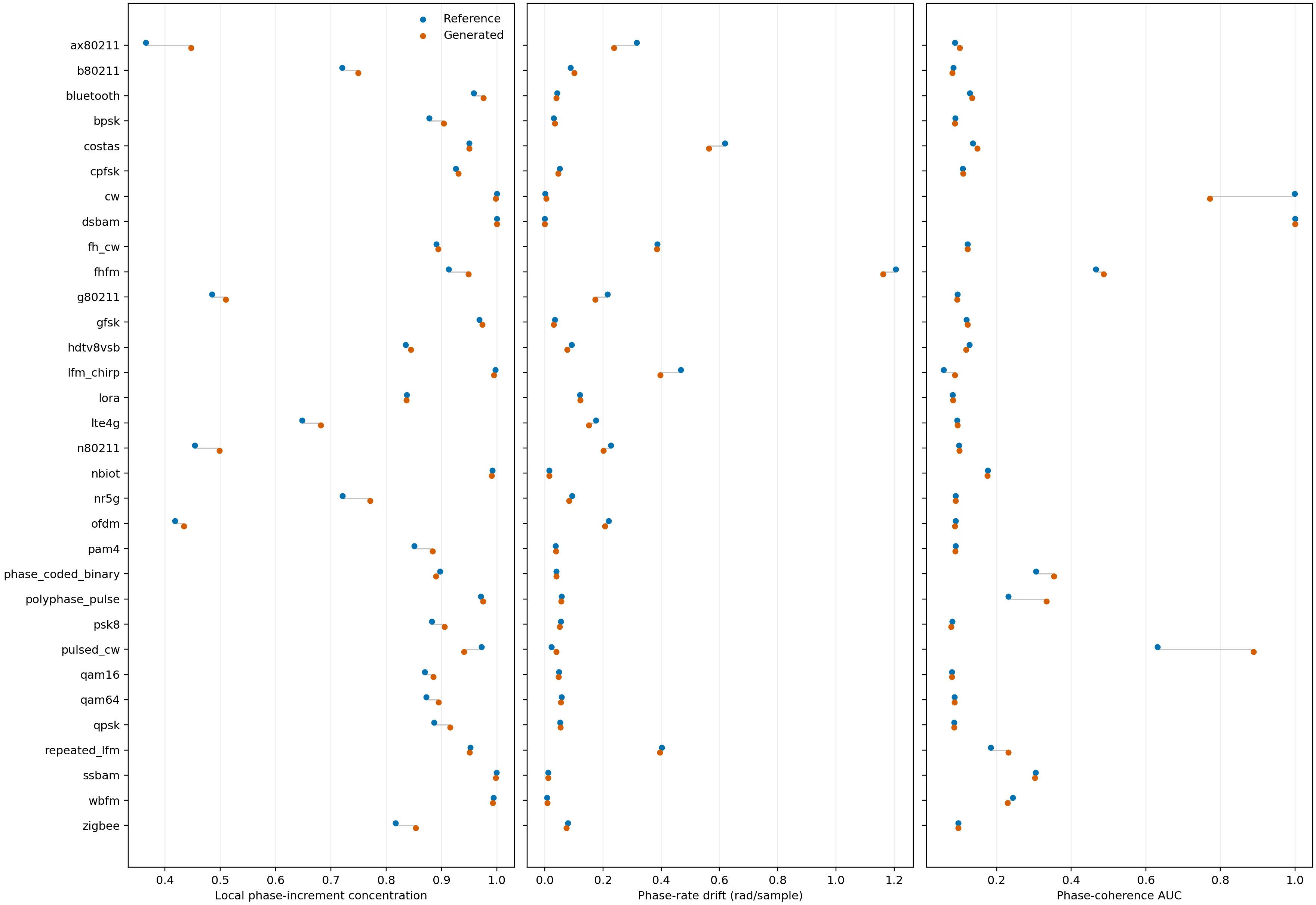}
    \caption{Within-record phase behavior by waveform family. Each point is the median across records; connecting lines pair generated and reference distributions for each waveform. The panels report local phase-increment concentration, phase-rate drift across normalized windows, and phase-only coherence averaged over normalized lags, where closer reference and generated dots means better performance.}
    \label{fig:phase-waveform}
\end{figure}

Phase encodes timing and carrier-related structure essential for coherent wireless reception; even small local phase errors can disrupt demodulation and cross-antenna coherence, making phase preservation a stringent test of long-recording generation. We evaluated whether ALF preserves within-recording phase behavior using 100,000
generated recordings conditioned on a subset of the channel-held-out test-set
conditions and compared them with the corresponding reference recordings. We
retained the 71,673 recordings with stored mean per-antenna pre-combining
$\mathrm{SNR}>10$~dB, as lower SNR substantially reduces phase
interpretability, and treated the 4 antennas within each recording as
repeated views rather than independent observations. For the complex signal at
antenna $a$, $z_a[n]$, we measured adjacent-sample phase increments,
\begin{equation}
\Delta\phi_a[n] = \angle\!\left(z_a[n+1]z_a[n]^*\right),
\end{equation}
excluding comparisons whose magnitude was below 10\% of that antenna's RMS
amplitude. We summarize local phase regularity by the circular concentration of
$\Delta\phi_a[n]$ within 16 normalized windows and summarize variation across
each recording by the circular spread of the window-level mean increments.
Because generated and reference recordings are independent draws, the
evaluation compares their distributions rather than requiring pointwise phase
agreement.

Figure~\ref{fig:phase-length} shows that ALF closely follows the reference
phase behavior for recording lengths from 2k through 32k samples. In
particular, the median phase-rate drift remains small and does not increase
with recording length: it is $0.066$~rad/sample at 2k samples and
$0.047$~rad/sample at 32k samples for generated recordings, compared with
$0.071$ and $0.048$~rad/sample for the reference recordings. Local
phase-increment concentration is similarly stable across recording lengths.
The waveform-class breakdown in Figure~\ref{fig:phase-waveform} shows close
agreement for most waveform families. CW and pulsed CW are the most stringent
cases because their definitions impose unusually regular phase evolution: CW
shows the clearest long-range gap, while pulsed CW has modest differences in
local concentration and drift. These exceptions do not dominate the aggregate
result. Overall, the length-stratified and waveform-class-stratified results
support that ALF produces phase-consistent recordings without systematic
degradation at longer recording lengths.

\subsection{Cross-Antenna Consistency}

We evaluated whether ALF preserves array-level structure using the same 100K
channel-held-out conditioning set as in
Section~\ref{sec:within-record-phase}. We again retained the 71,673 recordings with
stored mean per-antenna pre-combining $\mathrm{SNR}>10$~dB and compared each
generated 4-antenna recording with its corresponding reference
recording under the same metadata. This test probes a property often omitted
from RF-generation evaluations: realistic marginal antenna signals alone do
not establish that a model has learned the shared timing and relative-phase
structure of a multi-antenna recording.

For an adjacent antenna pair $(a,a+1)$, we estimate the relative delay from
the peak of the overlap-normalized, demeaned complex cross-correlation,
\begin{equation}
\hat{\tau}_{a,a+1}
=
\arg\max_{\ell\in[-16,16]}
\left|
\frac{\sum_{n\in\mathcal{I}_{\ell}}
\tilde z_a[n+\ell]\tilde z_{a+1}^{*}[n]}
{\sqrt{
\sum_{n\in\mathcal{I}_{\ell}}|\tilde z_a[n+\ell]|^2
\sum_{n\in\mathcal{I}_{\ell}}|\tilde z_{a+1}[n]|^2
}}
\right|,
\end{equation}
where $\tilde z_a$ is the mean-removed complex signal at antenna $a$ and
$\mathcal{I}_{\ell}$ contains the valid sample indices at lag $\ell$. We report
the within-recording adjacent-pair spread,
\begin{equation}
\Delta_\tau =
\max_{a\in\{1,2,3\}}\hat{\tau}_{a,a+1}
-
\min_{a\in\{1,2,3\}}\hat{\tau}_{a,a+1}.
\end{equation}
A small $\Delta_\tau$ indicates that all adjacent pairs share the same bulk
timing relationship. We restrict this measure to recordings with identifiable
correlation peaks, which includes approximately 89\% of both generated and
reference recordings.

To assess temporal relative-phase stability, we divide each recording into 16
normalized windows, estimate the phase of the complex adjacent-pair correlation
in each window, unwrap the resulting phase trajectory, and fit its weighted
slope $\hat{\beta}_{a,a+1}$. The reported recording-level phase drift is
\begin{equation}
D_\phi =
\max_{a\in\{1,2,3\}}
\frac{|\hat{\beta}_{a,a+1}|(N_{\mathrm{samples}}-1)}{2\pi},
\end{equation}
measured in cycles over the full recording. This metric tests for
antenna-specific frequency offsets or time-varying relative phase; it does not
require equal static phase across antennas, since spatial separation naturally
produces antenna-dependent path phases.

Figure~\ref{fig:cross-antenna-consistency} shows that ALF preserves tightly
shared adjacent-pair timing across all recording lengths: the median delay
spread is zero samples in every length bin, with generated interquartile ranges
no larger than one sample. Relative phase drift is also low and does not worsen
for longer recordings. ALF has modestly larger drift than the reference
distribution---from $0.023$ cycles at 2k samples to $0.0055$ cycles at 32k
samples, versus $0.0105$ and $0.0027$ cycles for the reference
recordings---but these values remain small over each complete recording. Thus,
ALF learns the principal shared-timing and stable relative-phase behavior of
the receive array.

The effective-rank results qualify this conclusion. ALF has a slightly lower
spatial-covariance effective rank, particularly in the upper tail, which
indicates modest compression of antenna-specific multipath diversity and is
consistent with increased pairwise correlation. Thus, ALF does not perfectly
reproduce every array statistic. Nevertheless, the timing and phase-drift
results show that generated antenna streams remain coherently related over long
recordings, which is the central array-level property evaluated here.

\begin{figure}[H]
    \centering
    \includegraphics[width=\linewidth]{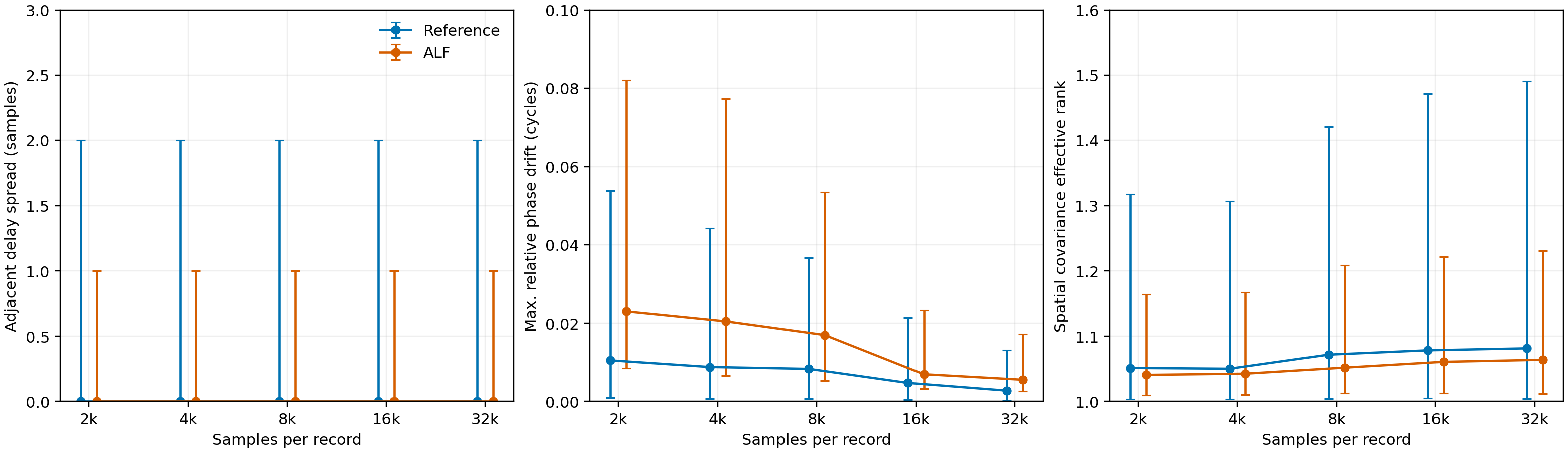}
    \caption{Cross-antenna consistency by recording length. Points show medians and bars show interquartile ranges. Adjacent delay spread measures agreement among the
    three adjacent-pair correlation-peak lags; relative phase drift measures
    the largest change in adjacent-pair correlation phase over a recording; and
    effective rank summarizes spatial-covariance diversity.}
    \label{fig:cross-antenna-consistency}
\end{figure}

\newpage
\section{Autoregressive Post-Training}
\label{ap:ar-gen}
ALF is trained to sample a recording's $P$ patch tokens jointly up to $\text{max}(P)=32$. This equates to a maximum generation length of 32,768 samples long. Although we have shown that these generation lengths are achieved with SOTA quality, additional downstream tasks like integrated sensing will require data that is orders of magnitude longer. Thus, we extend ALF to conditional suffix generation, $p(x_{k:P-1}\mid x_{0:k-1})$, through post-training a secondary mode. Not only can we demonstrate wireless signal generation of arbitrary length, but we can also exhibit custom label-driven distribution shifts (Figure~\ref{fig:ar-plot}). 

Other diffusion models similarly extend inference with an autoregressive loop. Google's DiffusionGemma denoises a canvas, writes a copy to a context window, and then generates the next block with added conditioning \citep{diffusiongemma}. ALF uses the same high-level windowed conditioning approach, but maintains continuity by explicitly using the last $k$ patches and their clean latent representations as anchors to generate the following $P-k$ patches based on their labels. This method preserves the RF signal structure, resulting in smoother transitions. 

\begin{figure}[H]
    \centering
    \includegraphics[width=1.0\linewidth]{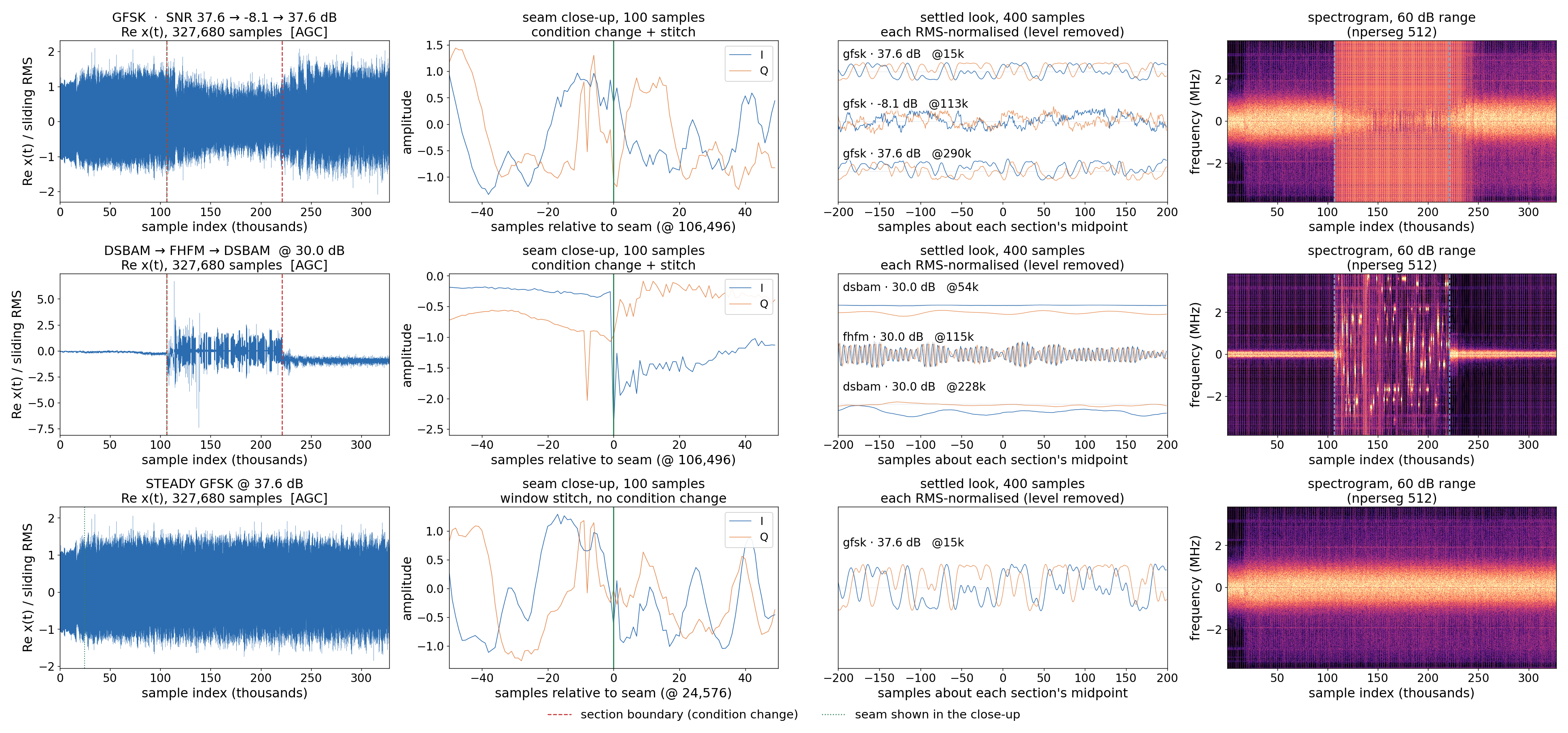}
    \vspace{-10pt}
    \caption{Autoregressive adaptation of AL generating 327,680-length samples. \textbf{Top row:} a pulsed CW signal transitions from $44.9$ to $-13.8$dB in SNR and back (noise shift). \textbf{Middle row:} Bluetooth signal transitions to repeated LFM and back (waveform shift). \textbf{Bottom row:} steady pulsed CW signal at $44.9$ dB provides a no-transition control.}
    \label{fig:ar-plot}
\end{figure}
\vspace{-10pt}
 We selected $P=16, k=8$ for post-trained ALF from a warm-start checkpoint seen in Figure~\ref{fig:ar-plot}. Training is split into two sections. \textbf{Mode~0} keeps the same anchored flow matching objective $\mathcal{L}_{\text{Stage 2}}$ described in section~\ref{subsec:anchoredflow}. \textbf{Mode~1} continuity training is applied to the first half of each length group, using independent noise times and prefix lengths. The sampled prefix is then assigned a zero velocity target, effectively masking its flow while still being used for mixing. The total objective combines flow matching over unmasked suffix tokens, a zero-velocity pin loss on prefix tokens, and an edge loss on the first token after the prefix. 
 \begin{align}
\mathcal{L}_{\mathrm{free}}
&=
\frac{
    \sum_{b,i:\,m_{b,i}=0}
    \left\lVert
        \hat{\mathbf{v}}_{b,i}
        - \mathbf{v}^{\ast c}_{b,i}
    \right\rVert_2^2
}{
    N_{\mathrm{free}}D
},
\qquad\mathcal{L}_{\mathrm{pin}}
=
\frac{
    \sum_{b,i:\,m_{b,i}=1}
    \left\lVert
        \hat{\mathbf{v}}_{b,i}
    \right\rVert_2^2
}{
    N_{\mathrm{given}}D
}, \\
\mathcal{L}_{\mathrm{edge}}
&=
\frac{1}{BD}
\sum_b
\left\lVert
    \hat{\mathbf{v}}_{b,k_b}
    - \mathbf{v}^{\ast c}_{b,k_b}
\right\rVert_2^2,
\qquad
\hat{\mathbf{v}}
=
v_\theta\!\left(\mathbf{x}^{c}_t,t^c,c,\mathbf{m}\right).
\end{align}
\begin{equation}
\mathcal{L}_{\mathrm{cont}}
=
\mathcal{L}_{\mathrm{free}}
+ \lambda_{\text{pin}}\mathcal{L}_{\mathrm{\text{pin}}}
+ \lambda_{\text{edge}}\mathcal{L}_{\mathrm{\text{edge}}}, \qquad \lambda_{\text{pin}}=0.1, \lambda_{\text{edge}}=0.3
\end{equation}
The total post-training loss combines this objective with the original generation loss with unit weight,
\begin{equation}
\mathcal{L}
=
\mathcal{L}_{\mathrm{\text{Stage 2}}}
+
\lambda_{\text{cont}}\mathcal{L}_{\mathrm{\text{cont}}},\qquad \lambda_{\text{cont}}=1
\end{equation}

\newpage
\section{Additional Figures and Tables}
\label{ap:add-fig-tex}


\subsection{Distributional Metrics}
\label{ap:distributional-metrics}
\begin{table*}[h]
\centering
\begin{threeparttable}

\caption{Comparison across PSD and spectral MSE metrics.}
\label{tab:dist-comparison}

\begin{tabular}{@{}llrr@{}}
\toprule
Length & Method & PSD MSE $\downarrow$ & Spec MSE $\downarrow$ \\
\midrule
\multirow[c]{4}{*}{2,048}
& WaveStitch & 158.08 & 379.50 \\
& RF Diffusion$^\dagger$ & 925.4 & 1094.4 \\
& Time Weaver & 44.57 & 122.39 \\
& \textbf{ALF (ours)} & \textbf{20.42} & \textbf{117.36} \\
\midrule
\multirow[c]{3}{*}{4,096}
& WaveStitch & 163.07 & 399.07 \\
& Time Weaver & 46.55 & 128.38 \\
& \textbf{ALF (ours)} & \textbf{17.17} & \textbf{117.35} \\
\midrule
\multirow[c]{3}{*}{8,192}
& WaveStitch & 164.21 & 403.27 \\
& Time Weaver & 42.98 & 120.78 \\
& \textbf{ALF (ours)} & \textbf{15.91} & \textbf{119.77} \\
\midrule
\multirow[c]{3}{*}{16,384}
& WaveStitch & 166.73 & 406.38 \\
& Time Weaver & 41.50 & \textbf{114.90} \\
& \textbf{ALF (ours)} & \textbf{15.79} & 119.00 \\
\midrule
\multirow[c]{3}{*}{32,768}
& WaveStitch & 169.61 & 421.28 \\
& Time Weaver & 44.00 & \textbf{112.91} \\
& \textbf{ALF (ours)} & \textbf{14.96} & 122.18 \\
\midrule
\multirow[c]{3}{*}{\textbf{Avg.}}
& WaveStitch & 164.3400 & 401.9000 \\
& Time Weaver & 43.9200 & 119.8720 \\
& \textbf{ALF (ours)} & \textbf{16.8500} & \textbf{119.1320} \\
\midrule
\multirow[c]{3}{*}{\textbf{ALF Advantage}}
& \textbf{vs.\ WaveStitch}
& \textbf{89.75\%} & \textbf{70.36\%} \\
& \textbf{vs.\ RF Diffusion}$^\dagger$
& \textbf{97.79\%} & \textbf{89.28\%} \\
& \textbf{vs.\ Time Weaver}
& \textbf{61.63\%} & \textbf{0.62\%} \\
\bottomrule
\end{tabular}

\begin{tablenotes}[flushleft]
\footnotesize
\item[] $\dagger$ RF Diffusion could not be trained for records of length greater than 2,048 due to memory limitations; longer generations were not considered. All other models are evaluated on 100k recordings.
\end{tablenotes}

\end{threeparttable}
\end{table*}
Table~\ref{tab:dist-comparison} compares ALF against its competitors using more distributional metrics. Although downstream tasks are a better judgment of quality, we include these to provide a full view of ALF's performance (\citet{baur2025evaluation}). The table shows we are significantly superior in PSD MSE, and overall better in Spectrogram MSE. However, this further validates our claim that distributional alignment is not fine enough of an evaluation metric for highly constrained representations of wireless signals.

\newpage
\subsection{Other Generation Quality Metrics}
\label{ap:add-quality-metrics}
\begin{table*}[h]
\centering
\caption{Effect on quality metrics due to rescaling on energy-linked objectives}
\label{tab:norm-vs-rescaled}

\resizebox{\textwidth}{!}{%
\setlength{\tabcolsep}{3pt}%
\begin{tabular}{@{}ll|cc|cc|cc|cc@{}}
\toprule
& & \multicolumn{2}{c|}{Delay Spread ($k=9$)}
& \multicolumn{2}{c|}{Indoor/Outdoor ($k=2$)}
& \multicolumn{2}{c|}{Urban/Rural ($k=2$)}
& \multicolumn{2}{c}{LOS/NLOS ($k=2$)} \\
\cmidrule(lr){3-4}
\cmidrule(lr){5-6}
\cmidrule(lr){7-8}
\cmidrule(lr){9-10}
Power mode & Method
& TSTR & TRTS
& TSTR & TRTS
& TSTR & TRTS
& TSTR & TRTS \\
\midrule
\multirow[c]{3}{*}{Normalized}
& WaveStitch
& 0.11 & 0.11 & 0.50 & 0.49 & 0.50 & 0.51 & 0.49 & 0.53 \\
& TimeWeaver
& 0.17 & 0.15 & 0.54 & \textbf{0.64} & \textbf{0.57} & 0.54 & 0.56 & 0.61 \\
& ALF
& \textbf{0.19} & \textbf{0.18} & \textbf{0.56} & 0.64 & 0.54 & \textbf{0.56} & \textbf{0.57} & \textbf{0.71} \\
\midrule
\multirow[c]{3}{*}{Rescaled}
& WaveStitch
& 0.11 & 0.12 & 0.50 & 0.74 & 0.50 & 0.55 & 0.56 & 0.75 \\
& TimeWeaver
& 0.20 & 0.16 & 0.74 & 0.73 & \textbf{0.57} & 0.56 & 0.84 & 0.76 \\
& ALF
& \textbf{0.22} & \textbf{0.20} & \textbf{0.75} & \textbf{0.74} & 0.56 & \textbf{0.57} & \textbf{0.85} & \textbf{0.79} \\
\midrule
\multirow[c]{3}{*}{\shortstack{Rescaling\\Gain ($\Delta$)}}
& WaveStitch
& 0.00 & 0.01 & $-0.01$ & 0.24 & 0.00 & 0.04 & 0.06 & 0.22 \\
& TimeWeaver
& 0.04 & 0.01 & 0.20 & 0.09 & $-0.00$ & 0.02 & 0.28 & 0.15 \\
& ALF
& 0.04 & 0.02 & 0.19 & 0.10 & 0.02 & 0.00 & 0.28 & 0.08 \\
\midrule
\multirow[c]{3}{*}{\shortstack{Above Chance\\($\Delta_{\mathrm{chance}}$)}}
& WaveStitch
& 0.00 & 0.01 & 0.00 & \textbf{0.24} & 0.00 & 0.05 & 0.06 & 0.25 \\
& TimeWeaver
& 0.09 & 0.05 & 0.24 & 0.23 & \textbf{0.07} & 0.06 & 0.34 & 0.26 \\
& ALF
& \textbf{0.11} & \textbf{0.09} & \textbf{0.25} & \textbf{0.24}
& 0.06 & \textbf{0.07} & \textbf{0.35} & \textbf{0.29} \\
\bottomrule
\end{tabular}%
}

\vspace{2pt}
\parbox{\textwidth}{\footnotesize\emph{Note:} $k$ = number of classes; pooled = mean over lengths 2,048--32,768. $\Delta$ = Rescaled minus Normalized. $\Delta_{\mathrm{chance}}$ = Rescaled minus chance accuracy, where chance is $1/k$. LOS = Line-of-Sight; NLOS = No-Line-of-Sight.}
\end{table*}
Table~\ref{tab:norm-vs-rescaled} shows the rest of the classification objectives unable to fit into the main paper (Table~\ref{tab:tstr-headline}, $\S$~\ref{sec:results}). These tasks show lower utility performance than Time Weaver; however, this is due to their energy-based nature. Both models output normalized data that lacks attenuation, which makes scene classifications (Indoor/Outdoor, LOS/NLOS) far easier to detect than shape alone. After rescaling each recording to match its ground-truth reference from each respective class, we see classification performance equalize across WaveStitch, Time Weaver, and ALF for Indoor/Outdoor and LOS/NLOS tasks. The other two objectives have accuracies barely over their chance probabilities. These noise readings indicate these as more difficult representations.

\end{document}